\documentclass[]{article}
\usepackage{amscd,mathrsfs,amscd,amsmath,latexsym,amsfonts,extarrows}
\usepackage{amssymb,amsthm,bbm,graphicx,color,CJK}
\usepackage{pslatex,cite,multirow,subeqnarray,cases,setspace}

\renewcommand{\theequation}{\arabic{section}.\arabic{equation}}
\newtheorem{proposition}{Proposition}[section]

\newtheorem{lemma}{Lemma}[section]
\newtheorem{theorem}{Theorem}[section]
\newtheorem{corollary}{Corollary}[section]

\theoremstyle{remark}
\newtheorem{remark}{Remark}[section]

\begin{document}
\begin{CJK*}{GBK}{song}
\title{Dispersive Limit of the Euler-Poisson System\\ in Higher Dimensions}
\author
         {{Xueke Pu}\\
         {\small Department of Mathematics, Chongqing University, Chongqing 400044, P.R.China;}\\
         {\small Mathematical Sciences Research Institute in Chongqing}\\
         {\small xuekepu@cqu.edu.cn}\\
         \date{}
         }
\maketitle

\begin{abstract}
In this paper, we consider the dispersive limit of the Euler-Poisson system for ion-acoustic waves. We establish that under the Gardner-Morikawa type transformations, the solutions of the Euler-Poisson system converge globally to the Kadomtsev-Petviashvili II equation in $\Bbb R^2$ and the Zakharov-Kuznetsov equation in $\Bbb R^3$ for well-prepared initial data, under different scalings. This justifies rigorously the KP-II limit and the ZKE limit of the Euler-Poisson equation.
\end{abstract}
\begin{center}
 \begin{minipage}{120mm}
\end{minipage}
\end{center}
\begin{center}
\begin{minipage}{120mm}
\end{minipage}
\end{center}


\section{Introduction}
\setcounter{section}{1}\setcounter{equation}{0}
Consider the Euler-Poisson (EP) system in 2D,
\begin{subequations}\label{ep2}
\begin{numcases}{}
\partial_tn+\nabla\cdot(n\textbf{u})=0,\label{ep2-n}\\
\partial_t\textbf{u}+\textbf{u}\cdot\nabla\textbf{u}+T_i\frac{\nabla n}{n}=-\nabla\phi,\label{ep2-u}\\
\Delta\phi=e^{\phi}-n,\label{ep2-p}
\end{numcases}
\end{subequations}
where $n(t,x),\textbf{u}(t,x)=(u_1(t,x),u_2(t,x))$ and $\phi(t,x)$ are respectively the density, velocity of the ions and the electric potential at time $t\geq 0$ and position $x=(x_1,x_2)\in\Bbb R^2$. Here $T_i\geq0$ denotes the ion temperature, and all the other physical parameters are set to be 1.  

The Euler-Poisson system \eqref{ep2} is a fundamental two-fluid model describing the dynamics of a plasma, in which compressible ion and electron fluids interact with their self-consistent electrostatic force. Here, the hot isothermal electrons are described by the Boltzmann distribution. Such an Euler-Poisson system was widely investigated in the past years. The interested readers may refer to \cite{CG00, ELT01, GGP11, Guo98, GJ10, GP11, GP12, LMZ10, LT02,LT03} and the references therein. Many important nonlinear dispersive PDEs, such as the Kadomtsev-Petviashvili II (KP-II) equations \cite{KP70} and the Zakharov-Kuznetsov equations (ZKE) \cite{ZK74} can be formally derived from the Euler-Poisson system, and they are widely used as approximate models of the Euler-Poisson system in some limit sense, in many physical contexts. These equations are higher dimensional generalizations of the Korteweg-de Vries (KdV) equation \cite{KdV} and were extensively studied in the past decades. In particular, the KP-II equation describes the propagation of long nonlinear waves along the $x_1$-axis on the surface of a media when the variation along the $x_2$-axis proceeds slowly. However, there is up to now no rigorous mathematical justifications of such dispersive limits. The purpose of this paper is to justify rigorously these formal limits at least for well-prepared initial data. We also remark that in recent years, the KdV and the KP-I limits are justified from different interesting models, such as the water wave problem, the Schr\"odinger equation and the Euler-Poisson equation \cite{CR10,SW00,GP12}. 

%

This paper is organized as follows. In Section 2, we formally derive the 2D KP-II and state the main Theorem \ref{th}. For clarity, the formal derivation of the ZKE from a slightly different 3D Euler-Poisson system \eqref{equ1} is postponed to Appendix B. To give a unified treatment for the KP-II limit and the ZKE limit, we write the remainder equations into a unified system \eqref{re} by introducing some new differential notations. Section 3 and 4 are dedicated to the proof of Theorem \ref{th}. In Section 3, we prove Theorem \ref{th} for the case of $T_i>0$ for both the KP-II limit in 2D and the ZKE limit in 3D. In Section 4, we prove Theorem \ref{th} for the case of $T_i=0$ for the ZKE limit in 3D.

Throughout this paper, we use $[A,B]=AB-BA$ to denote the commutator of $A$ and $B$ and $\|\cdot\|_X$ to denote an $X$-norm. When $X=L^2$, the subscript of $\|\cdot\|_{L^2}$ is usually omitted. We also use $\langle f,g\rangle_0=\int fgdx$ to denote the inner product of two $L^2$ functions.




\section{Formal derivation and the main results}
\setcounter{section}{2}\setcounter{equation}{0}
\subsection{Formal derivation of the KPE}
In this subsection, we derive the KP equation from the 2D Euler-Poisson equation \eqref{ep2}. Consider the following Gardner-Morikawa type of transformation in \eqref{ep2}
\begin{equation}\label{e99}
\begin{split}
\varepsilon^{1/2}(x_1-Vt)\to x_1,\ \ \varepsilon x_2\to x_2,\ \ \varepsilon^{3/2}t\to t,
\end{split}
\end{equation}
where $\varepsilon$ stands for the amplitude of the initial disturbance and is assumed to be small compared with unity and $V$ is the wave speed to be determined. Then we obtain the parameterized system
\begin{subequations}\label{k1}
\begin{numcases}{}
\varepsilon\partial_tn-V\partial_{x_1}n+\partial_{x_1}(nu_1) +\varepsilon^{1/2}\partial_{x_2}(nu_2)=0,\label{k1-n}\\
\varepsilon\partial_tu_1-V\partial_{x_1}u_1+u_1\partial_{x_1}u_1 +\varepsilon^{1/2}u_2\partial_{x_2}u_1+T_i\frac{\partial_{x_1}n}{n}= -\partial_{x_1}\phi,\label{k1-1}\\
\varepsilon\partial_tu_2-V\partial_{x_1}u_2+u_1\partial_{x_1}u_2 +\varepsilon^{1/2}u_2\partial_{x_2}u_2+T_i\frac{\varepsilon^{1/2}\partial_{x_2}n}{n}= -\varepsilon^{1/2}\partial_{x_2}\phi,\label{k1-2}\\
\varepsilon\partial_{x_1}^2\phi+\varepsilon^2\partial_{x_2}^2\phi=e^{\phi}-n.\label{k1-p}
\end{numcases}
\end{subequations}
Consider the following formal expansion
\begin{subequations}\label{formal-kp}
\begin{numcases}{}
n=1+\varepsilon n^{(1)}+\varepsilon^2n^{(2)}+\varepsilon^3n^{(3)}+\cdots,\label{formal-kp-n}\\
u_1=\varepsilon u_1^{(1)}+\varepsilon^2u_1^{(2)}+\varepsilon^3u_1^{(3)}+\cdots,\label{formal-kp-1}\\
u_2=\varepsilon^{3/2}u_2^{(1)} +\varepsilon^{5/2}u_2^{(2)}+\varepsilon^{7/2}u_2^{(3)}+\cdots,\label{formal-kp-2}\\
\phi=\varepsilon\phi^{(1)}+\varepsilon^2\phi^{(2)}+\varepsilon^3\phi^{(3)} +\cdots.\label{formal-kp-p}
\end{numcases}
\end{subequations}
Plugging this formal expansion into the system \eqref{k1}, we get a power series of $\varepsilon$, whose coefficients depend on $(n^{(k)},\textbf{u}^{(k)},\phi^{(k)})$ for $k\geq1$, where $\textbf{u}^{(k)}=(u_1^{(k)},u_2^{(k)})^T$.

\subsubsection{Derivation of the KPE for $n^{(1)}$}
From the power series of $\varepsilon$ we thus obtained, we get at the order of $\varepsilon$:\\
\emph{Coefficients of $\varepsilon^1$:}
\begin{subequations}\label{k-order1}
\begin{numcases}{}
-V\partial_{x_1}n^{(1)}+\partial_{x_1}u_1^{(1)}=0,\label{k-order1-n}\\
-V\partial_{x_1}u_1^{(1)}+T_i\partial_{x_1}n^{(1)}=-\partial_{x_1}\phi^{(1)},\label{k-order1-1}\\
0=\phi^{(1)}-n^{(1)}.\label{k-order1-p}
\end{numcases}
\end{subequations}
To get a nontrivial solution of $(n^{(1)},u_1^{(1)},\phi^{(1)})$, we let the determinant of the coefficient matrix of \eqref{k-order1} to vanish to obtain
\begin{equation}\label{k2}
\begin{split}
V^2=T_i+1.
\end{split}
\end{equation}
For definiteness, we set $V=\sqrt{T_i+1}$ in the following.

At higher orders, we obtain\\
\emph{Coefficients of $\varepsilon^{3/2}$:}
\begin{equation}\label{k-order32}
\begin{split}
-V\partial_{x_1}u_2^{(1)}+T_i\partial_{x_2}n^{(1)}=-\partial_{x_2}\phi^{(1)}.
\end{split}
\end{equation}
\\
\emph{Coefficients of $\varepsilon^2$:}
\begin{subequations}\label{k-order2}
\begin{numcases}{}
\partial_tn^{(1)}-V\partial_{x_1}n^{(2)}+\partial_{x_1}u_1^{(2)} +\partial_{x_1}(n^{(1)}u_1^{(1)})+\partial_{x_2}u_2^{(1)}=0,\label{k-order2-n}\\
\partial_tu_1^{(1)}-V\partial_{x_1}u_1^{(2)}+u_1^{(1)}\partial_{x_1}u_1^{(1)} +T_i\{\partial_{x_1}n^{(2)}-n^{(1)}\partial_{x_1}n^{(1)}\}=-\partial_{x_1}\phi^{(1)},\label{k-order2-1}\\
\partial_{x_1}^2\phi^{(1)}=\phi^{(2)}+\frac12(\phi^{(1)})^2-n^{(2)}.\label{k-order2-p}
\end{numcases}
\end{subequations}

From \eqref{k-order1}, we may assume that
\begin{equation}\label{k3}
\begin{split}
u^{(1)}_1=Vn^{(1)},\ \ \phi^{(1)}=n^{(1)},
\end{split}
\end{equation}
which also make \eqref{k-order1} valid, thanks to \eqref{k2}. Then from \eqref{k-order32}, we have
\begin{equation}\label{k4}
\begin{split}
\partial_{x_1}u_2^{(1)}=V\partial_{x_2}n^{(1)},
\end{split}
\end{equation}
thanks to \eqref{k3}. Therefore, to solve $n^{(1)}, \textbf{u}^{(1)}$ and $\phi^{(1)}$, we need only to solve $n^{(1)}$.

To find out the equation satisfied by $n^{(1)}$, we take $\partial_{x_1}$ of \eqref{k-order2-p}, multiply \eqref{k-order2-n} by $V$, and then add them to \eqref{k-order2-1}. We obtain
\begin{equation}
\begin{split}
\partial_{t}n^{(1)}+Vn^{(1)}\partial_{x_1}n^{(1)}+\frac{1}{2V}\partial_{x_1}^3n^{(1)} +\frac12\partial_{x_2}u_2^{(1)}=0.
\end{split}
\end{equation}
Differentiating this equation with respect to $x_1$, and using \eqref{k4}, we obtain
\begin{equation}\label{KPE}
\begin{split}
\partial_{x_1}\{\partial_{t}n^{(1)}+Vn^{(1)}\partial_{x_1}n^{(1)} +\frac{1}{2V}\partial_{x_1}^3n^{(1)}\} +\frac{V}{2}\partial_{x_2}^2n^{(1)}=0.
\end{split}
\end{equation}
This is the Kadomtsev-Petviashvili II equation satisfied by the first order profile $n^{(1)}$.

\begin{proposition}\label{prop-kp1}
The Cauchy problem for the KP-II equation is well-posed in $H^s({\Bbb R^2})$ for $s\geq 0$.
\end{proposition}

This theorem is proved in the seminal paper of Bourgain \cite{Bo93}. Actually, it is known to be well-posed in spaces of much lower regularity \cite{HHK09,MST11,ILM07}. However, Proposition \ref{prop-kp1} is enough for our purpose. 

\begin{remark}
The system of \eqref{k3}, \eqref{k4} and \eqref{KPE} is a closed system. Once $n^{(1)}$ is solved from \eqref{KPE}, we have all the other first order profiles $\textbf{u}^{(1)}$ and $\phi^{(1)}$ from \eqref{k3} and \eqref{k4}.
\end{remark}
From \eqref{k-order2-n} and \eqref{k-order2-p}, we may assume
\begin{subequations}\label{k5}
\begin{numcases}{}
u_1^{(2)}=Vn^{(2)}+\underline{u_1}_{KP}^{(1)},\\
\phi^{(2)}=n^{(2)}+\underline{\phi}_{KP}^{(1)},
\end{numcases}
\end{subequations}
where $\underline{u_1}_{KP}^{(1)}$ and $\underline{\phi}_{KP}^{(1)}$ depend only on $n^{(1)}$, which is smooth in some time interval $[0,\tau_*)$ thanks to Proposition \ref{prop-kp1}. 

At the order $\varepsilon^{5/2}$, we obtain\\
\emph{Coefficients of $\varepsilon^{5/2}$:}
\begin{equation}\label{k-order52}
\begin{split}
-V\partial_{x_1}u_2^{(2)}+u_1^{(1)}\partial_{x_1}u_2^{(1)} +T_i\{\partial_{x_2}n^{(2)}-n^{(1)}\partial_{x_2}n^{(1)}\}=-\partial_{x_2}\phi^{(2)}.
\end{split}
\end{equation}
By using \eqref{k5} and rearranging, we have
\begin{equation}\label{k6}
\begin{split}
V\partial_{x_1}u_2^{(2)}=(T_i+1)\partial_{x_2}n^{(2)}+u_1^{(1)}\partial_{x_1}u_2^{(1)} -T_in^{(1)}\partial_{x_2}n^{(1)}+\partial_{x_2}\underline{\phi}_{KP}^{(1)}.
\end{split}
\end{equation}
\subsubsection{Derivation of the Linearized KPE for $n^{(k)}$}
From \eqref{k5} and \eqref{k6}, we see that to determine $(n^{(2)},\textbf{u}^{(2)},\phi^{(2)})$, we need only to determine $n^{(2)}$.

At the order of $\varepsilon^3$, we obtain\\
\emph{Coefficients of $\varepsilon^3$:}
\begin{subequations}\label{k-order3}
\begin{numcases}{}
\partial_tn^{(2)}-V\partial_{x_1}n^{(3)}+\partial_{x_1}u_1^{(3)}\nonumber\\
\ \ \ \ \ \ \ \ \ \ \ +\partial_{x_1}(n^{(1)}u_1^{(2)}+n^{(2)}u_1^{(1)}) +\partial_{x_2}(u_2^{(2)}+n^{(1)}u_2^{(1)})=0,\label{k-order3-n}\\
\partial_tu_1^{(2)}-V\partial_{x_1}u_1^{(3)}+\partial_{x_1}(u_1^{(1)}u_1^{(2)}) +T_i\partial_{x_1}n^{(3)}\nonumber\\
\ \ \ \ \ \ \ \ \ \ \ +T_i\{\partial_{x_1}n^{(1)}(\frac12(n^{(1)})^2-n^{(2)}) -\partial_{x_1}n^{(2)}n^{(1)}\}= -\partial_{x_1}\phi^{(3)},\label{k-order3-1}\\
\partial_{x_1}^2\phi^{(2)}+\partial_{x_2}^2\phi^{(1)} =\phi^{(3)}+\phi^{(1)}\phi^{(2)}+\frac1{3!}(\phi^{(1)})^3-n^{(3)}.\label{k-order3-p}
\end{numcases}
\end{subequations}

Taking $\partial_{x_1}$ of \eqref{k-order3-p}, multiplying \eqref{k-order3-n} with $V$, and then adding them to \eqref{k-order3-1}, we obtain the linearized inhomogeneous KP equation
\begin{equation}\label{lin-KP2}
\begin{split}
\partial_{x_1}\{\partial_{t}n^{(2)}+2V\partial_{x_1}(n^{(1)}n^{(2)}) +\frac{1}{2V}\partial_{x_1}^3n^{(2)}\} +\frac{V}{2}\partial_{x_2}^2n^{(2)}=\underline{G}_{KP}^{(1)},
\end{split}
\end{equation}
where we have used \eqref{k5} and \eqref{k6}. Here $\underline{G}_{KP}^{(1)}$ depends only on $n^{(1)}$ and comes form the inhomogeneous dependence of $u_1^{(2)}$ and $\phi^{(2)}$ on $n^{(2)}$ in \eqref{k5}.

At the order $\varepsilon^{7/2}$, we obtain
\\
\emph{Coefficients of $\varepsilon^{7/2}$:}
\begin{equation}\label{k-order72}
\begin{split}
\partial_tu_2^{(2)}-&V\partial_{x_1}u_2^{(3)}+u_1^{(1)}\partial_{x_1}u_2^{(2)} +u_1^{(2)}\partial_{x_1}u_2^{(1)} +u_2^{(1)}\partial_{x_2}u_2^{(1)}  +T_i\partial_{x_1}n^{(3)}\\
&+T_i\{\partial_{x_2}n^{(1)}(\frac12(n^{(1)})^2-n^{(2)}) -\partial_{x_2}n^{(2)}n^{(1)}\}=-\partial_{x_2}\phi^{(3)}.
\end{split}
\end{equation}

Inductively, we can derive all the profiles $(n^{(k)},\textbf{u}^{(k)},\phi^{(k)})$ for $k\geq3$. Proceeding as above, we obtain the following linearized inhomogeneous KP equation for $n^{(k)}$ for $k\geq 3$:
\begin{equation}\label{lin-KPk}
\begin{split}
\partial_{x_1}\{\partial_{t}n^{(k)}+2V\partial_{x_1}(n^{(1)}n^{(k)}) +\frac{1}{2V}\partial_{x_1}^3n^{(k)}\} +\frac{V}{2}\partial_{x_2}^2n^{(k)}=\underline{G}_{KP}^{(k-1)},
\end{split}
\end{equation}
where the inhomogeneous term $\underline{G}_{KP}^{(k-1)}$ depends only on $(n^{(j)},\textbf{u}^{(j)},\phi^{(j)})$ for $1\leq j\leq k-1$.

Since \eqref{lin-KPk} is linear in $n^{(k)}$, we easily obtain
\begin{proposition}\label{prop-kp4}
The Cauchy problem of the linearized inhomogeneous KPE \eqref{lin-KPk} ($k\geq 2$) is well-posed in $H^s(\Bbb R^3)$ for $s\geq 0$.
\end{proposition}

\subsection{Remainder system}
To make the above procedure rigorous, we need to cut off and consider the remainder terms. For this, we consider the following expansion
\begin{subequations}\label{exp-kp}
\begin{numcases}{}
n=1+\varepsilon n^{(1)}+\varepsilon^2n^{(2)}+\varepsilon^3n^{(3)} +\varepsilon^2n^{\varepsilon}_R,\label{exp-kp-n}\\
u_1=\varepsilon u_1^{(1)}+\varepsilon^2u_1^{(2)}+\varepsilon^3u_1^{(3)} +\varepsilon^2{u_1}^{\varepsilon}_R,\label{exp-kp-1}\\
u_2=\varepsilon^{3/2}u_2^{(1)} +\varepsilon^{5/2}u_2^{(2)}+\varepsilon^{7/2}u_2^{(3)} +\varepsilon^2{u_2}^{\varepsilon}_R,\label{exp-kp-2}\\
\phi=\varepsilon\phi^{(1)}+\varepsilon^2\phi^{(2)}+\varepsilon^3\phi^{(3)} +\varepsilon^2\phi^{\varepsilon}_R,\label{exp-kp-p}
\end{numcases}
\end{subequations}
where $(n^{(i)},\textbf{u}^{(i)}, \phi^{(i)})$ for $1\leq i\leq 3$ are the first few profiles constructed above and $(n^{\varepsilon}_R,\textbf{u}^{\varepsilon}_R, \phi^{\varepsilon}_R)$ are the remainder terms that may depend on $\varepsilon$ and $(n^{(i)},\textbf{u}^{(i)}, \phi^{(i)})$ for $1\leq i\leq 3$. More precisely, $n^{(1)}$ satisfies \eqref{KPE}, $u^{(1)}_1,\phi^{(1)}$ satisfy \eqref{k3} and $u_2^{(1)}$ satisfies \eqref{k4}; $n^{(2)}$ satisfies \eqref{lin-KP2}, $u^{(2)}_1,\phi^{(2)}$ satisfy \eqref{k5} and $u_2^{(2)}$ satisfies \eqref{k6}; and similarly for the third order profile $(n^{(3)},\textbf{u}^{(3)},\phi^{3})$.

Inserting \eqref{exp-kp} into \eqref{k1}, and then subtracting the systems of the coefficient up to order $\varepsilon^{7/2}$, we obtain the remainder system of $(n^{\varepsilon}_R,\textbf{u}^{\varepsilon}_R, \phi^{\varepsilon}_R)$. For notational convenience, we denote $\textbf{u}=({u_1},{u_2})^T$, $\widetilde{\textbf{u}}=(\widetilde u_1,\widetilde u_2)^T$, $\textbf{u}^{\varepsilon}_R=({u_1}^{\varepsilon}_R, {u_2}^{\varepsilon}_R)^T$ and
\begin{equation}\label{e2}
\begin{split}
\widetilde n=&n^{(1)}+\varepsilon n^{(2)}+\varepsilon^2n^{(3)},\ \ \ \ \ \widetilde{\phi}=\phi^{(1)}+\varepsilon\phi^{(2)}+\varepsilon^2\phi^{(3)}\\
\widetilde u_1=&u_1^{(1)}+\varepsilon u_1^{(2)}+\varepsilon^2u_1^{(3)},\ \ \ \
\widetilde u_2=\varepsilon^{1/2}u_2^{(1)}+\varepsilon^{3/2} u_2^{(2)}+\varepsilon^{5/2}u_2^{(3)}.
\end{split}
\end{equation}
\begin{proposition}\label{prop-kp3}
Let $(n,\textbf{u},\phi)$ in \eqref{exp-kp} be a solution of the Euler-Poisson equation \eqref{ep2}, then the remainders $n^{\varepsilon}_R,\textbf{u}^{\varepsilon}_R$ and $\phi^{\varepsilon}_R$ in \eqref{exp-kp} satisfy
\begin{subequations}\label{rem-kp}
\begin{numcases}{}
\partial_tn^{\varepsilon}_R -\frac{V-u_1}{\varepsilon}\partial_{x_1}n^{\varepsilon}_R +\frac{\varepsilon^{1/2}u_2}{\varepsilon} \partial_{x_2}n^{\varepsilon}_R +\frac{n}{\varepsilon} \partial_{x_1}{u_1}^{\varepsilon}_R +\frac{\varepsilon^{1/2}n}{\varepsilon} \partial_{x_2}{u_2}^{\varepsilon}_R\nonumber\\
\ \ \ \ \ \ \ \ \ \ +n^{\varepsilon}_R\partial_{x_1}\widetilde{u}_1 +\varepsilon^{1/2}n^{\varepsilon}_R\partial_{x_2}\widetilde{u}_2 +{u_1}^{\varepsilon}_R\partial_{x_1}\widetilde{n} +\varepsilon^{1/2}{u_2}^{\varepsilon}_R\partial_{x_2}\widetilde{n} +\varepsilon R_n=0,\label{rem-kp-n}\\
\partial_t{u_1}^{\varepsilon}_R-\frac{V-{u_1}}{\varepsilon} \partial_{x_1}{u_1}^{\varepsilon}_R +\frac{\varepsilon^{1/2}{u_2}}{\varepsilon} \partial_{x_2}{u_1}^{\varepsilon}_R +{u_1}^{\varepsilon}_R\partial_{x_1}\widetilde{u}_1 +\varepsilon^{1/2}{u_2}^{\varepsilon}_R\partial_{x_2} \widetilde{u}_1\nonumber\\
\ \ \ \ \ \ \ \ \ \ +T_i\frac{\partial_{x_1}n^{\varepsilon}_R}{\varepsilon n} -T_i\frac{p_1}{n}n^{\varepsilon}_R -T_i\frac{\varepsilon R_{T1}}{n}+\varepsilon R_1=-\frac{1}{\varepsilon} \partial_{x_1}\phi^{\varepsilon}_R,\label{rem-kp-1}\\
\partial_t{u_2}^{\varepsilon}_R-\frac{V-{u_1}}{\varepsilon} \partial_{x_1}{u_2}^{\varepsilon}_R +\frac{\varepsilon^{1/2}{u_2}}{\varepsilon} \partial_{x_2}{u_2}^{\varepsilon}_R +{u_1}^{\varepsilon}_R\partial_{x_1}\widetilde{u}_2 +\varepsilon^{1/2}{u_2}^{\varepsilon}_R\partial_{x_2} \widetilde{u}_2\nonumber\\
\ \ \ \ \ \ \ \ \ \ +T_i\frac{\varepsilon^{1/2}\partial_{x_2}n^{\varepsilon}_R}  {\varepsilon n} -T_i\frac{\varepsilon^{1/2}p_2}{n}n^{\varepsilon}_R -T_i\frac{\varepsilon^{1/2}\varepsilon R_{T2}}{n}+\varepsilon R_2=-\frac{\varepsilon^{1/2}}{\varepsilon} \partial_{x_2}\phi^{\varepsilon}_R,\label{rem-kp-2}\\
\varepsilon\partial_{x_1}^2\phi^{\varepsilon}_R  +\varepsilon^2\partial_{x_2}^2\phi^{\varepsilon}_R =\phi^{\varepsilon}_R-n^{\varepsilon}_R +\varepsilon\phi^{(1)}\phi^{\varepsilon}_R +\varepsilon^{3/2}R_{\phi},\label{rem-kp-p}
\end{numcases}
\end{subequations}
where $(\widetilde{n}, \widetilde{u}_1, \widetilde{u}_2)$ are given in \eqref{e2} and $(n^{(i)},{u_1}^{(i)},{u_2}^{(i)},\phi^{(i)})$ for $1\leq i\leq 3$ satisfy the systems \eqref{k-order1}, \eqref{k-order32}, \eqref{k-order2}, \eqref{k-order52}, \eqref{k-order3} and \eqref{k-order72}. Here, $R_n$, $\textbf{R}_{\textbf{u}}=(R_1,R_2)$ and the coefficients $p_i$ and $R_{Ti}$ for $i=1,2$ depend only on $(\widetilde{n}, \widetilde{u}_1, \widetilde{u}_2, \widetilde{\phi})$. In \eqref{rem-kp-n}, $\textbf{e}_1=(1,0,0)'$ is a constant vector.
\end{proposition}

We also give some basic estimates for the remainder term $R_{\phi}$ in the following
\begin{lemma}\label{lem-A}
Let $k\geq 0$ be an integer, then there exists a constant $1\leq C_1=C_1(\sqrt{\varepsilon}\|\phi^{\varepsilon}_R\|_{H^{\delta}})$, such that
\begin{align}
\|R_{\phi}\|_{H^{k}}\leq &C_1(\sqrt{\varepsilon}\|\phi^{\varepsilon}_R\|_{H^{\delta}}) (1+\|\phi^{\varepsilon}_R\|_{H^{k}}), \ \ \text{ and}\label{e52-1}\\
\|\partial_tR_{\phi}\|_{H^{k}}\leq &C_1(\sqrt{\varepsilon}\|\phi^{\varepsilon}_R\|_{H^{\delta}}) (1+\|\partial_t\phi^{\varepsilon}_R\|_{H^{k}}),\label{e52-2}
\end{align}
where $\delta=\max\{2,k-1\}$. Furthermore, the constant $C_1(\cdot)$ can be chosen to be nondecreasing.
\end{lemma}

The proof of Proposition \ref{prop-kp3} and Lemma \ref{lem-A} is given in the Appendix A. To prove the KPE limit rigorously, we need only to derive some uniform estimates for the remainder $(n^{\varepsilon}_R,\textbf{u}^{\varepsilon}_R, \phi^{\varepsilon}_R)$.

From a different scaling, we can derive the Zakharov-Kuznetsov equation from the Euler-Poisson system with static magnetic field. The derivations of such a ZKE and its remainder equation are detailed in Appendix \ref{app-b}. See Proposition \ref{prop-rem-zk}.

\subsection{Unified remainder system}
We introduce a unified form the remainder system \eqref{rem-kp} for the KP-II limit in 2D and \eqref{rem} for the ZKE limit in 3D. This will substantially simplify the presentation of this paper. For this purpose, we define
\begin{equation}\label{e1}
\begin{split}
\overline{\nabla}=
\begin{cases}
(\partial_{x_1},\sqrt{\varepsilon}\partial_{x_2})^T\ \ \ &d=2;\\
(\partial_{x_1},\partial_{x_2},\partial_{x_3})^T\ \ \ &d=3,
\end{cases}
\end{split}
\end{equation}
where $\cdot^T$ means transpose and $\overline{\Delta}=\overline{\nabla}\cdot\overline{\nabla}$.
We also denote
\begin{equation}
\begin{split}
\textbf{p}=
\begin{cases}
(p_1,\sqrt{\varepsilon}p_2)^T\ \ \ &d=2,\\
(p_1,p_2,p_3)^T\ \ \ &d=3,
\end{cases}
\end{split}
\end{equation}
and
\begin{equation}
\begin{split}
\textbf{R}_T=
\begin{cases}
(R_{T1},\sqrt{\varepsilon}R_{T2})^T\ \ \ &d=2,\\
(R_{T1},R_{T2},R_{T3})^T\ \ \ &d=3,
\end{cases}
\end{split}
\end{equation}
where $p_i$ and $R_{Ti}$ are those in \eqref{rem-kp}.
Let $(\widetilde n,\widetilde{\textbf{u}},\widetilde{\phi})$ denote either the expressions in \eqref{e2} for the 2D case or the expressions in \eqref{equ95} for the 3D case. Under these notations, we have
\begin{equation}\label{e3}
\begin{split}
(n,\textbf{u},\phi)=(1+\varepsilon\widetilde n,\varepsilon\widetilde{\textbf{u}},\varepsilon\widetilde{\phi})+(\varepsilon^2n^{\varepsilon}_R, \varepsilon^2\textbf{u}^{\varepsilon}_R,\varepsilon^2\phi^{\varepsilon}_R)
\end{split}
\end{equation}
for both the 2D case and the 3D case.

\begin{proposition}\label{prop-kp5}
Under these notations of \eqref{e1}-\eqref{e3}, the remainder equations for \eqref{rem-kp} and \eqref{rem} can be unified into
\begin{subequations}\label{re}
\begin{numcases}{}
\partial_tn^{\varepsilon}_R -\frac{V\textbf{e}_1-\textbf{u}}{\varepsilon}\cdot\overline{\nabla} n^{\varepsilon}_R +\frac{n}{\varepsilon}\overline{\nabla}\cdot\textbf{u}^{\varepsilon}_R +n^{\varepsilon}_R\overline{\nabla}\cdot\widetilde{\textbf{u}} +\textbf{u}^{\varepsilon}_R\cdot\overline{\nabla}\widetilde{n}+\varepsilon R_n=0,\label{re-1}\\
\partial_t\textbf{u}^{\varepsilon}_R-\frac{V\textbf{e}_1-\textbf{u}}{\varepsilon} \cdot\overline{\nabla}\textbf{u}^{\varepsilon}_R +\textbf{u}^{\varepsilon}_R\cdot\overline{\nabla}\widetilde{\textbf{u}} +\frac{T_i}{\varepsilon n}\overline{\nabla}n^{\varepsilon}_R\nonumber\\
\ \ \ \ \ \ \ \ \ \ \ -\frac{T_i \textbf{p}}{\varepsilon n}n^{\varepsilon}_R+\varepsilon \textbf{R}_\textbf{u}-\frac{T_i \varepsilon}{n}\textbf{R}_{T} =-\frac{1}{\varepsilon}\overline{\nabla}\phi^{\varepsilon}_R +\frac{b}{\varepsilon^{3/2}}\textbf{u}^{\varepsilon}_R\times \textbf{e}_1,\label{re-2}\\
\varepsilon\overline{\Delta}\phi^{\varepsilon}_R=\phi^{\varepsilon}_R-n^{\varepsilon}_R +\varepsilon\phi^{(1)}\phi^{\varepsilon}_R +\varepsilon^{3/2}R_{\phi},\label{re-3}
\end{numcases}
\end{subequations}
where the constant $b=0$ in 2D and $b=1$ in 3D to indicate the presence of a static magnetic field. Furthermore, \eqref{re-3} is equivalent to the following
\begin{equation}\label{e35}
\begin{split}
\varepsilon\overline{\Delta}\phi^{\varepsilon}_R=\phi^{\varepsilon}_R-n^{\varepsilon}_R +\varepsilon\phi^{(1)}\phi^{\varepsilon}_R +\frac{\varepsilon^2}{2}(\phi^{\varepsilon}_R)^2 +\varepsilon^{2}\overline{R}_{\phi}.
\end{split}
\end{equation}
Here, $R_{\phi}$ and $\overline{R}_{\phi}$ satisfy the estimates in Lemma \ref{lem-A}.
\end{proposition}

From Proposition \ref{prop-kp1} and \ref{prop-kp4}, we may assume that the known profiles $(\widetilde n, \widetilde{\textbf{u}}, \widetilde{\phi})$ are smooth enough 
such that there exist some $C>0$ and some $s\geq 4$,
\begin{equation}\label{assump}
\begin{split}
\sup_{[0,\tau_*]}\|(\widetilde n, \widetilde{\textbf{u}}, \widetilde{\phi}),R_n,\textbf{R}_\textbf{u},\textbf{R}_T\|_{H^s}\leq C,
\end{split}
\end{equation}
where $\tau_*$ is the existence time in Proposition \ref{prop-kp1} or Proposition \eqref{prop-ZKE}. 

\subsection{Main results}
Now, we are in a good position to state the main results of this paper. We first introduce the following $\varepsilon$-dependent norms. We denote (the triple norm)
\begin{equation}\label{tri-norm}
\begin{cases}
|\!|\!|n^{\varepsilon}_R|\!|\!|_{s'}^2 =\|\textbf{u}^{\varepsilon}_R\|_{H^{s'}}^2, 
\\
|\!|\!|\textbf{u}^{\varepsilon}_R|\!|\!|_{s'}^2 =\|\textbf{u}^{\varepsilon}_R\|_{H^{s'}}^2 +\varepsilon \|\overline{\nabla}\textbf{u}^{\varepsilon}_R\|_{H^{s'}}^2,
\\
|\!|\!|\phi^{\varepsilon}_R|\!|\!|_{s'}^2=\|\phi^{\varepsilon}_R\|_{H^{s'}}^2 +\varepsilon\|\overline{\nabla}\phi^{\varepsilon}_R\|_{H^{s'}}^2 +\varepsilon^2\|\overline{\Delta}\phi^{\varepsilon}_R\|_{H^{s'}}^2, 
\end{cases}
\end{equation}
where $\|\cdot\|_{H^{s'}}$ is the standard Sobolev norm.
\begin{theorem}\label{th}
Let $s\geq 4$ be such that \eqref{assump} holds and $(n^{(i)},\textbf{u}^{(i)},\phi^{(i)})\in H^{s}$ for $1\leq i\leq3$ (resp. $1\leq i\leq 6$) be solutions constructed on the interval $[0,\tau_*)$ in Proposition \ref{prop-kp1} (resp. Proposition \ref{prop-ZKE}) with initial data $(n^{(i)}_0, \textbf{u}_{0}^{(i)},\phi^{(i)}_0)\in H^{s}$
. Let $4\leq s'\leq s$ and assume that the initial data $(n_0, \textbf{u}_0, \phi_0)$ for the EP system \eqref{ep2} (resp. EP system \eqref{equ1}) has the expansion of the form \eqref{exp-kp} (resp. \eqref{expan}) and $({n^{\varepsilon}_R}, {u^{\varepsilon}_R}, {\phi^{\varepsilon}_R})|_{t=0}=({n^{\varepsilon}_R}_0, {u^{\varepsilon}_R}_0, {\phi^{\varepsilon}_R}_0)$ satisfy \eqref{re}. Then for any $0<\tau_0<\tau_*$, there exist $\varepsilon_0>0$ and $C_{\tau_0}>0$ such that when $0<\varepsilon<\varepsilon_0$, the solutions of the EP system \eqref{ep2} (resp. \eqref{equ1}) with initial data $(n_0,\textbf{u}_0,\phi_0)$ can be expressed in the expansion \eqref{exp-kp} (resp. \eqref{expan}), such that the solutions $({n^{\varepsilon}_R}, {u^{\varepsilon}_R}, {\phi^{\varepsilon}_R})$ of \eqref{re} satisfy
\par 1) when $T_i>0$, for either $d=2$ or $d=3$,
  \begin{equation}\label{estimate1}
  \sup_{[0,\tau_0]}\|(n^{\varepsilon}_R,\textbf{u}^{\varepsilon}_R, \phi^{\varepsilon}_R)\|_{H^{s'}}^2\leq C_{\tau_0}(1+\|({n^{\varepsilon}_R}_0, {u^{\varepsilon}_R}_0, {\phi^{\varepsilon}_R}_0)\|_{H^{s'}}^2),
  \end{equation}
\par 2) when $T_i=0$, for $d=3$
  \begin{equation}\label{estimate2}
  \sup_{[0,\tau_0]}|\!|\!|(n^{\varepsilon}_R,\textbf{u}^{\varepsilon}_R, \phi^{\varepsilon}_R)|\!|\!|_{s'}^2\leq C_{\tau_0}(1+|\!|\!|({n^{\varepsilon}_R}_0, {u^{\varepsilon}_R}_0, {\phi^{\varepsilon}_R}_0)|\!|\!|_{{s'}}^2),
  \end{equation}
where $|\!|\!|(n^{\varepsilon}_R,\textbf{u}^{\varepsilon}_R, \phi^{\varepsilon}_R)|\!|\!|_{s'}^2 =|\!|\!|n^{\varepsilon}_R|\!|\!|_{s'}^2 +|\!|\!|\textbf{u}^{\varepsilon}_R|\!|\!|_{s'}^2 +|\!|\!|\phi^{\varepsilon}_R|\!|\!|_{s'}^2$ is defined in \eqref{tri-norm}. 
\end{theorem}

\begin{remark}
When $T_i>0$, the result for the KP-II limit in 2D can be generalized to any dimensions $d\geq 2$, following the same lines of the proof in Section 3. But the index $s'$ will be replaced by a greater one depending on the dimension $d\geq 2$.
\end{remark}

This theorem provides a rigorous justification of the  Kadomtsev-Petviashvili equation in 2D and the Zakharov-Kuznetsov equation in 3D from the Euler-Poisson system in the long wavelength limit (in the ZKE limit, the Euler-Poisson we used is \eqref{equ1}, see Appendix B). To prove this result, we need to derive a uniform bound for the remainder $(n^{\varepsilon}_R, \textbf{u}^{\varepsilon}_R, \phi^{\varepsilon}_R)$ in \eqref{re}. However, this is not starightforward, especially when $T_i=0$. When $T_i=0$, the system of \eqref{re-1} and \eqref{re-2} for $(n^{\varepsilon}_R, \textbf{u}^{\varepsilon}_R)$ does not match the common structure of Friedrich's symmetric systems. Because of this, the approach by Grenier \cite{Gre97} cannot be applied, which depends heavily on the symmetrizability of the underlying system. To overcome this difficulty, we need to combine the energy estimates with the delicate structure of the Poisson equation carefully. This is why we introduce the norm $|\!|\!|\cdot|\!|\!|_{s'}$ in the case $T_i=0$.

\section{Proof of Theorem \ref{th} for $T_i>0$}
\setcounter{section}{3}\setcounter{equation}{0}
This section is dedicated to the proof of Theorem \ref{th} for the case of $T_i>0$. For this purpose, we need only to derive a uniform bound for the remainder equation \eqref{re}. To slightly simplify the presentation, we assume that \eqref{re} has smooth solutions in a small time $\tau_{\varepsilon}$ dependent on $\varepsilon$. Let $\tilde C$ be a constant, which will be determined later, much larger than the bound of $\|(n^{\varepsilon}_R,\textbf{u}^{\varepsilon}_R,\phi^{\varepsilon}_R)\|_{s'}$, such that on $[0,\tau_{\varepsilon}]$
\begin{equation}\label{assumption}
\begin{split}
\sup_{[0,\tau_{\varepsilon}]}\|(n^{\varepsilon}_R,\textbf{u}^{\varepsilon}_R, \phi^{\varepsilon}_R)\|_{s'}\leq \tilde C.
\end{split}
\end{equation}
We will prove that $\tau_{\varepsilon}>\tau_0$ as $\varepsilon\to 0$ for any $0<\tau_0<\tau_*$, where $\tau_*$ is the existence time of the limit equation \eqref{KPE} or \eqref{ZKE}. Recalling the expressions for $n$ and $\textbf{u}$ in \eqref{e3}, we immediately know that there exists some $\varepsilon_1=\varepsilon_1(\tilde C)>0$ such that on $[0,\tau_{\varepsilon}]$,
\begin{equation}\label{assumption1}
\begin{split}
1/2<n<3/2,\ \ \ |\textbf{u}|\leq 1/2,
\end{split}
\end{equation}
for all $0<\varepsilon<\varepsilon_1$.

\begin{lemma}\label{L1}
Let $(n^{\varepsilon}_R,\textbf{u}^{\varepsilon}_R,\phi^{\varepsilon}_R)$ be a solution to \eqref{re} and $0\leq k\leq s'\leq s$ be an integer. There exist some $0<\varepsilon_1<1$ and $C,C_0>0$ such that for every $0<\varepsilon<\varepsilon_1$,
\begin{equation}\label{equ90}
\begin{split}
\|n^{\varepsilon}_R\|_{H^{k}}^2\leq C\|\phi^{\varepsilon}_R\|_{H^{k}}^2 +C\varepsilon\|\overline{\nabla}\phi^{\varepsilon}_R\|_{H^{k}}^2 +C\varepsilon^2\|\overline{\Delta}\phi^{\varepsilon}_R\|_{H^{k}}^2 +CC_0^2\varepsilon^2,
\end{split}
\end{equation}
\begin{equation}\label{equ91}
\begin{split}
\|\phi^{\varepsilon}_R\|_{H^{k}}^2 +\varepsilon\|\overline{\nabla}\phi^{\varepsilon}_R\|_{H^{k}}^2 +\varepsilon^2\|\overline{\Delta}\phi^{\varepsilon}_R\|_{H^{k}}^2\leq C\|n^{\varepsilon}_R\|_{H^{k}}^2+CC_0^2\varepsilon^2.
\end{split}
\end{equation}
\end{lemma}
\begin{proof}
Taking $H^{k}$ inner product of \eqref{re-3} with $\phi^{\varepsilon}_R$, and integrating by parts, we have
\begin{equation}\label{equ13}
\begin{split}
\varepsilon\|\overline{\nabla}\phi^{\varepsilon}_R\|_{H^{k}}^2 +\|\phi^{\varepsilon}_R\|_{H^{k}}^2= & \langle n^{\varepsilon}_R,\phi^{\varepsilon}_R\rangle_{H^{k}} -\langle \varepsilon\phi^{(1)}\phi^{\varepsilon}_R, \phi^{\varepsilon}_R\rangle_{H^{k}} -\langle \varepsilon^{3/2}R_{\phi}, \phi^{\varepsilon}_R\rangle_{H^{k}}\\
\leq & 2\|n^{\varepsilon}_R\|_{H^{k}}^2+\frac14\|\phi^{\varepsilon}_R\|_{H^{k}}^2 +C\varepsilon\|\phi^{\varepsilon}_R\|_{H^{k}}^2 +2\varepsilon^3\|R_{\phi}\|_{H^{k}}^2.
\end{split}
\end{equation}
From \eqref{assumption} and Lemma \ref{lem-A}, there exist some constant $\varepsilon_1=\varepsilon_1(\tilde C)>0$ and $C_0:=C_1(1)>0$, such that $C_1(\sqrt{\varepsilon}\|\phi^{\varepsilon}_R\|_{H^{\alpha}})\leq C_0$ when $0<\varepsilon<\varepsilon_1$. This enables us to get $\|R_{\phi}\|_{H^{k}}^2\leq 2C_0(1+\|\phi^{\varepsilon}_R\|_{H^{k}}^2)$.
Therefore, there exists some $\varepsilon_1=\varepsilon_1(\tilde C)>0$ (still denoted as $\varepsilon_1$) such that when $\varepsilon<\varepsilon_1$, we have
\begin{equation}\label{equ14}
\begin{split}
C\varepsilon\|\phi^{\varepsilon}_R\|_{H^{k}}^2 +2\varepsilon^3\|R_{\phi}\|_{H^{k}}^2\leq \frac14\|\phi^{\varepsilon}_R\|_{H^{k}}^2+4C_0^2\varepsilon^2.
\end{split}
\end{equation}
Combining \eqref{equ13} and \eqref{equ14}, we have
\begin{equation}\label{equ15}
\begin{split}
\varepsilon\|\overline{\nabla}\phi^{\varepsilon}_R\|_{H^{k}}^2 +\|\phi^{\varepsilon}_R\|_{H^{k}}^2 \leq & C\|n^{\varepsilon}_R\|_{H^{k}}^2+CC_0^2\varepsilon^2,
\end{split}
\end{equation}
for some universal constant $C>0$.

Taking $H^{k}$ inner product of \eqref{rem-3} with $\varepsilon\overline{\Delta}\phi^{\varepsilon}_R$ and integrating by parts, we have similarly
\begin{equation*}
\begin{split}
\varepsilon^2\|\overline{\Delta}&\phi^{\varepsilon}_R\|_{H^{k}}^2 +\varepsilon\|\overline{\nabla}\phi^{\varepsilon}_R\|_{H^{k}}^2\\
=& \varepsilon\langle n^{\varepsilon}_R,\overline{\Delta}\phi^{\varepsilon}_R\rangle_{H^{k}} -\varepsilon^2\langle \phi^{(1)}\phi^{\varepsilon}_R,\overline{\Delta}\phi^{\varepsilon}_R\rangle_{H^{k}} +\varepsilon^{5/2}\langle \overline{\nabla}R_{\phi},\overline{\nabla}\phi^{\varepsilon}_R\rangle_{H^{k}}\\
\leq & 2\|n^{\varepsilon}_R\|_{H^{k}}^2 +\frac{\varepsilon^2}{8}\|\overline{\Delta}\phi^{\varepsilon}_R\|_{H^{k}}^2 +C\varepsilon^2\|\overline{\nabla}\phi^{\varepsilon}_R\|_{H^{k}}^2 +\frac{\varepsilon}{8}\|\overline{\nabla}\phi^{\varepsilon}_R\|_{H^{k}}^2 +2\varepsilon^4\|\overline{\nabla}R_{\phi}\|_{H^{k}}^2.
\end{split}
\end{equation*}
Similarly, from Lemma \ref{lem-A}, there exists some $\varepsilon_1=\varepsilon_1(\tilde C)>0$ such that when $\varepsilon<\varepsilon_1$, we have
\begin{equation}
\begin{split}
\|\overline{\nabla}R_{\phi}\|_{H^{k}}\leq 2C_0 (1+\|\overline{\nabla}\phi^{\varepsilon}_R\|_{H^k}).
\end{split}
\end{equation}
It then follows that
\begin{equation}\label{equ16}
\begin{split}
\varepsilon^2\|\overline{\Delta}&\phi^{\varepsilon}_R\|_{H^{k}}^2 +\varepsilon\|\overline{\nabla}\phi^{\varepsilon}_R\|_{H^{k}}^2\leq C\|n^{\varepsilon}_R\|_{H^{k}}^2+CC_0^2\varepsilon^2.
\end{split}
\end{equation}
By adding \eqref{equ15} and \eqref{equ16} together, we know there exist some constants $\varepsilon_1>0$, $C$ and $C_0$ such that
\begin{equation}\label{equ17}
\begin{split}
\varepsilon^2\|\overline{\Delta}&\phi^{\varepsilon}_R\|_{H^{k}}^2 +\varepsilon\|\overline{\nabla}\phi^{\varepsilon}_R\|_{H^{k}}^2 +\|\phi^{\varepsilon}_R\|_{H^{k}}^2 \leq C\|n^{\varepsilon}_R\|_{H^{k}}^2+CC_0^2\varepsilon^2.
\end{split}
\end{equation}

On the other hand, by taking $H^{\alpha}$ norm of \eqref{rem-3}, we have
\begin{equation}\label{equ18}
\begin{split}
\|n^{\varepsilon}_R\|_{H^{k}}^2\leq & \varepsilon^2\|\overline{\Delta}\phi^{\varepsilon}_R\|_{H^{k}}^2 +C\|\phi^{\varepsilon}_R\|_{H^{k}}^2 +C\varepsilon^3\|R_\phi\|_{H^{k}}^2\\
\leq & \varepsilon^2\|\overline{\Delta}\phi^{\varepsilon}_R\|_{H^{k}}^2 +C\|\phi^{\varepsilon}_R\|_{H^{k}}^2 +CC_0^2\varepsilon^2.
\end{split}
\end{equation}
Combining \eqref{equ17} and \eqref{equ18}, we complete the proof.
\end{proof}


\begin{lemma}\label{L2}
Let $(n^{\varepsilon}_R, \textbf{u}^{\varepsilon}_R, \phi^{\varepsilon}_R)$ be a solution to \eqref{rem} and $s'\leq s$  be an integer. Then for any $1\leq k\leq s'$, there exist some $0<\varepsilon_1<1$ and $C>0$ such that for every $0<\varepsilon<\varepsilon_1$,
\begin{equation}\label{equ19}
\begin{split}
\|\varepsilon\partial_tn^{\varepsilon}_R\|_{H^{k-1}}^2\leq C \{1+\|\textbf{u}^{\varepsilon}_R\|_{H^{\delta+1}}^2 +\|n^{\varepsilon}_R\|_{H^{\delta+1}}^2\},
\end{split}
\end{equation}
or equivalently,
\begin{equation}\label{e78}
\begin{split}
\|\varepsilon\partial_tn^{\varepsilon}_R\|_{H^{k-1}}^2\leq C \{1+\|\textbf{u}^{\varepsilon}_R\|_{H^{\delta+1}}^2 +|\!|\!|\phi^{\varepsilon}_R|\!|\!|_{{\delta+1}}^2\},
\end{split}
\end{equation}
where $\delta=\max\{2,k-1\}$.
\end{lemma}
\begin{proof}
Multiply \eqref{re-1} by $\varepsilon$, and then take $H^{k-1}$ norm to obtain
\begin{equation}\label{equ20}
\begin{split}
\|\varepsilon\partial_tn^{\varepsilon}_R\|_{H^{k-1}}^2\leq & C\|\overline{\nabla} n^{\varepsilon}_R\|_{H^{k-1}}^2 +C\|\textbf{u}\cdot\overline{\nabla} n^{\varepsilon}_R\|_{H^{k-1}}^2 +C\|n\overline{\nabla}\cdot\textbf{u}^{\varepsilon}_R\|_{H^{k-1}}^2\\
&+C\varepsilon^2\|n^{\varepsilon}_R \cdot\overline{\nabla}\widetilde{\textbf{u}}\|_{H^{k-1}}^2 +C\varepsilon^2\|\textbf{u}^{\varepsilon}_R \cdot\overline{\nabla}\widetilde{n}\|_{H^{k-1}}^2 +C\varepsilon^4\|R_n\|_{H^{k-1}}^2.
\end{split}
\end{equation}
Recall that $n=1+\varepsilon\widetilde n+\varepsilon^2n^{\varepsilon}_R$ and $\textbf{u}=\varepsilon\widetilde{\textbf{u}} +\varepsilon^2\textbf{u}^{\varepsilon}_R$ in \eqref{e3}. By Lemma \ref{Le-inequ}, we have
\begin{equation}\label{equ21}
\begin{split}
\|\textbf{u}\cdot\overline{\nabla}n^{\varepsilon}_R\|_{H^{k-1}}^2\leq & C\varepsilon^2\|(\widetilde{\textbf{u}} +\varepsilon\textbf{u}^{\varepsilon}_R) \cdot\overline{\nabla}n^{\varepsilon}_R\|_{H^{k-1}}^2\\
\leq & C\varepsilon^2\|\overline{\nabla}n^{\varepsilon}_R\|_{H^{k-1}}^2 +C\varepsilon^4\{\|\textbf{u}^{\varepsilon}_R\|_{H^{k-1}}^2\|\overline{\nabla} n^{\varepsilon}_R\|_{L^{\infty}}^2+\|\overline{\nabla} n^{\varepsilon}_R\|_{H^{k-1}}^2 \|\textbf{u}^{\varepsilon}_R\|_{L^{\infty}}^2\}\\
\leq & C\varepsilon^2 \{1+C(\varepsilon^2\|\textbf{u}^{\varepsilon}_R\|_{H^{\delta}}^2)\} \|\overline{\nabla} n^{\varepsilon}_R\|_{H^{\delta}}^2,
\end{split}
\end{equation}
where $\delta=\max\{2,k-1\}$. Similarly, we have
\begin{equation}\label{equ22}
\begin{split}
\|n\overline{\nabla}\cdot\textbf{u}^{\varepsilon}_R\|_{H^{k-1}}^2\leq & \|(1+\varepsilon\widetilde n +\varepsilon^2n^{\varepsilon}_R) \overline{\nabla}\cdot\textbf{u}^{\varepsilon}_R\|_{H^{k-1}}^2\\
\leq & C\|\overline{\nabla}\textbf{u}^{\varepsilon}_R\|_{H^{k-1}}^2 +C\varepsilon^4\{\|\overline{\nabla}\textbf{u}^{\varepsilon}_R\|_{H^{k-1}}^2 \|n^{\varepsilon}_R\|_{L^{\infty}}^2 +\|n^{\varepsilon}_R\|_{H^{k-1}}^2\|\overline{\nabla} \textbf{u}^{\varepsilon}_R\|_{L^{\infty}}^2\}\\
\leq & C\{1+\varepsilon^2\|\overline{\nabla} \textbf{u}^{\varepsilon}_R\|_{H^{\delta}}^2\} \{\|n^{\varepsilon}_R\|_{H^{\delta}}^2 +\|\textbf{u}^{\varepsilon}_R\|_{H^{\delta+1}}^2\},
\end{split}
\end{equation}
where $\delta=\max\{2,k-1\}$. On the other hand, since $\widetilde{\textbf{u}},\widetilde n\in H^{s}$ and $R_n$ depends only on $\widetilde{\textbf{u}}$ and $\widetilde n$, the last three terms on the RHS of \eqref{equ20} are easily bounded by
\begin{equation}\label{equ23}
\begin{split}
C\varepsilon^2\{1+\|\textbf{u}^{\varepsilon}_R\|_{H^{k-1}}^2 +\| n^{\varepsilon}_R\|_{H^{k-1}}^2\}.
\end{split}
\end{equation}
Inserting \eqref{equ21}-\eqref{equ23} into \eqref{equ20}, we obtain
\begin{equation}\label{equ24}
\begin{split}
\|\varepsilon\partial_tn^{\varepsilon}_R\|_{H^{k-1}}^2\leq C\{1+\varepsilon^2\|\textbf{u}^{\varepsilon}_R\|_{H^{\delta+1}}^2\} \{1+\|n^{\varepsilon}_R\|_{H^{\delta+1}}^2 +\|\textbf{u}^{\varepsilon}_R\|_{H^{\delta}}^2\},
\end{split}
\end{equation}
where $\delta=\max\{2,k-1\}$. Since $\|\textbf{u}^{\varepsilon}_R\|_{H^{\delta+1}}^2\leq |\!|\!|(u^{\varepsilon}_R,\phi^{\varepsilon}_R)|\!|\!|^2\leq \tilde C$ by assumption \eqref{assumption}, there exists some $\varepsilon_1\in (0,1)$ depending on $\tilde C$ such that $\varepsilon^2\|\textbf{u}^{\varepsilon}_R\|_{H^{\delta+1}}^2\leq 1$ when $0<\varepsilon<\varepsilon_1$. Therefore \eqref{equ19} is proved. Invoking Lemma \ref{L1}, we obtain
\begin{equation}\label{equ24}
\begin{split}
\|\varepsilon\partial_tn^{\varepsilon}_R\|_{H^{k-1}}^2\leq C \{1+\|\textbf{u}^{\varepsilon}_R\|_{H^{\delta+1}}^2 +|\!|\!|\phi^{\varepsilon}_R|\!|\!|_{{\delta+1}}^2\},
\end{split}
\end{equation}
for any $1\leq k\leq s'$, where $\delta=\max\{2,k-1\}$ and $|\!|\!|\phi^{\varepsilon}_R|\!|\!|_{{\delta+1}}^2$ is given by \eqref{tri-norm}. The proof is complete.
\end{proof}


\begin{lemma}\label{L3}
Let $(n^{\varepsilon}_R,\textbf{u}^{\varepsilon}_R,\phi^{\varepsilon}_R)$ be a solution to \eqref{rem} and $s'\leq s$  be an integer. Then for any $1\leq k\leq s'$, there exist some $0<\varepsilon_1<1$ and $C,C_0>0$ such that for every $0<\varepsilon<\varepsilon_1$,
\begin{equation}\label{e79}
\begin{split}
|\!|\!|\partial_t\phi^{\varepsilon}_R|\!|\!|_{{k-1}}^2 \leq  C\|\partial_tn^{\varepsilon}_R\|_{H^{k-1}}^2+ C\varepsilon\|\phi^{\varepsilon}_R\|_{H^{k-1}}^2+CC_0^2\varepsilon.
\end{split}
\end{equation}
where $\delta=\max\{2,k-1\}$.
\end{lemma}
\begin{proof}
Let $1\leq k\leq s'$. We take $\partial_t$ of \eqref{re-3} and then take $H^{k-1}$ inner product with $\partial_t\phi^{\varepsilon}_R$ to obtain
\begin{equation*}
\begin{split}
\varepsilon\|\overline{\nabla}\partial_t&\phi^{\varepsilon}_R\|_{H^{k-1}}^2 +\|\partial_t\phi^{\varepsilon}_R\|_{H^{k-1}}^2\\
= & \langle \partial_tn^{\varepsilon}_R,\partial_t\phi^{\varepsilon}_R\rangle_{H^{k-1}} -\varepsilon\langle\partial_t(\phi^{(1)}\phi^{\varepsilon}_R), \partial_t\phi^{\varepsilon}_R\rangle_{H^{k-1}} -\varepsilon^{3/2}\langle\partial_tR_{\phi}, \partial_t\phi^{\varepsilon}_R\rangle_{H^{k-1}}\\
\leq & 2\|\partial_tn^{\varepsilon}_R\|_{H^{k-1}}^2 +\frac14\|\partial_t\phi^{\varepsilon}_R\|_{H^{k-1}}^2 +C\varepsilon\|\partial_t\phi^{\varepsilon}_R\|_{H^{k-1}}^2 +C\varepsilon\|\phi^{\varepsilon}_R\|_{H^{k-1}}^2 +2\varepsilon^3\|\partial_tR_{\phi}\|_{H^{k-1}}^2.
\end{split}
\end{equation*}
The fourth term in the last line comes from the term $\partial_t(\phi^{(1)}\phi^{\varepsilon}_R)$ when $\partial_t$ acts on $\phi^{(1)}$. As in the proof of Lemma \ref{L1}, by \eqref{e52-2} in Lemma \ref{lem-A}, there exists some $\varepsilon_1=\varepsilon_1(\tilde C)>0$ such that when $\varepsilon<\varepsilon_1$, we have
\begin{equation}\label{equ25}
\begin{split}
\varepsilon\|\overline{\nabla}\partial_t&\phi^{\varepsilon}_R\|_{H^{k-1}}^2 +\|\partial_t\phi^{\varepsilon}_R\|_{H^{k-1}}^2\leq C\|\partial_tn^{\varepsilon}_R\|_{H^{k-1}}^2+ C\varepsilon\|\phi^{\varepsilon}_R\|_{H^{k-1}}^2+CC_0^2\varepsilon^2
\end{split}
\end{equation}
for some universal constant $C>0$.

On the other hand, by taking $\partial_t$ of \eqref{rem-3} and then taking $H^{k-1}$ inner product with $\varepsilon\partial_t\overline{\Delta}\phi^{\varepsilon}_R$, we obtain
\begin{equation}\label{equ26}
\begin{split}
\varepsilon^2\|\partial_t\overline{\Delta}\phi^{\varepsilon}_R\|_{H^{k-1}}^2 &+\varepsilon\|\partial_t\overline{\nabla} \phi^{\varepsilon}_R\|_{H^{k-1}}^2 \leq  C\|\partial_tn^{\varepsilon}_R\|_{H^{k-1}}^2+ C\varepsilon^2(\|\partial_t\phi^{\varepsilon}_R\|_{H^{k-1}}^2 +\|\phi^{\varepsilon}_R\|_{H^{k-1}}^2) +CC_0^2\varepsilon,
\end{split}
\end{equation}
for $\varepsilon<\varepsilon_1$ for some $0<\varepsilon_1<1$. Adding \eqref{equ25} and \eqref{equ26}, by choosing $\varepsilon_1$ sufficiently small such that $C\varepsilon_1^2\leq 1/2$, we obtain
\begin{equation}\label{equ27}
\begin{split}
\varepsilon^2\|\partial_t\overline{\Delta}\phi^{\varepsilon}_R\|_{H^{k-1}}^2  &+\varepsilon\|\partial_t\overline{\nabla}\phi^{\varepsilon}_R\|_{H^{k-1}}^2 +\|\partial_t\phi^{\varepsilon}_R\|_{H^{k-1}}^2
\leq  C\|\partial_tn^{\varepsilon}_R\|_{H^{k-1}}^2+ C\varepsilon\|\phi^{\varepsilon}_R\|_{H^{k-1}}^2+CC_0^2\varepsilon.
\end{split}
\end{equation}
The proof is complete.
\end{proof}
\begin{corollary}\label{cor2}
Under the same assumptions of Lemma \ref{L3}, we have
\begin{equation}\label{e22}
\begin{split}
|\!|\!|\varepsilon\partial_t\phi^{\varepsilon}_R|\!|\!|_{{k-1}}^2 \leq C\{1+\|(\textbf{u}^{\varepsilon}_R, n^{\varepsilon}_R, \phi^{\varepsilon}_R)\|_{H^{\delta+1}}^2\},
\end{split}
\end{equation}
or equivalently,
\begin{equation}\label{e75}
\begin{split}
|\!|\!|\varepsilon\partial_t\phi^{\varepsilon}_R|\!|\!|_{{k-1}}^2 \leq C\{1+\|\textbf{u}^{\varepsilon}_R\|_{H^{\delta+1}}^2 +|\!|\!|\phi^{\varepsilon}_R|\!|\!|_{{\delta+1}}^2\},
\end{split}
\end{equation}
where $\delta=\max\{2,k-1\}$.
\end{corollary}
\begin{proof}
The proof is complete by multiplying \eqref{e79} with $\varepsilon^2$ and then using Lemma \ref{L2}.
\end{proof}

\begin{corollary}\label{rmk5}
Let $\alpha$ be a multi-index with $|\alpha|=k$, then
\begin{equation*}
\begin{cases}
\varepsilon^3\|\varepsilon\partial_t\overline{\nabla} \partial_x^{\alpha}\phi^{\varepsilon}_R\|_{L^2}^2\leq C\varepsilon^2\|\varepsilon\partial_t\overline{\Delta} \phi^{\varepsilon}_R\|_{H^{k-1}}^2,\ \ \ d=2;\\
\varepsilon^2\|\varepsilon\partial_t\overline{\nabla} \partial_x^{\alpha}\phi^{\varepsilon}_R\|_{L^2}^2\leq C\varepsilon^2\|\varepsilon\partial_t\overline{\Delta} \phi^{\varepsilon}_R\|_{H^{k-1}}^2,\ \ \ d=3,
\end{cases}
\end{equation*}
where $x\in \Bbb R^d$ and $\overline{\nabla}$ is defined in \eqref{e1} and $\overline{\Delta}=\overline{\nabla}\cdot\overline{\nabla}$.
\end{corollary}
\begin{proof}
By Reisz theorem \cite{Stein70}, we have%
\begin{equation*}
\begin{split}
\|\sqrt{\varepsilon}\partial_{x_i}\overline{\nabla}f\|_{L^2}\leq C\|\overline{\Delta}f\|_{L^2},
\end{split}
\end{equation*}
for those $f$ that makes sense. The proof is complete by letting $f=\varepsilon^2\partial_t\phi^{\varepsilon}_R$.
\end{proof}

\begin{proposition}\label{prop-kp2}
Let $s'\geq 4$ be an integer and $(n^{\varepsilon}_R,\textbf{u}^{\varepsilon}_R,\phi^{\varepsilon}_R)$ be a solution to \eqref{re}. Then for any integer $0\leq k\leq s'$, there holds
\begin{equation}\label{e-prop1}
\begin{split}
\frac12\frac{d}{dt}\|\partial_x^{\alpha}\textbf{u}^{\varepsilon}_R\|_{L^2}^2 &+\frac12\frac{d}{dt}\int\frac{1+\varepsilon\phi^{(1)} +\varepsilon^2\phi^{\varepsilon}_R}{n} |\partial_x^{\alpha}\phi^{\varepsilon}_R|^2\\
&+\frac12\frac{d}{dt}\int\frac{\varepsilon}{n} |\overline{\nabla}\partial_x^{\alpha}\phi^{\varepsilon}_R|^2 +\frac12\frac{d}{dt}\int\frac{T_i}{n^2} |\partial_x^{\alpha}n^{\varepsilon}_R|^2dx\\
\leq & CC_1(C_1+\varepsilon\|(\textbf{u}^{\varepsilon}_R, n^{\varepsilon}_R, \phi^{\varepsilon}_R)\|_{H^{s'}}^2)\{1 +\|(\textbf{u}^{\varepsilon}_R, n^{\varepsilon}_R, \phi^{\varepsilon}_R)\|_{H^{s'}}^2\},
\end{split}
\end{equation}
where $\alpha$ is any multi-index with $|\alpha|=k$.
\end{proposition}

For clarity, we divide the proof of this proposition into the following four lemmas.

\begin{lemma}\label{L8}
Let $s'\geq 4$, $k\leq s'$ be two non-negative integers and $\alpha$ be any multi-index with $|\alpha|=k$. Then for any solution  $(n^{\varepsilon}_R, \textbf{u}^{\varepsilon}_R, \phi^{\varepsilon}_R)$ of \eqref{re}, we have
\begin{equation}\label{e93}
\begin{split}
\frac12\frac{d}{dt}\|\partial_x^{\alpha}\textbf{u}^{\varepsilon}_R\|_{L^2}^2 &+\frac12\frac{d}{dt}\int\frac{1+\varepsilon\phi^{(1)} +\varepsilon^2\phi^{\varepsilon}_R}{n} |\partial_x^{\alpha}\phi^{\varepsilon}_R|^2 +\frac12\frac{d}{dt}\int\frac{\varepsilon}{n} |\overline{\nabla}\partial_x^{\alpha}\phi^{\varepsilon}_R|^2\\
\leq & CC_1(C_1+\varepsilon\|(\textbf{u}^{\varepsilon}_R, n^{\varepsilon}_R, \phi^{\varepsilon}_R)\|_{H^{s'}}^2)\{1 +\|(\textbf{u}^{\varepsilon}_R, n^{\varepsilon}_R, \phi^{\varepsilon}_R)\|_{H^{s'}}^2\}+I_{71},
\end{split}
\end{equation}
where
\begin{equation}\label{e94}
\begin{split}
I_{71}=-T_i\langle\frac{1}{\varepsilon n}\partial_x^{\alpha}\overline{\nabla}n^{\varepsilon}_R, \partial_x^{\alpha}\textbf{u}^{\varepsilon}_R\rangle_0.
\end{split}
\end{equation}
\end{lemma}
\begin{proof}
Let $\alpha$ be any multi-index with $|\alpha|=k$. We take $\partial_x^{\alpha}$ of \eqref{re-2} and then take $L^2$ inner product with $\partial_x^{\alpha}\textbf{u}^{\varepsilon}_R$, to obtain
\begin{equation}\label{e4}
\begin{split}
\frac12\frac{d}{dt}\|\partial_x^{\alpha}\textbf{u}^{\varepsilon}_R\|^2_{L^2}=& \frac{V}{\varepsilon}\langle\partial_{x_1}\partial_x^{\alpha} \textbf{u}^{\varepsilon}_R, \partial_x^{\alpha} \textbf{u}^{\varepsilon}_R\rangle_{0} -\frac{1}{\varepsilon}\langle \partial_x^{\alpha} (\textbf{u}\cdot\overline{\nabla}\textbf{u}^{\varepsilon}_R), \partial_x^{\alpha}\textbf{u}^{\varepsilon}_R\rangle_{0} -\langle\partial_x^{\alpha} (\textbf{u}^{\varepsilon}_R\cdot\overline{\nabla}\widetilde{\textbf{u}}), \partial_x^{\alpha}\textbf{u}^{\varepsilon}_R\rangle_{0}\\
&-\varepsilon\langle \partial_x^{\alpha}\textbf{R}_\textbf{u}, \partial_x^{\alpha}\textbf{u}^{\varepsilon}_R\rangle_{0} +\frac{b}{\varepsilon^{3/2}}\langle \partial_x^{\alpha}\textbf{u}^{\varepsilon}_R\times \textbf{e}_1, \partial_x^{\alpha}\textbf{u}^{\varepsilon}_R\rangle_{0} -\frac{1}{\varepsilon}\langle\partial_x^{\alpha}\overline{\nabla}\phi^{\varepsilon}_R, \partial_x^{\alpha}\textbf{u}^{\varepsilon}_R\rangle_{0}\\
&-T_i\langle\partial_x^{\alpha} (\frac{\overline{\nabla}n^{\varepsilon}_R}{\varepsilon n}), \partial_x^{\alpha}\textbf{u}^{\varepsilon}_R\rangle_{0} +T_i\langle\partial_x^{\alpha}(\frac{\textbf{p}}{n}n^{\varepsilon}_R), \partial_x^{\alpha}\textbf{u}^{\varepsilon}_R\rangle_{0} +T_i\langle\partial_x^{\alpha}(\frac{\varepsilon \textbf{R}_T}{n}),\partial_x^{\alpha}\textbf{u}^{\varepsilon}_R \rangle_{0}\\
=&:I_1+\cdots+I_9.
\end{split}
\end{equation}

\emph{Estimate of $I_1$.} By integrating by parts, we have $I_1=0.$

\emph{Estimate of $I_2$.} Using the commutator, we have
\begin{equation}\label{equ28}
\begin{split}
I_2=-\frac{1}{\varepsilon}\langle {\textbf{u}}\cdot\overline{\nabla} \partial_x^{\alpha}\textbf{u}^{\varepsilon}_R, \partial_x^{\alpha}\textbf{u}^{\varepsilon}_R\rangle_0 -\frac{1}{\varepsilon}\langle[\partial_x^{\alpha},{\textbf{u}}] \cdot\overline{\nabla}\textbf{u}^{\varepsilon}_R, \partial_x^{\alpha}\textbf{u}^{\varepsilon}_R\rangle_0=I_{21}+I_{22}.
\end{split}
\end{equation}
Since
\begin{equation*}
\begin{split}
-\langle {\textbf{u}}\cdot\overline{\nabla}\partial_x^{\alpha} \textbf{u}^{\varepsilon}_R, \partial_x^{\alpha}\textbf{u}^{\varepsilon}_R\rangle_0 =&\langle\overline{\nabla} \cdot{\textbf{u}}\textbf{u}^{\varepsilon}_R, \textbf{u}^{\varepsilon}_R\rangle_0 +\langle{\textbf{u}}\otimes\partial_x^{\alpha}\textbf{u}^{\varepsilon}_R, \overline{\nabla}\partial_x^{\alpha}\textbf{u}^{\varepsilon}_R\rangle_0\\
=&\langle\overline{\nabla}\cdot{\textbf{u}} \partial_x^{\alpha}\textbf{u}^{\varepsilon}_R, \partial_x^{\alpha}\textbf{u}^{\varepsilon}_R\rangle_0 +\langle{\textbf{u}}\cdot\overline{\nabla} \partial_x^{\alpha}\textbf{u}^{\varepsilon}_R, \partial_x^{\alpha}\textbf{u}^{\varepsilon}_R\rangle_0,
\end{split}
\end{equation*}
we have
\begin{equation*}
\begin{split}
-\langle {\textbf{u}}\cdot\overline{\nabla}\partial_x^{\alpha} \textbf{u}^{\varepsilon}_R, \partial_x^{\alpha}\textbf{u}^{\varepsilon}_R\rangle_0 =\frac12\langle\overline{\nabla}\cdot{\textbf{u}} \partial_x^{\alpha}\textbf{u}^{\varepsilon}_R, \partial_x^{\alpha}\textbf{u}^{\varepsilon}_R\rangle_0.
\end{split}
\end{equation*}
Recalling \eqref{e3}, we have by Sobolev embedding,
\begin{equation}\label{equ29}
\begin{split}
|I_{21}| \leq & C\|\overline{\nabla}(\widetilde{\textbf{u}} +\varepsilon\textbf{u}^{\varepsilon}_R)\|_{L^{\infty}} \|\partial_x^{\alpha}\textbf{u}^{\varepsilon}_R\|_{L^2}^2\\
\leq & C(1+\varepsilon\|\textbf{u}^{\varepsilon}_R\|_{H^3}) \|\partial_x^{\alpha}\textbf{u}^{\varepsilon}_R\|_{L^2}^2.
\end{split}
\end{equation}
On the other hand, recalling \eqref{e3}, by commutator estimates in \cite{KP88}, we have
\begin{equation*}
\begin{split}
\|[\partial_x^{\alpha},{\textbf{u}}] \cdot\overline{\nabla}\textbf{u}^{\varepsilon}_R\|_{L^2} \leq & C(\|\overline{\nabla}\textbf{u}^{\varepsilon}_R\|_{\dot H^{k-1}} \|{\nabla}\textbf{u}\|_{L^{\infty}} +\|\textbf{u}\|_{\dot H^{k}} \|\overline{\nabla}\textbf{u}^{\varepsilon}_R\|_{L^{\infty}})\\
\leq & C\varepsilon(\|\textbf{u}^{\varepsilon}_R\|_{\dot H^{k}}(1+\varepsilon\|{\nabla}\textbf{u}^{\varepsilon}_R\|_{L^{\infty}}) + \|\nabla\textbf{u}^{\varepsilon}_R\|_{L^{\infty}} (1+\varepsilon\|\textbf{u}^{\varepsilon}_R\|_{H^k}))\\
\leq & C\varepsilon(1+\varepsilon\|\textbf{u}^{\varepsilon}_R\|_{H^{s'}}) \|\textbf{u}^{\varepsilon}_R\|_{H^{s'}},
\end{split}
\end{equation*}
where $k\leq s'$ and $s'\geq 3$. For $I_{22}$ in \eqref{equ28}, we then have
\begin{equation}\label{e15}
\begin{split}
|I_{22}|\leq C(1+\|\textbf{u}^{\varepsilon}_R\|_{H^{s'}}) \|\textbf{u}^{\varepsilon}_R\|_{H^{s'}}^2.
\end{split}
\end{equation}
Combining \eqref{equ29} and \eqref{e15}, we have
\begin{equation}\label{e16}
\begin{split}
|I_{2}|\leq C(1+\|\textbf{u}^{\varepsilon}_R\|_{H^{s'}}) \|\textbf{u}^{\varepsilon}_R\|_{H^{s'}}^2.
\end{split}
\end{equation}

\emph{Estimate of $I_3$.} By multiplicative estimates in \cite{KP88} and \eqref{assump}, we have
\begin{equation*}
\begin{split}
\|\partial_x^{\alpha} (\textbf{u}^{\varepsilon}_R \cdot\overline{\nabla}\widetilde{\textbf{u}})\|_{L^2}\leq & C(\|\textbf{u}^{\varepsilon}_R\|_{\dot H^k} \|\overline{\nabla}\widetilde{\textbf{u}}\|_{L^{\infty}} +\|\textbf{u}^{\varepsilon}_R\|_{L^{\infty}} \|\overline{\nabla}\widetilde{\textbf{u}}\|_{\dot H^k})\\
\leq & C\|\textbf{u}^{\varepsilon}_R\|_{H^{s'}},
\end{split}
\end{equation*}
which follows that
\begin{equation*}
\begin{split}
|I_{3}| \leq  C\|\textbf{u}^{\varepsilon}_R\|_{H^{s'}}^2.
\end{split}
\end{equation*}

\emph{Estimate of $I_4$.} Similarly, from \eqref{assump}, we have
\begin{equation*}
\begin{split}
|I_{4}| \leq  C\|\textbf{u}^{\varepsilon}_R\|_{H^{k}}^2+C\varepsilon^2.
\end{split}
\end{equation*}

\emph{Estimate of $I_5$.} It is easy to see $I_5=0$.

\emph{Estimate of $I_7$.} Recalling $I_7$ in \eqref{e4}, we have
\begin{equation}\label{e14}
\begin{split}
I_7=& -T_i\langle\partial_x^{\alpha}(\frac{\overline{\nabla} n^{\varepsilon}_R}{\varepsilon n}), \partial_x^{\alpha}\textbf{u}^{\varepsilon}_R\rangle_0\\
=& -T_i\langle\frac{1}{\varepsilon n}\partial_x^{\alpha}\overline{\nabla}n^{\varepsilon}_R, \partial_x^{\alpha}\textbf{u}^{\varepsilon}_R\rangle_0 -T_i\langle[\partial^{\alpha},\frac{1}{\varepsilon n}]\overline{\nabla}n^{\varepsilon}_R, \partial_x^{\alpha}\textbf{u}^{\varepsilon}_R\rangle_0\\
=&:I_{71}+I_{72}.
\end{split}
\end{equation}
We note that when $k=0$, there is no such commutator term. From \eqref{assumption1}, we have
\begin{equation}\label{e12}
\begin{split}
\|\nabla(\frac{1}{\varepsilon n})\|_{L^{\infty}}\leq & C\|\nabla\widetilde n\|_{L^{\infty}}+C\varepsilon\|\nabla n^{\varepsilon}_R\|_{L^{\infty}}\\
\leq & C+C\varepsilon\|n^{\varepsilon}_R\|_{H^3},
\end{split}
\end{equation}
where we have used the expression \eqref{e3}. Using \eqref{e3}, \eqref{assumption} and \eqref{assumption1}, we have
\begin{equation}\label{e13}
\begin{split}
\|\frac{1}{\varepsilon n}\|_{\dot H^k}\leq C+\varepsilon C_1\|n^{\varepsilon}_R\|_{H^k},
\end{split}
\end{equation}
for $k\geq 1$, where $C_1=C_1(\varepsilon\tilde C)$ is some constant depending on $\varepsilon\tilde C$. Since $s'\geq 3$ and $k\leq s'$, we have
\begin{equation*}
\begin{split}
\|[\partial_x^{\alpha},\frac{1}{\varepsilon n}] \overline{\nabla}n^{\varepsilon}_R\|_{L^2} \leq & C(\|\overline{\nabla}n^{\varepsilon}_R\|_{\dot H^{k-1}} \|\nabla(\frac{1}{\varepsilon n})\|_{L^{\infty}} +\|\overline{\nabla}n^{\varepsilon}_R\|_{L^{\infty}} \|\frac{1}{\varepsilon n}\|_{\dot H^k})\\
\leq & C(1+\varepsilon C_1\|n^{\varepsilon}_R\|_{H^{s'}}) \|n^{\varepsilon}_R\|_{H^{s'}},
\end{split}
\end{equation*}
where we have used \eqref{e12} and \eqref{e13}. Therefore, we have
\begin{equation*}
\begin{split}
|I_{72}|\leq C(1+\varepsilon C_1\|n^{\varepsilon}_R\|_{H^{s'}}) (\|n^{\varepsilon}_R\|_{H^{s'}}^2 +\|\textbf{u}^{\varepsilon}_R\|_{H^{s'}}^2).
\end{split}
\end{equation*}

\emph{Estimate of $I_8$.} Recall
\begin{equation*}
\begin{split}
I_8=T_i\sum_{|\alpha|=k}\langle\partial_{x}^{\alpha} (\frac{\textbf{p}}{n}n^{\varepsilon}_R), \partial_{x}^{\alpha}\textbf{u}^{\varepsilon}_R\rangle_0.
\end{split}
\end{equation*}
From \eqref{assumption} and \eqref{assumption1}, we have
\begin{equation}\label{e5}
\begin{split}
\|\partial_{x}^{\alpha}(\frac{\textbf{p}}{n})\|\leq C\varepsilon +CC_1\varepsilon^2\|n^{\varepsilon}_R\|_{H^k},
\end{split}
\end{equation}
where $C_1=C_1(\varepsilon\tilde C)$ is a constant depending on $\varepsilon\tilde C$.
By multiplicative estimate, we then have
\begin{equation*}
\begin{split}
\|\partial_{x}^{\alpha}(\frac{\textbf{p}}{n}n^{\varepsilon}_R)\|\leq & C(\|\frac{\textbf{p}}{n}\|_{L^{\infty}}\|n^{\varepsilon}_R\|_{\dot H^k} +\|n^{\varepsilon}_R\|_{L^{\infty}}\|\frac{\textbf{p}}{n}\|_{\dot H^k})\\
\leq & C\|n^{\varepsilon}_R\|_{\dot H^k} +C\varepsilon\|n^{\varepsilon}_R\|_{H^{3}} +CC_1\varepsilon^2\|n^{\varepsilon}_R\|_{H^{3}}\|n^{\varepsilon}_R\|_{H^k}\\
\leq & C(1+\varepsilon C_1)\|n^{\varepsilon}_R\|_{H^k} +C\varepsilon\|n^{\varepsilon}_R\|_{H^{3}},
\end{split}
\end{equation*}
where we have used $\|n^{\varepsilon}_R\|_{H^{3}}\leq \tilde C$ by \eqref{assumption} and $s'\geq 3$. Therefore, from \eqref{e5}, we have
\begin{equation*}
\begin{split}
|I_8|\leq & C(1+\varepsilon C_1)\|n^{\varepsilon}_R\|_{H^k}\|\textbf{u}^{\varepsilon}_R\|_{H^k} +C\varepsilon\|n^{\varepsilon}_R\|_{H^3}\|\textbf{u}^{\varepsilon}_R\|_{H^k}\\
\leq & C(1+\varepsilon C_1)\|n^{\varepsilon}_R\|_{H^{s'}}^2+C\|\textbf{u}^{\varepsilon}_R\|_{H^k}^2,
\end{split}
\end{equation*}
where $C_1=C_1(\varepsilon\tilde C)$ is a constant depending on $\varepsilon\tilde C$.

\emph{Estimate of $I_9$.} The estimate of $I_9$ in \eqref{e4} is similar to $I_8$. Recalling \eqref{assump}, we have
\begin{equation*}
\begin{split}
|I_9|\leq C\varepsilon^2+CC_1\varepsilon^2\|n^{\varepsilon}_R\|_{H^k}^2 +\|\textbf{u}^{\varepsilon}_R\|_{H^k}^2.
\end{split}
\end{equation*}

Summarizing, we have
\begin{equation}\label{equ57}
\begin{split}
\sum_{i=1}^5|I_{i}|+|I_8|+|I_9|
\leq  CC_1(1+\varepsilon\|(\textbf{u}^{\varepsilon}_R, n^{\varepsilon}_R)\|_{H^{s'}}) (1+\|(\textbf{u}^{\varepsilon}_R,n^{\varepsilon}_R)\|_{H^{s'}}^2).
\end{split}
\end{equation}

\emph{Estimate of $I_6$.} The estimate for the $I_6$ is not straightforward, and is very delicate since we need to use the structure of the remainder system \eqref{rem} very carefully. By integration by parts, we have
\begin{equation*}
\begin{split}
I_{6}=\frac{1}{\varepsilon}\langle\partial_x^{\alpha}\phi^{\varepsilon}_R, \overline{\nabla}\cdot\partial_x^{\alpha}\textbf{u}^{\varepsilon}_R\rangle_{0},
\end{split}
\end{equation*}
where $|\alpha|=k$. Taking $\partial_x^{\alpha}$ of \eqref{re-1} with $|\alpha|=k$, we have
\begin{equation}\label{equ62}
\begin{split}
\frac{1}{\varepsilon}\overline{\nabla}\cdot\partial_x^{\alpha} \textbf{u}^{\varepsilon}_R=& \frac{1}{n}\{\frac{V\textbf{e}_1-\textbf{u}}{\varepsilon} \cdot\overline{\nabla} \partial_x^{\alpha}n^{\varepsilon}_R -\partial_t\partial_x^{\alpha}n^{\varepsilon}_R -[\partial_x^{\alpha},\frac{\textbf{u}}{\varepsilon}] \cdot\overline{\nabla}n^{\varepsilon}_R\\
&-[\partial_x^{\alpha}, \frac{n}{\varepsilon}]\overline{\nabla}\cdot\textbf{u}^{\varepsilon}_R -\partial_x^{\alpha}(n^{\varepsilon}_R\nabla\cdot\widetilde{\textbf{u}} +\textbf{u}^{\varepsilon}_R\cdot\nabla\widetilde{n}+\varepsilon R_n)\}.
\end{split}
\end{equation}
Accordingly, $I_6$ is divided into five parts
\begin{equation}\label{equ31}
\begin{split}
I_6=&\langle\partial_x^{\alpha}\phi^{\varepsilon}_R,\frac{V\textbf{e}_1-\textbf{u}}{\varepsilon n}\cdot\overline{\nabla}\partial_x^{\alpha}n^{\varepsilon}_R\rangle_{0} -\langle\partial_x^{\alpha}\phi^{\varepsilon}_R,\frac{1}{n}\partial_t\partial_x^{\alpha}n^{\varepsilon}_R \rangle_{0}\\
&-\langle\partial_x^{\alpha}\phi^{\varepsilon}_R, \frac{1}{\varepsilon n}[\partial_x^{\alpha},{\textbf{u}}] \cdot\overline{\nabla}n^{\varepsilon}_R\rangle_{0} -\langle\partial_x^{\alpha}\phi^{\varepsilon}_R,\frac{1}{\varepsilon n}[\partial_x^{\alpha},{n}]\overline{\nabla} \cdot\textbf{u}^{\varepsilon}_R\rangle_{0}\\
& -\langle\partial_x^{\alpha}\phi^{\varepsilon}_R, \frac{1}{n}\partial_x^{\alpha} (n^{\varepsilon}_R\nabla\cdot\widetilde{\textbf{u}} +\textbf{u}^{\varepsilon}_R\cdot\nabla\widetilde{n}+\varepsilon R_n)\rangle_{0}\\
=&:I_{61}+\cdots+I_{65}.
\end{split}
\end{equation}

We first estimate $I_{63}-I_{65}$.

On the other hand, recalling \eqref{e3}, by commutator estimates in \cite{KP88}, we have
\begin{equation}\label{e24}
\begin{split}
\|[\partial_x^{\alpha},{\textbf{u}}] \cdot\overline{\nabla}n^{\varepsilon}_R\|_{L^2} \leq & C(\|\overline{\nabla}n^{\varepsilon}_R\|_{\dot H^{k-1}} \|{\nabla}\textbf{u}\|_{L^{\infty}} +\|\textbf{u}\|_{\dot H^{k}} \|\overline{\nabla}n^{\varepsilon}_R\|_{L^{\infty}})\\
\leq & C\varepsilon(\|n^{\varepsilon}_R\|_{\dot H^{k}}(1+\varepsilon\|{\nabla}\textbf{u}^{\varepsilon}_R\|_{L^{\infty}}) + \|\nabla\textbf{u}^{\varepsilon}_R\|_{L^{\infty}} (1+\varepsilon\|n^{\varepsilon}_R\|_{H^k}))\\
\leq & C\varepsilon(1+\varepsilon\|\textbf{u}^{\varepsilon}_R\|_{H^{s'}}) (\|\textbf{u}^{\varepsilon}_R\|_{H^{s'}} +\|n^{\varepsilon}_R\|_{H^{s'}}),
\end{split}
\end{equation}
where $k\leq s'$ and $s'\geq 3$. Therefore, from \eqref{assumption1}, we have
\begin{equation}\label{e17}
\begin{split}
|I_{63}|\leq & C\|\partial_x^{\alpha}\phi^{\varepsilon}_R\|_{L^2} \|[\partial_x^{\alpha},{\textbf{u}}] \cdot\overline{\nabla}n^{\varepsilon}_R\|_{L^2}\\
\leq & C(1+\varepsilon^2\|\textbf{u}^{\varepsilon}_R\|_{H^{s'}}^2) \|(\textbf{u}^{\varepsilon}_R,n^{\varepsilon}_R)\|_{H^{s'}}^2 +C\|\phi^{\varepsilon}_R\|_{H^{s'}}^2.
\end{split}
\end{equation}
Similarly, we have
\begin{equation}\label{e27}
\begin{split}
\|[\partial_x^{\alpha},n] \overline{\nabla}\cdot\textbf{u}^{\varepsilon}_R\|_{L^2} \leq & C\varepsilon(1+\varepsilon\|\textbf{u}^{\varepsilon}_R\|_{H^{s'}}) (\|\textbf{u}^{\varepsilon}_R\|_{H^{s'}} +\|n^{\varepsilon}_R\|_{H^{s'}}),
\end{split}
\end{equation}
from which it follows that
\begin{equation}\label{e18}
\begin{split}
|I_{64}|\leq & C(1+\varepsilon^2\|\textbf{u}^{\varepsilon}_R\|_{H^{s'}}^2) \|(\textbf{u}^{\varepsilon}_R,n^{\varepsilon}_R)\|_{H^{s'}}^2 +C\|\phi^{\varepsilon}_R\|_{H^{s'}}^2.
\end{split}
\end{equation}
By multiplicative estimates, we have
\begin{equation}\label{e28}
\begin{split}
\|\partial_x^{\alpha} (n^{\varepsilon}_R\nabla\cdot\widetilde{\textbf{u}} +\textbf{u}^{\varepsilon}_R\cdot\nabla\widetilde{n}+\varepsilon R_n)\| \leq & C\|(\textbf{u}^{\varepsilon}_R,n^{\varepsilon}_R)\|_{H^{s'}} +C\varepsilon^2.
\end{split}
\end{equation}
It follows that
\begin{equation}\label{e19}
\begin{split}
|I_{65}|\leq & C\varepsilon(1+\varepsilon^2\|\textbf{u}^{\varepsilon}_R\|_{H^{s'}}^2) \|(\textbf{u}^{\varepsilon}_R,n^{\varepsilon}_R)\|_{H^{s'}}^2 +C\varepsilon\|\phi^{\varepsilon}_R\|_{H^{s'}}^2.
\end{split}
\end{equation}
Summarizing \eqref{e17}, \eqref{e18} and \eqref{e19}, we have
\begin{equation*}
\begin{split}
|I_{63}+I_{64}+I_{65}|\leq & C\|\phi^{\varepsilon}_R\|_{H^{s'}}^2 +C(1+\varepsilon^2\|\textbf{u}^{\varepsilon}_R\|_{H^{s'}}^2) \|(n^{\varepsilon}_R,\textbf{u}^{\varepsilon}_R)\|_{H^{s'}}^2.
\end{split}
\end{equation*}

For clarity, we estimate $I_{61}$ and $I_{62}$ in Lemma \ref{L4} and \ref{L5} respectively.

We complete the proof of Lemma \ref{L8} by the following Lemma \ref{L4} and \ref{L5}.
\end{proof}

\begin{lemma}\label{L4}
Let $(n^{\varepsilon}_R,\textbf{u}^{\varepsilon}_R,\phi^{\varepsilon}_R)$ be a solution to \eqref{re} and $0\leq k\leq s'$ be an integer, then 
\begin{equation}\label{e84}
\begin{split}
\left|I_{61}\right|\leq C(1+C_1(\sqrt{\varepsilon}\|\phi^{\varepsilon}_R\|_{H^{s'}}) +\varepsilon^2\|(n^{\varepsilon}_R,\textbf{u}^{\varepsilon}_R)\|_{H^{4}}^2) (1+\|\phi^{\varepsilon}_R\|_{H^{s'}}^2 +\varepsilon\|\overline{\nabla}\phi^{\varepsilon}_R\|_{H^{s'}}^2),
\end{split}
\end{equation}
where $I_{61}$ is given in \eqref{equ31} and $|\alpha|=k\leq s'$.
\end{lemma}
\begin{proof}
Let $\alpha$ be a multi-index such that $|\alpha|=k$. Taking $\partial_x^{\alpha}$ of \eqref{re-3}, we have
\begin{equation}\label{equ39}
\begin{split}
\partial_x^{\alpha}n^{\varepsilon}_R =\partial_x^{\alpha}\phi^{\varepsilon}_R -\varepsilon\overline{\Delta}\partial_x^{\alpha}\phi^{\varepsilon}_R +\varepsilon\partial_x^{\alpha}(\phi^{(1)}\phi^{\varepsilon}_R) +\varepsilon^{3/2}\partial_x^{\alpha}R_{\phi}.
\end{split}
\end{equation}
Then $I_{61}$ in \eqref{equ31} is divided into
\begin{equation*}\label{e85}
\begin{split}
I_{61}=&\langle\partial_x^{\alpha}\phi^{\varepsilon}_R, \frac{V\textbf{e}_1-\textbf{u}}{\varepsilon n}\cdot\overline{\nabla}\partial_x^{\alpha}\phi^{\varepsilon}_R\rangle_0 -\varepsilon\langle\partial_x^{\alpha}\phi^{\varepsilon}_R, \frac{V\textbf{e}_1-\textbf{u}}{\varepsilon n}\cdot\overline{\nabla}\overline{\Delta}\partial_x^{\alpha}\phi^{\varepsilon}_R\rangle_0\\
&+\varepsilon\langle\partial_x^{\alpha}\phi^{\varepsilon}_R, \frac{V\textbf{e}_1-\textbf{u}}{\varepsilon n}\cdot\overline{\nabla} \partial_x^{\alpha}(\phi^{(1)}\phi^{\varepsilon}_R)\rangle_0 +\varepsilon^{3/2}\langle\partial_x^{\alpha}\phi^{\varepsilon}_R, \frac{V\textbf{e}_1-\textbf{u}}{\varepsilon n}\cdot\overline{\nabla}\partial_x^{\alpha}R_{\phi}\rangle_0\\
=&:I_{611}+I_{612}+I_{613}+I_{614}.
\end{split}
\end{equation*}

\emph{Estimate of $I_{611}$.} By integrating by parts, we have
\begin{equation}\label{equ32}
\begin{split}
I_{611}=&-\frac12\langle\partial_x^{\alpha}\phi^{\varepsilon}_R, \overline{\nabla}\cdot(\frac{V\textbf{e}_1-\textbf{u}}{\varepsilon n})\partial_x^{\alpha}\phi^{\varepsilon}_R\rangle_0.
\end{split}
\end{equation}
By direct computation, we have
\begin{equation*}
\begin{split}
\overline{\nabla}\cdot(\frac{V\textbf{e}_1-\textbf{u}}{\varepsilon n})=&-\frac{(V\textbf{e}_1-\textbf{u})}{n^2} \cdot(\overline{\nabla}\widetilde n+\varepsilon\overline{\nabla}n^{\varepsilon}_R) -\frac{1}{n}(\overline{\nabla}\cdot\widetilde{\textbf{u}} +\varepsilon\overline{\nabla}\cdot\textbf{u}^{\varepsilon}_R).
\end{split}
\end{equation*}
By assumption \eqref{assumption}, \eqref{assumption1} and Sobolev embedding, we know that
\begin{equation}\label{equ33}
\begin{split}
\|\overline{\nabla}\cdot(\frac{V\textbf{e}_1-\textbf{u}}{\varepsilon n})\|_{L^{\infty}}\leq & C(1+\varepsilon\|\overline{\nabla} n^{\varepsilon}_R\|_{L^{\infty}} +\varepsilon\|\overline{\nabla} \textbf{u}^{\varepsilon}_R\|_{L^{\infty}})\\
\leq & C(1+\varepsilon\|(n^{\varepsilon}_R, \textbf{u}^{\varepsilon}_R)\|_{H^3}).
\end{split}
\end{equation}
Using H\"older inequality in \eqref{equ32}, we obtain
\begin{equation}\label{e82}
\begin{split}
|I_{611}|\leq C(1+\varepsilon^2\|(n^{\varepsilon}_R, \textbf{u}^{\varepsilon}_R)\|_{H^{3}}^2) \|\partial_x^{\alpha}\phi^{\varepsilon}_R\|^2.
\end{split}
\end{equation}

\emph{Estimate of $I_{612}$.} By integration by parts twice, we have
\begin{equation}\label{equ34}
\begin{split}
I_{612}=& \varepsilon\langle\overline{\nabla}\partial_x^{\alpha}\phi^{\varepsilon}_R, \frac{V\textbf{e}_1-\textbf{u}}{\varepsilon n}\overline{\Delta}\partial_x^{\alpha}\phi^{\varepsilon}_R\rangle_0 +\varepsilon\langle\partial_x^{\alpha}\phi^{\varepsilon}_R, \overline{\nabla}\cdot(\frac{V\textbf{e}_1-\textbf{u}}{\varepsilon n})\overline{\Delta}\partial_x^{\alpha}\phi^{\varepsilon}_R\rangle_0\\
=&-\varepsilon\langle\overline{\nabla} \partial_x^{\alpha}\phi^{\varepsilon}_R, \overline{\nabla}(\frac{V\textbf{e}_1-\textbf{u}}{\varepsilon n})\overline{\nabla}\partial_x^{\alpha}\phi^{\varepsilon}_R\rangle_0 -\frac{\varepsilon}{2}\langle\overline{\nabla} \partial_x^{\alpha}\phi^{\varepsilon}_R, \overline{\nabla}\cdot(\frac{V\textbf{e}_1-\textbf{u}}{\varepsilon n})\overline{\nabla}\partial_x^{\alpha}\phi^{\varepsilon}_R\rangle_0\\
&-\varepsilon\langle\partial_x^{\alpha}\phi^{\varepsilon}_R, \overline{\nabla}\overline{\nabla}\cdot(\frac{V\textbf{e}_1-\textbf{u}}{\varepsilon n})\overline{\nabla}\partial_x^{\alpha}\phi^{\varepsilon}_R\rangle_0\\
=&:I_{6121}+I_{6122}+I_{6123}.
\end{split}
\end{equation}
From \eqref{equ33}, the first two terms on the RHS can be bounded by
\begin{equation}\label{equ35}
\begin{split}
I_{6121},I_{6122}\leq & C\varepsilon(1+\varepsilon^2\|(n^{\varepsilon}_R, \textbf{u}^{\varepsilon}_R)\|_{H^{3}}^2) \|\overline{\nabla}\partial_x^{\alpha}\phi^{\varepsilon}_R\|^2.
\end{split}
\end{equation}
Similar to \eqref{equ33}, we have
\begin{equation}\label{equ69}
\begin{split}
\left|\overline{\nabla}\overline{\nabla}\cdot(\frac{V\textbf{e}_1-\textbf{u}}{\varepsilon n})\right|\leq & C(1+\varepsilon|\overline{\nabla}^2n^{\varepsilon}_R| +\varepsilon|\overline{\nabla}\overline{\nabla}\cdot\textbf{u}^{\varepsilon}_R| +\varepsilon^2|\overline{\nabla} n^{\varepsilon}_R|\\
&+\varepsilon^2|\overline{\nabla}\cdot\textbf{u}^{\varepsilon}_R| +\varepsilon^3|\overline{\nabla} n^{\varepsilon}_R|^2 +\varepsilon^3|\overline{\nabla}\cdot\textbf{u}^{\varepsilon}_R|^2),
\end{split}
\end{equation}
where we have used \eqref{assump}. Therefore, we have
\begin{equation}\label{equ36}
\begin{split}
|I_{6123}|\leq & C\varepsilon\|\partial_x^{\alpha}\phi^{\varepsilon}_R\| \|\overline{\nabla}\partial_x^{\alpha}\phi^{\varepsilon}_R\| +C\varepsilon^2\|\partial_x^{\alpha}\phi^{\varepsilon}_R\|_{L^2} (\|\overline{\nabla}^{2}n^{\varepsilon}_R\|_{L^{\infty}} +\|\overline{\nabla}^{2}\textbf{u}^{\varepsilon}_R\|_{L^{\infty}}) \|\overline{\nabla}\partial_x^{\alpha}\phi^{\varepsilon}_R\|\\
&+C\varepsilon\|\partial_x^{\alpha}\phi^{\varepsilon}_R\|_{L^2} (1+\varepsilon^3\|\overline{\nabla}n^{\varepsilon}_R\|_{L^{\infty}}^2 +\varepsilon^3\|\overline{\nabla} \cdot\textbf{u}^{\varepsilon}_R\|_{L^{\infty}}^2) \|\overline{\nabla}\partial_x^{\alpha}\phi^{\varepsilon}_R\|\\
\leq & C\{\|\partial_x^{\alpha}\phi^{\varepsilon}_R\|^2 +\varepsilon \|\overline{\nabla}\partial_x^{\alpha}\phi^{\varepsilon}_R\|^2\}\\
&+C\{\|\partial_x^{\alpha}\phi^{\varepsilon}_R\|_{L^2}^2 +\varepsilon^2\|(n^{\varepsilon}_R,\textbf{u}^{\varepsilon}_R)\|_{H^4}^2 (\varepsilon\|\overline{\nabla} \partial_x^{\alpha}\phi^{\varepsilon}_R\|_{L^2}^2)\}\\
& +C\{(1+\varepsilon^2\|(n^{\varepsilon}_R, \textbf{u}^{\varepsilon}_R)\|_{H^{3}}^2) (\|\partial_x^{\alpha}\phi^{\varepsilon}_R\|_{L^2}^2 +\varepsilon\|\overline{\nabla} \partial_x^{\alpha}\phi^{\varepsilon}_R\|^2)\}\\
\leq & C(1+\varepsilon^2\|(n^{\varepsilon}_R, \textbf{u}^{\varepsilon}_R)\|_{H^4}^2) (\|\partial_x^{\alpha}\phi^{\varepsilon}_R\|^2 +\varepsilon\|\overline{\nabla}\partial_x^{\alpha}\phi^{\varepsilon}_R\|^2).
\end{split}
\end{equation}
Therefore, combining \eqref{equ35}, \eqref{equ36} and \eqref{equ34}, we obtain
\begin{equation}\label{equ70}
\begin{split}
|I_{612}|\leq & C(1+\varepsilon^2\|(n^{\varepsilon}_R, \textbf{u}^{\varepsilon}_R)\|_{H^4}^2) (\|\partial_x^{\alpha}\phi^{\varepsilon}_R\|^2 +\varepsilon\|\overline{\nabla}\partial_x^{\alpha}\phi^{\varepsilon}_R\|^2).
\end{split}
\end{equation}

\emph{Estimate of $I_{613}$.} By integrating by parts, we have
\begin{equation}\label{equ37}
\begin{split}
I_{613}=&\langle\partial_x^{\alpha}\phi^{\varepsilon}_R, (\frac{(V\textbf{e}_1-\textbf{u})\phi^{(1)}}{n}) \cdot\overline{\nabla}\partial_x^{\alpha}\phi^{\varepsilon}_R\rangle_0 +\langle\partial_x^{\alpha}\phi^{\varepsilon}_R, \frac{(V\textbf{e}_1-\textbf{u})}{n}[\overline{\nabla}\partial_x^{\alpha}, \phi^{(1)}]\phi^{\varepsilon}_R\rangle_0\\
=&:I_{6131}+I_{6132}.
\end{split}
\end{equation}
By integrating by parts, we know
\begin{equation*}
\begin{split}
I_{6131}=&-\frac12\langle\partial_x^{\alpha}\phi^{\varepsilon}_R, \overline{\nabla}\cdot(\frac{(V\textbf{e}_1-\textbf{u})\phi^{(1)}}{n}) \partial_x^{\alpha}\phi^{\varepsilon}_R\rangle_0.
\end{split}
\end{equation*}
Similar to \eqref{equ33}, from \eqref{assumption}, \eqref{assumption1} and Sobolev embedding, we know that
\begin{equation}\label{equ38}
\begin{split}
\|\overline{\nabla}\cdot(\frac{(V\textbf{e}_1-\textbf{u})\phi^{(1)}}{n}) \|_{L^{\infty}}\leq C(1+\varepsilon\|(n^{\varepsilon}_R, \textbf{u}^{\varepsilon}_R\|_{H^3}).
\end{split}
\end{equation}
Hence, using H\"older inequality, we obtain
\begin{equation*}
\begin{split}
|I_{6131}| \leq & C(1+\varepsilon^2\|(n^{\varepsilon}_R, \textbf{u}^{\varepsilon}_R)\|_{H^{3}}^2) \|\partial_x^{\alpha}\phi^{\varepsilon}_R\|^2.
\end{split}
\end{equation*}
On the other hand, by commutator estimate, we know
\begin{equation}\label{e31}
\begin{split}
\|[\overline{\nabla}\partial_x^{\alpha}, \phi^{(1)}]\phi^{\varepsilon}_R\|_{L^2}\leq & C(\|\phi^{\varepsilon}_R\|_{H^k}\|\phi^{(1)}\|_{L^{\infty}} +\|\phi^{(1)}\|_{H^{k}}\|\phi^{\varepsilon}_R\|_{L^{\infty}}\\
\leq & C\|\phi^{\varepsilon}_R\|_{H^{s'}}.
\end{split}
\end{equation}
It follows that
\begin{equation*}
\begin{split}
|I_{6132}| \leq & C\|\phi^{\varepsilon}_R\|_{H^{s'}}^2.
\end{split}
\end{equation*}
Therefore, we have
\begin{equation}\label{equ72}
\begin{split}
|I_{613}| \leq & C(1+\varepsilon^2\|n^{\varepsilon}_R\|_{H^3}^2 +\varepsilon^2\|\textbf{u}^{\varepsilon}_R\|_{H^{3}}^2) \|\phi^{\varepsilon}_R\|_{H^k}^2.
\end{split}
\end{equation}

\emph{Estimate of $I_{614}$.} By H\"older inequality and \eqref{e81-1} in Corollary \ref{cor-a}, we have
\begin{equation}\label{e83}
\begin{split}
|I_{614}|\leq & \varepsilon^{1/2}\langle\partial_x^{\alpha}\phi^{\varepsilon}_R, \frac{V\textbf{e}_1-\textbf{u}}{n}\cdot\overline{\nabla} \partial_x^{\alpha}R_{\phi}\rangle_0\\
\leq & C\|\partial_x^{\alpha}\phi^{\varepsilon}_R\|^2 + C_1(\sqrt{\varepsilon}\|\phi^{\varepsilon}_R\|_{H^{s'}}) (1+\varepsilon\|\overline{\nabla}\phi^{\varepsilon}_R\|_{H^{s'}}^2),
\end{split}
\end{equation}
where $|\alpha|=k\leq s'$. From \eqref{e85}, adding \eqref{e82}, \eqref{equ70}, \eqref{equ72} and \eqref{e83} together, we obtain \eqref{e84}.
\end{proof}

\begin{lemma}\label{L5}
Let $(n^{\varepsilon}_R,\textbf{u}^{\varepsilon}_R,\phi^{\varepsilon}_R)$ be a solution to \eqref{rem} and $0\leq k\leq s'$, then
\begin{equation}\label{equ52}
\begin{split}
I_{62}
\leq & -\frac12\frac{d}{dt}\int\frac{1+\varepsilon\phi^{(1)} +\varepsilon^2\phi^{\varepsilon}_R}{n} |\partial_x^{\alpha}\phi^{\varepsilon}_R|^2 -\frac12\frac{d}{dt}\int\frac{\varepsilon}{n} |\overline{\nabla}\partial_x^{\alpha}\phi^{\varepsilon}_R|^2\\
&+ CC_1(C_1+\varepsilon\|(\textbf{u}^{\varepsilon}_R, n^{\varepsilon}_R, \phi^{\varepsilon}_R)\|_{H^{s'}}^2)\{1 +\|(\textbf{u}^{\varepsilon}_R, n^{\varepsilon}_R, \phi^{\varepsilon}_R)\|_{H^{s'}}^2\}.
\end{split}
\end{equation}
where $|\alpha|=k\leq s'$ and $I_{62}$ is given in \eqref{equ31}.
\end{lemma}
\begin{proof}
Taking $\partial_x^{\alpha}$ of \eqref{re-3}, we have
\begin{equation*}
\begin{split}
\partial_x^{\alpha}n^{\varepsilon}_R =\partial_x^{\alpha}\phi^{\varepsilon}_R -\varepsilon\overline{\Delta}\partial_x^{\alpha}\phi^{\varepsilon}_R +\varepsilon\partial_x^{\alpha}(\phi^{(1)}\phi^{\varepsilon}_R) +\frac{\varepsilon^2}{2}\partial_x^{\alpha}[(\phi^{\varepsilon}_R)^2] +\varepsilon^{2}\partial_x^{\alpha}\overline{R}_{\phi}.
\end{split}
\end{equation*}
Inserting this into $I_{62}$, we have
\begin{equation}\label{equ51}
\begin{split}
I_{62}=&-\langle\partial_x^{\alpha}\phi^{\varepsilon}_R, \frac{1}{n}\partial_t\partial_x^{\alpha}\phi^{\varepsilon}_R\rangle_0 +\varepsilon\langle\partial_x^{\alpha}\phi^{\varepsilon}_R, \frac{1}{n}\partial_t\overline{\Delta}\partial_x^{\alpha}\phi^{\varepsilon}_R\rangle_0\\
&-\varepsilon\langle\partial_x^{\alpha}\phi^{\varepsilon}_R, \frac{1}{n}\partial_t\partial_x^{\alpha}(\phi^{(1)}\phi^{\varepsilon}_R)\rangle_0 -\frac{\varepsilon^2}{2}\langle\partial_x^{\alpha}\phi^{\varepsilon}_R, \frac{1}{n}\partial_t\partial_x^{\alpha} [(\phi^{\varepsilon}_R)^2]\rangle_0 \\
&-\varepsilon^{2}\langle\partial_x^{\alpha}\phi^{\varepsilon}_R, \frac{1}{n}\partial_t\partial_x^{\alpha}\overline{R}_{\phi}\rangle_0\\
=&:I_{621}+I_{622}+I_{623}+I_{624}+I_{625}.
\end{split}
\end{equation}

\emph{Estimate of $I_{621}$.} By integration by parts in time, we obtain
\begin{equation}\label{equ40}
\begin{split}
I_{621}=&-\frac12\frac{d}{dt}\langle\partial_x^{\alpha}\phi^{\varepsilon}_R, \frac{1}{n}\partial_x^{\alpha}\phi^{\varepsilon}_R\rangle_0 +\frac12\langle\partial_x^{\alpha}\phi^{\varepsilon}_R, \partial_t(\frac{1}{n})\partial_x^{\alpha}\phi^{\varepsilon}_R\rangle_0.
\end{split}
\end{equation}
By direct computation, we have
\begin{equation*}
\begin{split}
\partial_t(\frac{1}{n})=&-\frac{1}{n^2}(\varepsilon\partial_{t}\widetilde n+\varepsilon^2\partial_{t}n^{\varepsilon}_R),
\end{split}
\end{equation*}
which yields that
\begin{equation}\label{equ41}
\begin{split}
\|\partial_t(\frac{1}{n})\|_{L^{\infty}}\leq C(\varepsilon+\varepsilon^2\|\partial_tn^{\varepsilon}_R\|_{L^{\infty}}),
\end{split}
\end{equation}
where we have used \eqref{assumption}, \eqref{assumption1} and Sobolev embedding. Therefore, from \eqref{equ40}, we have
\begin{equation}\label{equ50}
\begin{split}
I_{621}\leq &-\frac12\frac{d}{dt}\langle\partial_x^{\alpha}\phi^{\varepsilon}_R, \frac{1}{n}\partial_x^{\alpha}\phi^{\varepsilon}_R\rangle_0 +C\varepsilon(1+\varepsilon\|\partial_tn^{\varepsilon}_R\|_{L^{\infty}}) \|\partial_x^{\alpha}\phi^{\varepsilon}_R\|^2\\
\leq & -\frac12\frac{d}{dt}\langle\partial_x^{\alpha}\phi^{\varepsilon}_R, \frac{1}{n}\partial_x^{\alpha}\phi^{\varepsilon}_R\rangle_0 +C(1+\varepsilon\|(n^{\varepsilon}_R,\textbf{u}^{\varepsilon}_R)\|_{H^{3}}) \|\phi^{\varepsilon}_R\|_{H^{s'}}^2,
\end{split}
\end{equation}
where in the second inequality, we have used \eqref{equ19} in Lemma \ref{L2}.

\emph{Estimate of $I_{622}$.} By first integrating by parts in space and then in time, we have
\begin{equation}\label{equ42}
\begin{split}
I_{622}= &-\varepsilon\langle\overline{\nabla}\partial_x^{\alpha}\phi^{\varepsilon}_R, \frac{1}{n}\partial_t\overline{\nabla}\partial_x^{\alpha}\phi^{\varepsilon}_R\rangle_0 -\varepsilon\langle\partial_x^{\alpha}\phi^{\varepsilon}_R, \overline{\nabla}(\frac{1}{n})\cdot\partial_t\overline{\nabla} \partial_x^{\alpha} \phi^{\varepsilon}_R\rangle_0\\
=& -\frac12\frac{d}{dt} \langle\overline{\nabla}\partial_x^{\alpha}\phi^{\varepsilon}_R, \frac{\varepsilon}{n}\overline{\nabla}\partial_x^{\alpha} \phi^{\varepsilon}_R\rangle_0\\ &+\frac{\varepsilon}2\langle\overline{\nabla}\partial_x^{\alpha} \phi^{\varepsilon}_R, \partial_t(\frac{1}{n})\overline{\nabla}\partial_x^{\alpha} \phi^{\varepsilon}_R\rangle_0 -\varepsilon\langle\partial_x^{\alpha}\phi^{\varepsilon}_R, \overline{\nabla}(\frac{1}{n})\cdot \partial_t\overline{\nabla}\partial_x^{\alpha} \phi^{\varepsilon}_R\rangle_0.
\end{split}
\end{equation}
From \eqref{equ41}, the second term on the RHS of \eqref{equ42} is bounded by
\begin{equation}\label{equ43}
\begin{split}
\left|\frac12\langle\overline{\nabla}\partial_x^{\alpha}\phi^{\varepsilon}_R, \partial_t(\frac{\varepsilon}{n}) \overline{\nabla}\partial_x^{\alpha}\phi^{\varepsilon}_R\rangle_0\right| \leq C\varepsilon(\varepsilon+\varepsilon\|\varepsilon\partial_tn^{\varepsilon}_R \|_{L^{\infty}})\|\overline{\nabla}\partial_x^{\alpha} \phi^{\varepsilon}_R\|^2.
\end{split}
\end{equation}
By H\"older inequality, the third term on the RHS of \eqref{equ42} is bounded by
\begin{equation}\label{equ44}
\begin{split}
\left|\varepsilon\langle\partial_x^{\alpha}\phi^{\varepsilon}_R, \overline{\nabla}(\frac{1}{n})\cdot\partial_t\overline{\nabla}\partial_x^{\alpha} \phi^{\varepsilon}_R\rangle_0\right|
\leq & C\varepsilon^2\|\partial_x^{\alpha}\phi^{\varepsilon}_R\| \|\partial_t\overline{\nabla}\partial_x^{\alpha} \phi^{\varepsilon}_R\| +C\varepsilon^3\|\nabla n^{\varepsilon}_R\|_{L^{\infty}}\|\partial_x^{\alpha}\phi^{\varepsilon}_R\| \|\partial_t\overline{\nabla}\partial_x^{\alpha}\phi^{\varepsilon}_R\|\\
\leq & C(1+\varepsilon\|\overline{\nabla} n^{\varepsilon}_R\|_{L^{\infty}}^2)\|\partial_x^{\alpha}\phi^{\varepsilon}_R\|^2 +C\varepsilon^3\|\varepsilon\partial_t\overline{\nabla}\partial_x^{\alpha} \phi^{\varepsilon}_R\|^2\\
\leq & C(1+\varepsilon\|\overline{\nabla} n^{\varepsilon}_R\|_{L^{\infty}}^2) \|\partial_x^{\alpha}\phi^{\varepsilon}_R\|^2 +C\varepsilon^2\|\varepsilon\partial_t\overline{\Delta} \phi^{\varepsilon}_R\|_{H^{k-1}}^2,
\end{split}
\end{equation}
where we have used Corollary \ref{rmk5} and
\begin{equation*}
\begin{split}
\|\overline{\nabla}(\frac{1}{n})\|_{L^{\infty}}\leq C\varepsilon(1+\varepsilon\|\nabla n^{\varepsilon}_R\|_{L^{\infty}}).
\end{split}
\end{equation*}
Therefore, $I_{622}$ in \eqref{equ42} can be estimated as
\begin{equation*}
\begin{split}
I_{622}\leq & -\frac12\frac{d}{dt} \langle\overline{\nabla}\partial_x^{\alpha}\phi^{\varepsilon}_R, \frac{\varepsilon}{n}\overline{\nabla}\partial_x^{\alpha}\phi^{\varepsilon}_R\rangle_0\\ &+C(1+\varepsilon\|\nabla n^{\varepsilon}_R\|_{L^{\infty}}^2 +\varepsilon\|\varepsilon\partial_tn^{\varepsilon}_R \|_{L^{\infty}}^2)(\|\partial_x^{\alpha}\phi^{\varepsilon}_R\|^2 +\varepsilon\|\overline{\nabla}\partial_x^{\alpha} \phi^{\varepsilon}_R\|^2)\\
&+C\{1+\|(\textbf{u}^{\varepsilon}_R, n^{\varepsilon}_R, \phi^{\varepsilon}_R)\|_{H^{\delta+1}}^2\},
\end{split}
\end{equation*}
where $\delta=\max\{2,k-1\}$ and we have used \eqref{e22} in Corollary \ref{cor2}. Using Lemma \ref{L1} and \ref{L2}, we have
\begin{equation}\label{equ49}
\begin{split}
I_{622}\leq & -\frac12\frac{d}{dt} \langle\overline{\nabla}\partial_x^{\alpha}\phi^{\varepsilon}_R, \frac{\varepsilon}{n}\overline{\nabla}\partial_x^{\alpha}\phi^{\varepsilon}_R\rangle_0\\ &+C(1+\varepsilon\|(\textbf{u}^{\varepsilon}_R, n^{\varepsilon}_R, \phi^{\varepsilon}_R)\|_{H^{3}}^2)\{1+\|(\textbf{u}^{\varepsilon}_R, n^{\varepsilon}_R, \phi^{\varepsilon}_R)\|_{H^{s'}}^2\}.
\end{split}
\end{equation}

\emph{Estimate of $I_{623}$.} By integration by parts in time, we obtain
\begin{equation}\label{equ45}
\begin{split}
I_{623}=&-\langle\partial_x^{\alpha}\phi^{\varepsilon}_R, \frac{\varepsilon\phi^{(1)}}{n} \partial_t\partial_x^{\alpha}\phi^{\varepsilon}_R\rangle_0 -\langle\partial_x^{\alpha}\phi^{\varepsilon}_R, \frac{\varepsilon}{n}[\partial_x^{\alpha},\phi^{(1)}] \partial_t\phi^{\varepsilon}_R\rangle_0\\
&-\langle\partial_x^{\alpha}\phi^{\varepsilon}_R, \frac{\varepsilon}{n}\partial_x^{\alpha} (\partial_t\phi^{(1)}\phi^{\varepsilon}_R)\rangle_0\\
=&:I_{6231}+I_{6232}+I_{6233},
\end{split}
\end{equation}
where $|\alpha|=k\leq s'$. For the first on the RHS, by integration by parts, we have
\begin{equation}\label{equ46}
\begin{split}
I_{6231}=&-\frac12\frac{d}{dt}\langle\partial_x^{\alpha}\phi^{\varepsilon}_R, \frac{\varepsilon\phi^{(1)}}{n}\partial_x^{\alpha}\phi^{\varepsilon}_R\rangle_0 +\frac12\langle\partial_x^{\alpha}\phi^{\varepsilon}_R, \partial_t(\frac{\varepsilon\phi^{(1)}}{n}) \partial_x^{\alpha}\phi^{\varepsilon}_R\rangle_0.
\end{split}
\end{equation}
By direct computation, we have
\begin{equation*}
\begin{split}
\partial_t(\frac{\varepsilon\phi^{(1)}}{n})=&\frac{\varepsilon\partial_t\phi^{(1)}}{n} -\frac{\varepsilon\phi^{(1)}}{n^2}(\varepsilon\partial_{t}\widetilde n+\varepsilon^2\partial_{t}n^{\varepsilon}_R),
\end{split}
\end{equation*}
which yields that
\begin{equation*}
\begin{split}
\|\partial_t(\frac{\varepsilon\phi^{(1)}}{n})\|_{L^{\infty}}\leq C\varepsilon(1+\varepsilon^2\|\partial_tn^{\varepsilon}_R\|_{L^{\infty}}),
\end{split}
\end{equation*}
where we have used \eqref{assumption}, \eqref{assumption1} and Sobolev embedding. Therefore, from \eqref{equ46}, we have
\begin{equation}\label{equ80}
\begin{split}
I_{6231}\leq &-\frac12\frac{d}{dt}\langle\partial_x^{\alpha}\phi^{\varepsilon}_R, \frac{\varepsilon\phi^{(1)}}{n}\partial_x^{\alpha}\phi^{\varepsilon}_R\rangle_0 +C\varepsilon(1+\varepsilon^2\|\partial_tn^{\varepsilon}_R\|_{L^{\infty}}) \|\partial_x^{\alpha}\phi^{\varepsilon}_R\|^2.
\end{split}
\end{equation}
By commutator estimates, we have
\begin{equation*}
\begin{split}
\|[\partial_x^{\alpha},\phi^{(1)}] \partial_t\phi^{\varepsilon}_R\|_{L^2}\leq & C(\|\partial_t\phi^{\varepsilon}_R\|_{H^{k-1}}\|\phi^{(1)}\|_{L^{\infty}} +\|\partial_t\phi^{\varepsilon}_R\|_{L^{\infty}}\|\phi^{(1)}\|_{H^{k-1}})\\
\leq & C(\|\partial_t\phi^{\varepsilon}_R\|_{H^{k-1}} +\|\partial_t\phi^{\varepsilon}_R\|_{H^2}).
\end{split}
\end{equation*}
Therefore, by using \eqref{e22} in Corollary \ref{cor2}, we have for the second term on the RHS of \eqref{equ45}, we have
\begin{equation}\label{e101}
\begin{split}
I_{6232}\leq & C\|\partial_x^{\alpha}\phi^{\varepsilon}_R\|^2 +C\varepsilon^2\|\partial_t\phi^{\varepsilon}_R\|^2_{H^{\delta}}\\
\leq & C\{1+\|(n^{\varepsilon}_R,\textbf{u}^{\varepsilon}_R, \phi^{\varepsilon}_R\|_{H^{\delta+1}}^2\},
\end{split}
\end{equation}
where $\delta=\max\{2,k-1\}$. By multiplicative estimates, we have
\begin{equation*}
\begin{split}
\|\partial_x^{\alpha}(\partial_t\phi^{(1)}\phi^{\varepsilon}_R)\|_{L^2} \leq & C(\|\phi^{\varepsilon}_R\|_{H^{k}}\|\partial_t\phi^{(1)}\|_{L^{\infty}} +\|\phi^{\varepsilon}_R\|_{L^{\infty}}\|\partial_t\phi^{(1)}\|_{H^{k}})\\
\leq & C(\|\phi^{\varepsilon}_R\|_{H^{k}} +\|\phi^{\varepsilon}_R\|_{H^2}).
\end{split}
\end{equation*}
Then, for the third term on the RHS of \eqref{equ45}, we have
\begin{equation}\label{e102}
\begin{split}
I_{6233}\leq & C\varepsilon\|\phi^{\varepsilon}_R\|_{H^{s'}}^2,
\end{split}
\end{equation}
where $|\alpha|=k\leq s'$.
Therefore, adding \eqref{equ80}, \eqref{e101} and \eqref{e102} together and using \eqref{equ19} in Lemma \ref{L2}, we have
\begin{equation}\label{equ47}
\begin{split}
I_{623}\leq &-\frac12\frac{d}{dt}\langle\partial_x^{\alpha}\phi^{\varepsilon}_R, \frac{\varepsilon\phi^{(1)}}{n}\partial_x^{\alpha}\phi^{\varepsilon}_R\rangle_0\\
& +C(1+\varepsilon^2\|(n^{\varepsilon}_R,\textbf{u}^{\varepsilon}_R)\|_{H^3}^2) \{1+\|(\textbf{u}^{\varepsilon}_R, n^{\varepsilon}_R, \phi^{\varepsilon}_R)\|_{H^{s'}}^2\},
\end{split}
\end{equation}
where $\delta=\max\{2,k-1\}$, $k\leq s'$ and $s'\geq 3$.

\emph{Estimate of $I_{624}$.} We have
\begin{equation*}
\begin{split}
I_{624}=&-{\varepsilon^2}\langle\partial_x^{\alpha}\phi^{\varepsilon}_R, \frac{\phi^{\varepsilon}_R}{n} \partial_t\partial_x^{\alpha}\phi^{\varepsilon}_R\rangle_0 -{\varepsilon^2}\langle\partial_x^{\alpha}\phi^{\varepsilon}_R, \frac{1}{n}[\partial_x^{\alpha}, \phi^{\varepsilon}_R]\partial_t\phi^{\varepsilon}_R\rangle_0\\
=&:I_{6241}+I_{6242}.
\end{split}
\end{equation*}
By integrating by parts in time, we have
\begin{equation*}
\begin{split}
I_{6241}=&-\frac{1}{2}\frac{d}{dt} \int(\frac{\varepsilon^2\phi^{\varepsilon}_R}{n}) |\partial_x^{\alpha}\phi^{\varepsilon}_R|^2dx -\frac12\int\partial_t(\frac{\varepsilon^2\phi^{\varepsilon}_R}{n}) |\partial_x^{\alpha}\phi^{\varepsilon}_R|^2dx.
\end{split}
\end{equation*}
From \eqref{e3}, we have
\begin{equation*}
\begin{split}
\|\partial_t(\frac{\varepsilon^2\phi^{\varepsilon}_R}{n})\|_{L^{\infty}} \leq & C\varepsilon^2(\|\partial_t\phi^{\varepsilon}_R\|_{L^{\infty}} +\varepsilon\|\partial_t\widetilde n\|_{L^{\infty}} +\varepsilon^2\|\partial_tn^{\varepsilon}_R\|_{L^{\infty}})\\
\leq & C\varepsilon^2(\varepsilon+\|\partial_t\phi^{\varepsilon}_R\|_{H^2} +\varepsilon^2\|\partial_tn^{\varepsilon}_R\|_{H^2}).
\end{split}
\end{equation*}
From Lemma \ref{L2} and Corollary \ref{cor2}, it follows
\begin{equation*}
\begin{split}
\|\partial_t(\frac{\varepsilon^2\phi^{\varepsilon}_R}{n})\|_{L^{\infty}}
\leq & C\varepsilon(1+\|\textbf{u}^{\varepsilon}_R\|_{H^3} +\|n^{\varepsilon}_R\|_{H^3} +\|\phi^{\varepsilon}_R\|_{H^3}).
\end{split}
\end{equation*}
Hence, we obtain
\begin{equation}\label{e20}
\begin{split}
I_{6241}\leq & -\frac{1}{2}\frac{d}{dt} \int(\frac{\varepsilon^2\phi^{\varepsilon}_R}{n}) |\partial_x^{\alpha}\phi^{\varepsilon}_R|^2dx +C\varepsilon(1+\|(\textbf{u}^{\varepsilon}_R, n^{\varepsilon}_R, \phi^{\varepsilon}_R)\|_{H^3})\|\phi^{\varepsilon}_R\|_{H^k}^2.
\end{split}
\end{equation}

On the other hand, by commutator estimate, we have
\begin{equation*}
\begin{split}
\|[\partial_x^{\alpha}, \phi^{\varepsilon}_R]\partial_t\phi^{\varepsilon}_R\|_{L^2} \leq & C(\|\partial_t\phi^{\varepsilon}_R\|_{H^{k-1}} \|\phi^{\varepsilon}_R\|_{L^{\infty}} +\|\partial_t\phi^{\varepsilon}_R\|_{L^{\infty}} \|\phi^{\varepsilon}_R\|_{H^k})\\
\leq & C\{\|\phi^{\varepsilon}_R\|_{H^2} (1+\|\partial_tn^{\varepsilon}_R\|_{H^{k-1}} +\|\phi^{\varepsilon}_R\|_{H^{k-1}})\\
&+\|\phi^{\varepsilon}_R\|_{H^k} (1+\|\partial_tn^{\varepsilon}_R\|_{H^{2}} +\|\phi^{\varepsilon}_R\|_{H^{2}})\},
\end{split}
\end{equation*}
where we have used Lemma \ref{L3}. Since $k\leq s'$, by using Lemma \ref{L2}, we have
\begin{equation*}
\begin{split}
\varepsilon\|[\partial_x^{\alpha}, \phi^{\varepsilon}_R]\partial_t\phi^{\varepsilon}_R\|_{L^2} \leq & C\|\phi^{\varepsilon}_R\|_{H^{s'}} (1+\|(\textbf{u}^{\varepsilon}_R,n^{\varepsilon}_R, \phi^{\varepsilon}_R)\|_{H^{s'}}).
\end{split}
\end{equation*}
It then follows that
\begin{equation}\label{e21}
\begin{split}
|I_{6242}|\leq & C(1+\varepsilon\|(\textbf{u}^{\varepsilon}_R,n^{\varepsilon}_R, \phi^{\varepsilon}_R)\|_{H^{s'}})\|\phi^{\varepsilon}_R\|_{H^{s'}}^2.
\end{split}
\end{equation}
Combining \eqref{e20} and \eqref{e21}, we have
\begin{equation}\label{e20}
\begin{split}
I_{624}\leq & -\frac{1}{2}\frac{d}{dt} \int(\frac{\varepsilon^2\phi^{\varepsilon}_R}{n}) |\partial_x^{\alpha}\phi^{\varepsilon}_R|^2dx +C(1+\varepsilon\|(\textbf{u}^{\varepsilon}_R, n^{\varepsilon}_R, \phi^{\varepsilon}_R)\|_{H^{s'}})\|\phi^{\varepsilon}_R\|_{H^{s'}}^2.
\end{split}
\end{equation}

\emph{Estimate of $I_{625}$.} By \eqref{e52-2} in Lemma \ref{lem-A}, we have
\begin{equation}\label{equ48}
\begin{split}
I_{625} \leq & C\|\partial_x^{\alpha}\phi^{\varepsilon}_R\|^2 +\varepsilon^{4}\|\partial_t\partial_x^{\alpha}\overline{R}_{\phi}\|^2\\
\leq & C\|\partial_x^{\alpha}\phi^{\varepsilon}_R\|^2 +C_1(\sqrt{\varepsilon}\|\phi^{\varepsilon}_R\|_{H^{\delta}}) (1+\varepsilon^2\|\varepsilon\partial_t\phi^{\varepsilon}_R\|_{H^k}^2),
\end{split}
\end{equation}
where $\delta=\max\{2,k-1\}$ in Lemma \ref{lem-A}. Furthermore,
\begin{equation*}
\begin{split}
\varepsilon^2\|\varepsilon\partial_t\phi^{\varepsilon}_R\|_{H^k}^2 \leq & \varepsilon\|\varepsilon\partial_t(\varepsilon^{1/2}\nabla)\phi^{\varepsilon}_R\|_{H^{k-1}}^2 +\varepsilon^2\|\varepsilon\partial_t\phi^{\varepsilon}_R\|_{H^{k-1}}^2\\
\leq & \varepsilon\|\varepsilon\partial_t\overline{\nabla}\phi^{\varepsilon}_R\|_{H^{k-1}}^2 +\varepsilon^2\|\varepsilon\partial_t\phi^{\varepsilon}_R\|_{H^{k-1}}^2\\
\leq & C\|\varepsilon\partial_tn^{\varepsilon}_R\|_{H^{k-1}}^2+ C\varepsilon^3\|\phi^{\varepsilon}_R\|_{H^{k-1}}^2+CC_0^2\varepsilon^3,
\end{split}
\end{equation*}
where we have used \eqref{equ27}. By Corollary \ref{cor2}, when $0<\varepsilon<\varepsilon_1$, we have
\begin{equation*}
\begin{split}
\varepsilon^2\|\varepsilon\partial_t\phi^{\varepsilon}_R\|_{H^k}^2 \leq & C\{1+\|\textbf{u}^{\varepsilon}_R\|_{H^{s'}}^2 +\|n^{\varepsilon}_R\|_{H^{s'}}^2 +\varepsilon^3\|\phi^{\varepsilon}_R\|_{H^{k-1}}^2\}.
\end{split}
\end{equation*}
It then follows from \eqref{equ48} that
\begin{equation}\label{e74}
\begin{split}
I_{625} \leq & C_1(\sqrt{\varepsilon}\|\phi^{\varepsilon}_R\|_{H^{\delta}}) \{1+\|\textbf{u}^{\varepsilon}_R\|_{H^{s'}}^2 +\|n^{\varepsilon}_R\|_{H^{s'}}^2 +\|\phi^{\varepsilon}_R\|_{H^{s'}}^2\},
\end{split}
\end{equation}
where $\delta=\max\{2,k-1\}\leq s'-1$.

Summarizing, from \eqref{equ51}, \eqref{equ50}, \eqref{equ49}, \eqref{equ47}, \eqref{e20} and \eqref{e74}, we have \eqref{equ52}. The proof is complete.
\end{proof}

\begin{lemma}\label{Le-1}
Let $s'\geq 4$ be integer and $(n^{\varepsilon}_R,\textbf{u}^{\varepsilon}_R,\phi^{\varepsilon}_R)$ be a solution to \eqref{re}. The for any $0\leq k\leq s'$, there holds
\begin{equation}\label{e92}
\begin{split}
\frac12\frac{d}{dt}\int\frac{T_i}{n^2} |\partial_x^{\alpha}n^{\varepsilon}_R|^2dx\leq C(\varepsilon+\varepsilon\|(\textbf{u}^{\varepsilon}_R, n^{\varepsilon}_R)\|_{H^{s'}}) \{1+\|(n^{\varepsilon}_R,\textbf{u}^{\varepsilon}_R)\|_{H^{s'}}^2\}+J_{31},
\end{split}
\end{equation}
where $\alpha$ is any multi-index with $|\alpha|=k$ and
\begin{equation}\label{e95}
\begin{split}
J_{31}=-T_i\langle(\overline{\nabla} \cdot\partial_x^{\alpha}\textbf{u}^{\varepsilon}_R), \frac{1}{\varepsilon n}\partial_x^{\alpha}n^{\varepsilon}_R\rangle_0.
\end{split}
\end{equation}
\end{lemma}
\begin{proof}
Let $\alpha$ be a multi-index with $|\alpha|=k\leq s'$. We take $\partial_x^{\alpha}$ of \eqref{re-1}, and then take inner product of $\frac{T_i}{n^2}\partial_x^{\alpha}n^{\varepsilon}_R$ in $L^2$. Integrating in $t$, we obtain
\begin{equation}\label{e11}
\begin{split}
\frac12\frac{d}{dt}\int\frac{T_i}{n^2} |\partial_x^{\alpha}n^{\varepsilon}_R|^2dx =&\frac{1}2\int\partial_t(\frac{T_i}{n^2}) |\partial_x^{\alpha}n^{\varepsilon}_R|^2dx +\langle\partial_x^{\alpha}(\frac{V\textbf{e}_1-\textbf{u}}{\varepsilon} \cdot\overline{\nabla} n^{\varepsilon}_R), \frac{T_i}{n^2}\partial_x^{\alpha}n^{\varepsilon}_R\rangle_0\\
&-\langle\partial_x^{\alpha}(\frac{n}{\varepsilon}\overline{\nabla} \cdot\textbf{u}^{\varepsilon}_R), \frac{T_i}{n^2}\partial_x^{\alpha}n^{\varepsilon}_R\rangle_0 -\langle\partial_x^{\alpha} (n^{\varepsilon}_R\overline{\nabla}\cdot\widetilde{\textbf{u}}), \frac{T_i}{n^2}\partial_x^{\alpha}n^{\varepsilon}_R\rangle_0\\
&-\langle\partial_x^{\alpha} (\textbf{u}^{\varepsilon}_R\cdot\overline{\nabla}\widetilde{n}), \frac{T_i}{n^2}\partial_x^{\alpha}n^{\varepsilon}_R\rangle_0  -\langle\partial_x^{\alpha}(\varepsilon R_n), \frac{T_i}{n^2}\partial_x^{\alpha}n^{\varepsilon}_R\rangle_0\\
=& J_1+\cdots+J_6.
\end{split}
\end{equation}

\emph{Estimate of $J_{1}$.}  By direct computation, we have
\begin{equation*}
\begin{split}
\partial_t(\frac{T_i}{n^2})=&-2\frac{1}{n^3}(\varepsilon\partial_{t}\widetilde n+\varepsilon^2\partial_{t}n^{\varepsilon}_R),
\end{split}
\end{equation*}
which yields that
\begin{equation}\label{e6}
\begin{split}
\|\partial_t(\frac{T_i}{n^2})\|_{L^{\infty}}\leq C(\varepsilon+\varepsilon\|\varepsilon\partial_tn^{\varepsilon}_R\|_{L^{\infty}}),
\end{split}
\end{equation}
where we have used \eqref{assumption} and \eqref{assumption1}. By Lemma \ref{L2} and \ref{L1} and Sobolev embedding, we have
\begin{equation*}
\begin{split}
\|\varepsilon\partial_tn^{\varepsilon}_R\|_{L^{\infty}}\leq \|\varepsilon\partial_tn^{\varepsilon}_R\|_{H^2}\leq C\{1+\|\textbf{u}^{\varepsilon}_R\|_{H^{3}} +\|n^{\varepsilon}_R\|_{H^{3}}\}.
\end{split}
\end{equation*}
Therefore,
\begin{equation}\label{e87}
\begin{split}
|J_1|\leq C(\varepsilon+\varepsilon\|(\textbf{u}^{\varepsilon}_R, n^{\varepsilon}_R)\|_{H^{3}})\|n^{\varepsilon}_R\|_{H^k}^2.
\end{split}
\end{equation}

\emph{Estimate of $J_{2}$.} Expanding the expression of $J_2$, we have
\begin{equation}\label{e86}
\begin{split}
J_2=\langle\partial_x^{\alpha}(\frac{V}{\varepsilon} \cdot\partial_{x_1}n^{\varepsilon}_R), \frac{T_i}{n^2}\partial_x^{\alpha}n^{\varepsilon}_R\rangle_0 -\langle\partial_x^{\alpha}(\frac{\textbf{u}}{\varepsilon} \cdot\overline{\nabla} n^{\varepsilon}_R), \frac{T_i}{n^2}\partial_x^{\alpha}n^{\varepsilon}_R\rangle_0=:J_{21}+J_{22}.
\end{split}
\end{equation}
Similar to \eqref{e6}, we have
\begin{equation*}
\begin{split}
\|\partial_{x_1}(\frac{T_i}{n^2})\|_{L^{\infty}}\leq & C(\varepsilon +\varepsilon^2\|\partial_{x_1}n^{\varepsilon}_R\|_{L^{\infty}})\\
\leq & C(\varepsilon +\varepsilon^2\|n^{\varepsilon}_R\|_{H^3}).
\end{split}
\end{equation*}
Hence, by integrating by parts, we have
\begin{equation}\label{e7}
\begin{split}
|J_{21}|=& \left|-\frac{V}2\int\partial_{x_1}(\frac{T_i}{\varepsilon n^2})|\partial_x^{\alpha}n^{\varepsilon}_R|^2dx\right|\\
\leq & C(1+\varepsilon\|n^{\varepsilon}_R\|_{H^3})\|n^{\varepsilon}_R\|_{\dot H^k}^2.
\end{split}
\end{equation}
Next, we estimate $J_{22}$ in \eqref{e86}. We have
\begin{equation*}
\begin{split}
J_{22}=& -\left\langle\left[\partial_x^{\alpha},\frac{\textbf{u}}{\varepsilon}\right] \overline{\nabla}n^{\varepsilon}_R, \frac{T_i}{n^2}\partial_x^{\alpha}n^{\varepsilon}_R\right\rangle_0 -\left\langle\frac{\textbf{u}}{\varepsilon} \partial_x^{\alpha}\overline{\nabla}n^{\varepsilon}_R, \frac{T_i}{n^2}\partial_x^{\alpha}n^{\varepsilon}_R\right\rangle_0 =:J_{221}+J_{222}.
\end{split}
\end{equation*}
By commutator estimate, we have
\begin{equation*}
\begin{split}
\|\left[\partial_x^{\alpha},\frac{\textbf{u}}{\varepsilon}\right] \overline{\nabla}n^{\varepsilon}_R\|_{L^2}
& \leq  C(\|\overline{\nabla}n^{\varepsilon}_R\|_{H^{k-1}} \|\nabla(\frac{\textbf{u}}{\varepsilon})\|_{L^{\infty}} +\|\overline{\nabla}n^{\varepsilon}_R\|_{L^{\infty}} \|\frac{\textbf{u}}{\varepsilon}\|_{H^k})\\
& \leq  C\{\|n^{\varepsilon}_R\|_{H^{k}}(\|\widetilde{\textbf{u}}\|_{H^3} +\varepsilon\|\textbf{u}^{\varepsilon}_R\|_{H^3}) +\|n^{\varepsilon}_R\|_{H^{3}}(\|\widetilde{\textbf{u}}\|_{H^k} +\varepsilon\|\textbf{u}^{\varepsilon}_R\|_{H^k})\}\\
& \leq C(1+\varepsilon\|\textbf{u}^{\varepsilon}_R\|_{H^{s'}}) \|n^{\varepsilon}_R\|_{H^{s'}},
\end{split}
\end{equation*}
where we have used $k\leq s'$, $s'\geq3$ and \eqref{e3}. Therefore, we have
\begin{equation}\label{e8}
\begin{split}
|J_{221}|
\leq C(1+\varepsilon\|\textbf{u}^{\varepsilon}_R\|_{H^{s'}}) \|n^{\varepsilon}_R\|_{H^{s'}}^2.
\end{split}
\end{equation}
On the other hand, from \eqref{assumption1} and \eqref{e1}, by H\"older inequality, we know
\begin{equation*}
\begin{split}
\|\overline{\nabla}(\frac{\textbf{u}}{\varepsilon n^2})\|_{L^{\infty}}\leq & \|\frac{\overline{\nabla}\widetilde{\textbf{u}}+\varepsilon \overline{\nabla}\textbf{u}^{\varepsilon}_R}{n^2}\|_{L^{\infty}} +2\|\textbf{u}\frac{\overline{\nabla}\widetilde{n}+\varepsilon \overline{\nabla}n^{\varepsilon}_R}{n^3}\|_{L^{\infty}}\\
\leq & C+C\varepsilon(\|\overline{\nabla}\textbf{u}^{\varepsilon}_R\|_{L^{\infty}} +\|\overline{\nabla}n^{\varepsilon}_R\|_{L^{\infty}})\\
\leq & C+C\varepsilon(\|\textbf{u}^{\varepsilon}_R\|_{H^3} +\|n^{\varepsilon}_R\|_{H^3}).
\end{split}
\end{equation*}
By integrating by parts, we obtain
\begin{equation}\label{e9}
\begin{split}
|J_{222}|
=&\left|\frac{T_i}2\langle\overline{\nabla}(\frac{\textbf{u}}{\varepsilon n^2})\partial_x^{\alpha}n^{\varepsilon}_R, \partial_x^{\alpha}n^{\varepsilon}_R\rangle_0\right|\\
\leq & C(1+\varepsilon(\|\textbf{u}^{\varepsilon}_R\|_{H^3} +\|n^{\varepsilon}_R\|_{H^3}))\|\partial_x^{\alpha}n^{\varepsilon}_R\|^2.
\end{split}
\end{equation}
Adding the estimates \eqref{e7}, \eqref{e8} and \eqref{e9} together, we have
\begin{equation}\label{e10}
\begin{split}
|J_{2}|\leq & C(1+\varepsilon(\|\textbf{u}^{\varepsilon}_R\|_{H^{s'}} +\|n^{\varepsilon}_R\|_{H^{s'}}))\|n^{\varepsilon}_R\|_{H^{s'}}^2,
\end{split}
\end{equation}
where $k\leq s'$ and $s'\geq 3$.

\emph{Estimate of $J_{4}$.} For $J_4$ in \eqref{e11}, we have
\begin{equation}\label{e88}
\begin{split}
|J_4|\leq C\|n^{\varepsilon}_R\|_{H^{k}}^2.
\end{split}
\end{equation}

\emph{Estimate of $J_{5}$.} We have
\begin{equation}\label{e89}
\begin{split}
|J_5|\leq C\|\textbf{u}^{\varepsilon}_R\|^2_{H^k} +C\|n^{\varepsilon}_R\|^2_{H^k}.
\end{split}
\end{equation}

\emph{Estimate of $J_{6}$.} We have
\begin{equation}\label{e90}
\begin{split}
|J_6|\leq C\varepsilon^2+C\|n^{\varepsilon}_R\|^2_{H^k}.
\end{split}
\end{equation}

\emph{Estimate of $J_{3}$ in \eqref{e11}.} $J_3$ can be written in the commutator form
\begin{equation*}
\begin{split}
J_3=& -T_i\langle(\overline{\nabla} \cdot\partial_x^{\alpha}\textbf{u}^{\varepsilon}_R), \frac{1}{\varepsilon n}\partial_x^{\alpha}n^{\varepsilon}_R\rangle_0 -T_i\langle[\partial_x^{\alpha},\frac{n}{\varepsilon}]\overline{\nabla} \cdot\textbf{u}^{\varepsilon}_R, \frac{1}{n^2}\partial_x^{\alpha}n^{\varepsilon}_R\rangle_0\\
=& J_{31}+J_{32}.
\end{split}
\end{equation*}
By commutator estimate, we have
\begin{equation*}
\begin{split}
\|[\partial_x^{\alpha},\frac{n}{\varepsilon}]\overline{\nabla} \cdot\textbf{u}^{\varepsilon}_R\|_{L^2}\leq & C(\|\overline{\nabla}\cdot \textbf{u}^{\varepsilon}_R\|_{\dot H^{k-1}} \|\nabla(\frac{n}{\varepsilon})\|_{L^{\infty}} +\|\overline{\nabla}\cdot \textbf{u}^{\varepsilon}_R\|_{L^{\infty}}\|\frac{n}{\varepsilon}\|_{\dot H^k})\\
\leq & C(1+\varepsilon\|n^{\varepsilon}_R\|_{H^{s'}}) \|\textbf{u}^{\varepsilon}_R\|_{H^{s'}}.
\end{split}
\end{equation*}
Therefore, we have
\begin{equation}\label{e91}
\begin{split}
|J_{32}|\leq C(1+\varepsilon\|n^{\varepsilon}_R\|_{H^{s'}}) (\|\textbf{u}^{\varepsilon}_R\|_{H^{s'}} +\|n^{\varepsilon}_R\|_{H^{s'}}),
\end{split}
\end{equation}
thanks to \eqref{assumption1}.

From \eqref{e11}, by adding \eqref{e87}, \eqref{e10}, \eqref{e88}, \eqref{e89}, \eqref{e90} and \eqref{e91} together, we obtain \eqref{e92}. The proof is complete.
\end{proof}

\begin{proof}[\textbf{Proof of Proposition \ref{prop-kp2}}]
By adding \eqref{e93} and \eqref{e92} together, we obtain
\begin{equation}\label{e32}
\begin{split}
\frac12\frac{d}{dt}\|\partial_x^{\alpha}\textbf{u}^{\varepsilon}_R\|_{L^2}^2 &+\frac12\frac{d}{dt}\int\frac{1+\varepsilon\phi^{(1)} +\varepsilon^2\phi^{\varepsilon}_R}{n} |\partial_x^{\alpha}\phi^{\varepsilon}_R|^2\\
&+\frac12\frac{d}{dt}\int\frac{\varepsilon}{n} |\overline{\nabla}\partial_x^{\alpha}\phi^{\varepsilon}_R|^2 +\frac12\frac{d}{dt}\int\frac{T_i}{n^2} |\partial_x^{\alpha}n^{\varepsilon}_R|^2dx\\
\leq & CC_1(C_1+\varepsilon\|(\textbf{u}^{\varepsilon}_R, n^{\varepsilon}_R, \phi^{\varepsilon}_R)\|_{H^{s'}}^2)\{1 +\|(\textbf{u}^{\varepsilon}_R, n^{\varepsilon}_R, \phi^{\varepsilon}_R)\|_{H^{s'}}^2\}+I_{71}+J_{31},
\end{split}
\end{equation}
where $I_{71}$ and $J_{31}$ are given in \eqref{e94} and \eqref{e95} respectively. By integration by parts, we have
\begin{equation}\label{e96}
\begin{split}
I_{71}+J_{31}= & -T_i\langle\frac{1}{\varepsilon n}\partial_x^{\alpha}\overline{\nabla}n^{\varepsilon}_R, \partial_x^{\alpha}\textbf{u}^{\varepsilon}_R\rangle_0 -T_i\langle(\overline{\nabla} \cdot\partial_x^{\alpha}\textbf{u}^{\varepsilon}_R), \frac{1}{\varepsilon n}\partial_x^{\alpha}n^{\varepsilon}_R\rangle_0\\
= & T_i\langle\overline{\nabla}(\frac{1}{\varepsilon n})\partial_x^{\alpha}n^{\varepsilon}_R, \partial_x^{\alpha}\textbf{u}^{\varepsilon}_R\rangle_0\\
\leq & C(1+\varepsilon\|n^{\varepsilon}_R\|_{H^3}) \{\|n^{\varepsilon}_R\|_{H^{k}}+\|\textbf{u}^{\varepsilon}_R\|_{H^{k}}\}.
\end{split}
\end{equation}
Inserting \eqref{e96} into \eqref{e32}, we complete the proof of Proposition \ref{prop-kp2}.
\end{proof}

\begin{proof}[\textbf{Proof of Theorem \ref{th} for $T_i>0$}]
Recalling \eqref{assumption1}, 
and integrating \eqref{e-prop1} over $[0,t]$ and taking summation over $|\alpha|=k$ and $0\leq k\leq s'$, we obtain
\begin{equation}
\begin{split}
\|\textbf{u}^{\varepsilon}_R\|_{H^{s'}}^2 &+\|\phi^{\varepsilon}_R\|_{H^{s'}}^2 +\varepsilon \|\overline{\nabla}\phi^{\varepsilon}_R\|_{H^{s'}}^2 +T_i\|n^{\varepsilon}_R\|_{H^{s'}}^2\\
\leq & CC_{\varepsilon}(0)+CC_1\int_0^t(C_1 +\varepsilon\|(\textbf{u}^{\varepsilon}_R, n^{\varepsilon}_R, \phi^{\varepsilon}_R)\|_{H^{s'}}^2)\{1 +\|(\textbf{u}^{\varepsilon}_R, n^{\varepsilon}_R, \phi^{\varepsilon}_R)\|_{H^{s'}}^2\}dr,
\end{split}
\end{equation}
where $C_{\varepsilon}(0)=\|(\textbf{u}^{\varepsilon}_R, \phi^{\varepsilon}_R)(0)\|_{H^{s'}}^2 +\varepsilon \|\overline{\nabla}\phi^{\varepsilon}_R(0)\|_{H^{s'}}^2 +T_i\|n^{\varepsilon}_R(0)\|_{H^{s'}}^2$. From \eqref{assumption}, there exists some constant $0<\varepsilon_0<\varepsilon_1$ (in Lemma \ref{L1}) such that $\varepsilon\|(\textbf{u}^{\varepsilon}_R, n^{\varepsilon}_R, \phi^{\varepsilon}_R)\|_{H^{s'}}^2\leq 1$ for any $0<\varepsilon<\varepsilon_0$. Since $C_1=C_1(\sqrt{\varepsilon}\|n^{\varepsilon}_R\|_{H^{s'}})$ and is nondecreasing, we know that $C_1\leq C_1(1)$ when $0<\varepsilon<\varepsilon_0$. Since $T_i>0$, there exists some constant $C_3>1$ such that
\begin{equation}
\begin{split}
\|\textbf{u}^{\varepsilon}_R\|_{H^{s'}}^2 +\|\phi^{\varepsilon}_R\|_{H^{s'}}^2 +\|n^{\varepsilon}_R\|_{H^{s'}}^2\leq C_3C_{\varepsilon}(0) +C_3\int_0^t\{1+\|(\textbf{u}^{\varepsilon}_R, n^{\varepsilon}_R, \phi^{\varepsilon}_R)\|_{H^{s'}}^2\}dr.
\end{split}
\end{equation}

For any given $0<\tau_0<\tau_*$, let $C_0'=\sup_{0<\varepsilon<1}C_{\varepsilon}(0)$ and $\tilde C$ in \eqref{assumption} satisfy $\tilde C\geq 2(1+CC_0')e^{C_3\tau_0}$, then by Gronwall inequality, we obtain
\begin{equation}
\begin{split}
\sup_{0\leq t\leq\tau_0}\|\textbf{u}^{\varepsilon}_R\|_{H^{s'}}^2 +\|\phi^{\varepsilon}_R\|_{H^{s'}}^2 +\|n^{\varepsilon}_R\|_{H^{s'}}^2 \leq (1+CC_3)e^{C_3\tau_0}\leq \tilde C.
\end{split}
\end{equation}
It is then standard to obtain uniform estimates for $\|(n^{\varepsilon}_R,\textbf{u}^{\varepsilon}_R, \phi^{\varepsilon}_R)\|_{H^{s'}}$ independent of $\varepsilon$ by the continuity method. The proof is complete for the case $T_i>0$.
\end{proof}

\section{Proof of Theorem \ref{th} for $T_i=0$}
\setcounter{section}{4}\setcounter{equation}{0}
In this section, we prove Theorem \ref{th} for the case of $T_i=0$ and $d=3$, i.e., we prove \eqref{estimate2}. In this case, $\overline{\nabla}$ and $\overline{\Delta}$ reduce to $\nabla=(\partial_{x_1}, \partial_{x_2}, \partial_{x_3})$ and $\Delta=\partial_{x_1}^2+\partial_{x_2}^2+\partial_{x_3}^2$. Since $T_i=0$, we obtain $V=1$ from \eqref{v=1}.

We also assume that \eqref{re} has smooth solutions in a small time $\tau_{\varepsilon}$ dependent on $\varepsilon$. As in Section 3, we let $\tilde C$ be a constant, which will be determined later, much larger than the bound of $|\!|\!|(n^{\varepsilon}_R,\textbf{u}^{\varepsilon}_R, \phi^{\varepsilon}_R)(0)|\!|\!|_{s'}$, such that on $[0,\tau_{\varepsilon}]$
\begin{equation}\label{e25}
\begin{split}
\sup_{[0,\tau_{\varepsilon}]}|\!|\!|(n^{\varepsilon}_R,\textbf{u}^{\varepsilon}_R, \phi^{\varepsilon}_R)|\!|\!|_{s'}\leq \tilde C.
\end{split}
\end{equation}
We will prove that $\tau_{\varepsilon}>\tau_0$ as $\varepsilon\to 0$ for some $0<\tau_0<\tau_*$, where $\tau_*$ is the existence time of the limit equation \eqref{ZKE}. Recalling the expressions for $n$ and $\textbf{u}$ in \eqref{e3}, we immediately know that there exists some $\varepsilon_1=\varepsilon_1(\tilde C)>0$ such that on $[0,\tau_{\varepsilon}]$,
\begin{equation}\label{e26}
\begin{split}
1/2<n<3/2,\ \ \ |\textbf{u}|\leq 1/2,
\end{split}
\end{equation}
for all $0<\varepsilon<\varepsilon_1$.


\begin{proposition}\label{prop1}
Let $s'\geq 4$ be an integer and $(n^{\varepsilon}_R,\textbf{u}^{\varepsilon}_R,\phi^{\varepsilon}_R)$ be a solution to \eqref{re}. Then for any integer $0\leq k\leq s'$, there holds
\begin{equation}\label{e97}
\begin{split}
\frac12\frac{d}{dt}\|\partial_x^{\alpha} \textbf{u}^{\varepsilon}_R\|_{L^2}^2 &+\frac12\frac{d}{dt}\int\frac{1+\varepsilon\phi^{(1)} +\varepsilon^2\phi^{\varepsilon}_R}{n} |\partial_x^{\alpha}\phi^{\varepsilon}_R|^2 +\frac12\frac{d}{dt}\int\frac{\varepsilon}{n} |{\nabla}\partial_x^{\alpha}\phi^{\varepsilon}_R|^2\\
\leq & CC_1(1+\varepsilon|\!|\!|(\textbf{u}^{\varepsilon}_R, \phi^{\varepsilon}_R)|\!|\!|_{s'}^2) \{1+|\!|\!|(\textbf{u}^{\varepsilon}_R, \phi^{\varepsilon}_R)|\!|\!|_{s'}^2\},
\end{split}
\end{equation}
where $\alpha$ is any multi-index with $|\alpha|=k\leq s'$ and $|\!|\!|(\textbf{u}^{\varepsilon}_R,\phi^{\varepsilon}_R)|\!|\!|_{s'}^2$ is defined in \eqref{tri-norm}.
\end{proposition}

\begin{proof}
The proof of is basically contained in Section 3. Here, we only give the differences. Since $T_i=0$, there are no $I_7$, $I_8$ and $I_9$ in the proof of Lemma \ref{L8}. Then from Lemma \ref{L8}, we obtain
\begin{equation}
\begin{split}
\frac12\frac{d}{dt}\|\partial_x^{\alpha}\textbf{u}^{\varepsilon}_R\|_{L^2}^2 &+\frac12\frac{d}{dt}\int\frac{1+\varepsilon\phi^{(1)} +\varepsilon^2\phi^{\varepsilon}_R}{n} |\partial_x^{\alpha}\phi^{\varepsilon}_R|^2 +\frac12\frac{d}{dt}\int\frac{\varepsilon}{n} |{\nabla}\partial_x^{\alpha}\phi^{\varepsilon}_R|^2\\
\leq & CC_1(C_1+\varepsilon\|(\textbf{u}^{\varepsilon}_R, n^{\varepsilon}_R, \phi^{\varepsilon}_R)\|_{H^{s'}}^2)\{1 +\|(\textbf{u}^{\varepsilon}_R, n^{\varepsilon}_R, \phi^{\varepsilon}_R)\|_{H^{s'}}^2\}.
\end{split}
\end{equation}
%
The proof is then complete by replacing $\|n^{\varepsilon}_R\|_{H^{s'}}^2$ with $|\!|\!|\phi^{\varepsilon}_R|\!|\!|_{s'}$, thanks to Lemma \ref{L1}.
\end{proof}

To obtain uniform estimates for the remainder term $(\textbf{u}^{\varepsilon}_R, n^{\varepsilon}_R, \phi^{\varepsilon}_R)$, we will prove the following
\begin{proposition}\label{prop2}
Let $(n^{\varepsilon}_R,\textbf{u}^{\varepsilon}_R, \phi^{\varepsilon}_R)$ be a solution to \eqref{re} in 3D, then
\begin{equation}\label{equ61}
\begin{split}
\frac{\varepsilon}2\frac{d}{dt}\|{\nabla} \textbf{u}^{\varepsilon}_R\|^2_{\dot H^{s'}} &+\frac12\frac{d}{dt}\langle\partial_x^{\alpha}{\nabla} \phi^{\varepsilon}_R, \frac{\varepsilon(1+\varepsilon\phi^{(1)})}{n} \partial_x^{\alpha}{\nabla}\phi^{\varepsilon}_R\rangle_0 +\frac12\frac{d}{dt} \langle\partial_x^{\alpha}{\Delta}\phi^{\varepsilon}_R, \frac{\varepsilon^2}{n}\partial_x^{\alpha} {\Delta}\phi^{\varepsilon}_R\rangle_0\\
&\leq C(1+C_1(\sqrt{\varepsilon}\|\phi^{\varepsilon}_R\|_{H^{3}})) (1+\varepsilon^2|\!|\!|(\textbf{u}^{\varepsilon}_R, \phi^{\varepsilon}_R)|\!|\!|_{s'}^2) (1+|\!|\!|(\textbf{u}^{\varepsilon}_R, \phi^{\varepsilon}_R)|\!|\!|_{s'}^2),
\end{split}
\end{equation}
where $|\!|\!|(\textbf{u}^{\varepsilon}_R,\phi^{\varepsilon}_R)|\!|\!|_{s'}^2$ is defined in \eqref{tri-norm} and $|\alpha|=s'$ for any $s'\geq 4$.
\end{proposition}
\begin{proof}
Let $\alpha$ be any multi-index with $|\alpha|=k\leq s'$ for any $4\leq s'\leq s$. We take $\partial_x^{\alpha}$ of \eqref{re-2} and then take inner product with $\varepsilon{\Delta}\partial_x^{\alpha} \textbf{u}^{\varepsilon}_R$ in $L^2(\Bbb R^3)$. By integrating by parts, we obtain
\begin{equation}\label{equ87}
\begin{split}
\frac{\varepsilon}2\frac{d}{dt}\|\partial_x^{\alpha}{\nabla} \textbf{u}^{\varepsilon}_R\|^2= & \langle\partial_{x_1}\partial_x^{\alpha}{\nabla} \textbf{u}^{\varepsilon}_R, \partial_x^{\alpha}{\nabla} \textbf{u}^{\varepsilon}_R\rangle_{0} -\langle\partial_x^{\alpha}({\nabla}\textbf{u}\cdot {\nabla}\textbf{u}^{\varepsilon}_R), \partial_x^{\alpha}{\nabla}\textbf{u}^{\varepsilon}_R\rangle_{0}\\
&-\langle\partial_x^{\alpha}(\textbf{u}\cdot {\nabla}^2\textbf{u}^{\varepsilon}_R), \partial_x^{\alpha}{\nabla}\textbf{u}^{\varepsilon}_R\rangle_{0}
-\varepsilon\langle\partial_x^{\alpha}({\nabla} \textbf{u}^{\varepsilon}_R\cdot{\nabla} \widetilde{\textbf{u}}), \partial_x^{\alpha}{\nabla}\textbf{u}^{\varepsilon}_R\rangle_0\\
&-\varepsilon\langle\partial_x^{\alpha}(\textbf{u}^{\varepsilon}_R \cdot{\nabla}^2 \widetilde{\textbf{u}}), \partial_x^{\alpha}{\nabla}\textbf{u}^{\varepsilon}_R\rangle_0 -\varepsilon^2\langle \partial_x^{\alpha}{\nabla}\textbf{R}_\textbf{u}, \partial_x^{\alpha}{\nabla}\textbf{u}^{\varepsilon}_R\rangle_0\\
&+\frac{1}{\varepsilon^{1/2}}\langle \textbf{u}^{\varepsilon}_R\times \textbf{e}_1, \textbf{u}^{\varepsilon}_R\rangle_{\dot H^k} -\langle\partial_x^{\alpha}{\Delta}\phi^{\varepsilon}_R, \partial_x^{\alpha}{\nabla}\cdot \textbf{u}^{\varepsilon}_R\rangle_0\\
=&:K_1+\cdots+K_8.
\end{split}
\end{equation}

First, we know that $K_1$ and $K_7$ vanish by integration by parts.

\emph{Estimate of $K_2$.} Using the commutator estimate, we have
\begin{equation}
\begin{split}
|K_2|\leq & |\varepsilon\langle\partial_x^{\alpha}({\nabla} \widetilde{\textbf{u}}\cdot{\nabla}\textbf{u}^{\varepsilon}_R), \partial_x^{\alpha}{\nabla}\textbf{u}^{\varepsilon}_R\rangle_{0}| +|\varepsilon^2\langle\partial_x^{\alpha}({\nabla} \textbf{u}^{\varepsilon}_R\cdot{\nabla}\textbf{u}^{\varepsilon}_R), \partial_x^{\alpha}{\nabla}\textbf{u}^{\varepsilon}_R\rangle_{0}|\\
\leq & C\varepsilon\|{\nabla}\textbf{u}^{\varepsilon}_R\|_{H^k}^2 +\varepsilon^2\|{\nabla}\textbf{u}^{\varepsilon}_R\|_{L^{\infty}} \|{\nabla}\textbf{u}^{\varepsilon}_R\|_{H^k}^2\\
\leq & C(1+\varepsilon\|{\nabla}\textbf{u}^{\varepsilon}_R\|_{H^2}) (\varepsilon\|{\nabla}\textbf{u}^{\varepsilon}_R\|_{H^k}^2),
\end{split}
\end{equation}
where we have used Lemma \ref{Le-inequ} and Sobolev embedding.

\emph{Estimate of $K_3$.} By using commutator and integrating by parts, we have
\begin{equation}
\begin{split}
K_3= &-\langle{\textbf{u}}\cdot \partial_x^{\alpha}{\nabla}^2\textbf{u}^{\varepsilon}_R, \partial_x^{\alpha}{\nabla}\textbf{u}^{\varepsilon}_R\rangle_{0} -\langle[\partial_x^{\alpha},\textbf{u}]\cdot{\nabla}^2 \textbf{u}^{\varepsilon}_R,\partial_x^{\alpha}{\nabla} \textbf{u}^{\varepsilon}_R\rangle_{0}\\
=& \frac12\langle{\nabla}\cdot{\textbf{u}} \partial_x^{\alpha}{\nabla}\textbf{u}^{\varepsilon}_R, \partial_x^{\alpha}{\nabla}\textbf{u}^{\varepsilon}_R\rangle_{0} -\langle[\partial_x^{\alpha},\textbf{u}]\cdot{\nabla}^2 \textbf{u}^{\varepsilon}_R,\partial_x^{\alpha}{\nabla} \textbf{u}^{\varepsilon}_R\rangle_{0}\\
=&:K_{31}+K_{32}.
\end{split}
\end{equation}
Recalling \eqref{e3}, we know
\begin{equation}
\begin{split}
K_{31}\leq & \langle{\nabla}\cdot{\textbf{u}} \partial_x^{\alpha}{\nabla}\textbf{u}^{\varepsilon}_R, \partial_x^{\alpha}{\nabla}\textbf{u}^{\varepsilon}_R\rangle_0\\
\leq & C\varepsilon(1+\varepsilon\|{\nabla}\cdot \textbf{u}^{\varepsilon}_R\|_{L^{\infty}}) \|\partial_x^{\alpha}{\nabla}\textbf{u}^{\varepsilon}_R\|^2\\
\leq & C(1+\varepsilon\|\textbf{u}^{\varepsilon}_R\|_{H^3}) (\varepsilon \|{\nabla}\textbf{u}^{\varepsilon}_R\|_{H^{k}}^2).
\end{split}
\end{equation}
Similarly, using \eqref{e3} and Lemma \ref{Le-inequ}, we have 
\begin{equation}
\begin{split}
K_{32}\leq & C(\|{\nabla}^2\textbf{u}^{\varepsilon}_R\|_{H^{k-1}} \|{\textbf{u}}\|_{L^{\infty}} +\|{\nabla}^2\textbf{u}^{\varepsilon}_R\|_{L^{\infty}} \|{\textbf{u}}\|_{H^k}) \|\partial_x^{\alpha}{\nabla}\textbf{u}^{\varepsilon}_R\|^2\\
\leq & C(1+\varepsilon\|\textbf{u}^{\varepsilon}_R\|_{L^{\infty}}) (\varepsilon\|{\nabla}\textbf{u}^{\varepsilon}_R\|_{H^k}^2)\\
&+C\varepsilon\|{\nabla}^2\textbf{u}^{\varepsilon}_R\|_{L^{\infty}} (1+\varepsilon \|\textbf{u}^{\varepsilon}_R\|_{H^k}) \|\partial_x^{\alpha}{\nabla}\textbf{u}^{\varepsilon}_R\|_{L^{2}}\\
\leq & C(1+\varepsilon\|\textbf{u}^{\varepsilon}_R\|_{H^4}) (\|\textbf{u}^{\varepsilon}_R\|_{H^{s'}}^2 +\varepsilon\|{\nabla}\textbf{u}^{\varepsilon}_R\|_{H^{s'}}^2),
\end{split}
\end{equation}
where $k\leq s'$ and  $s'\geq 4$. Therefore, using \eqref{tri-norm}, we have
\begin{equation}
\begin{split}
K_{3} \leq & C(1+\varepsilon\|\textbf{u}^{\varepsilon}_R\|_{H^4}) |\!|\!|\textbf{u}^{\varepsilon}_R|\!|\!|_{{s'}}^2.
\end{split}
\end{equation}

\emph{Estimate of $K_4$, $K_5$ and $K_6$.} The estimate of $K_4$ and $K_5$ is similar to $K_2$ and $K_3$. We obtain
\begin{equation}
\begin{split}
K_{4},K_5 \leq & C|\!|\!|\textbf{u}^{\varepsilon}_R|\!|\!|_{{s'}}^2.
\end{split}
\end{equation}
Since $\textbf{R}_\textbf{u}$ depends only on $\widetilde{\textbf{u}}$, we know that
\begin{equation}
\begin{split}
K_{6} \leq & C(1+\varepsilon\|{\nabla}\textbf{u}^{\varepsilon}_R\|_{H^{s'}}^2).
\end{split}
\end{equation}

Summarizing, we have
\begin{equation}\label{equ88}
\begin{split}
\sum_{i=1}^7|K_{i}| \leq  C(1+\varepsilon\|\textbf{u}^{\varepsilon}_R\|_{H^4}) (1+|\!|\!|\textbf{u}^{\varepsilon}_R|\!|\!|_{{s'}}^2).
\end{split}
\end{equation}

\emph{Estimate of $K_8$.} Taking $\partial_x^{\alpha}$ of \eqref{re-1} with $\overline{\nabla}=\nabla$ in 3D, we obtain
\begin{equation}\label{e53}
\begin{split}
\frac{1}{\varepsilon}{\nabla}\cdot\partial_x^{\alpha} \textbf{u}^{\varepsilon}_R=& \frac{1}{n}\{\frac{V\textbf{e}_1-\textbf{u}}{\varepsilon} \cdot{\nabla} \partial_x^{\alpha}n^{\varepsilon}_R -\partial_t\partial_x^{\alpha}n^{\varepsilon}_R -[\partial_x^{\alpha},\frac{\textbf{u}}{\varepsilon}] \cdot{\nabla}n^{\varepsilon}_R\\
&-[\partial_x^{\alpha}, \frac{n}{\varepsilon}]{\nabla}\cdot\textbf{u}^{\varepsilon}_R -\partial_x^{\alpha}(n^{\varepsilon}_R\nabla\cdot\widetilde{\textbf{u}} +\textbf{u}^{\varepsilon}_R\cdot\nabla\widetilde{n}+\varepsilon R_n)\},
\end{split}
\end{equation}
where $|\alpha|=k$. Accordingly, we have the decomposition
\begin{equation}\label{equ67}
\begin{split}
K_8=&-\langle\partial_x^{\alpha}{\Delta}\phi^{\varepsilon}_R, \frac{V\textbf{e}_1-\textbf{u}}{n}\cdot {\nabla}\partial_x^{\alpha}n^{\varepsilon}_R\rangle_{0} +\langle\partial_x^{\alpha}{\Delta}\phi^{\varepsilon}_R, \frac{\varepsilon}{n}\partial_t\partial_x^{\alpha}n^{\varepsilon}_R \rangle_{0}\\
&+\langle\partial_x^{\alpha}{\Delta}\phi^{\varepsilon}_R, \frac{1}{n}[\partial_x^{\alpha},{\textbf{u}}] \cdot{\nabla}n^{\varepsilon}_R\rangle_{0} +\langle\partial_x^{\alpha}{\Delta}\phi^{\varepsilon}_R, \frac{1}{n}[\partial_x^{\alpha},{n}]{\nabla} \cdot\textbf{u}^{\varepsilon}_R\rangle_{0}\\
& +\langle\partial_x^{\alpha}{\Delta}\phi^{\varepsilon}_R, \frac{\varepsilon}{n}\partial_x^{\alpha} (n^{\varepsilon}_R\nabla\cdot\widetilde{\textbf{u}} +\textbf{u}^{\varepsilon}_R\cdot\nabla\widetilde{n}+\varepsilon R_n)\rangle_{0}\\
=&:K_{81}+\cdots+K_{85}.
\end{split}
\end{equation}

Recalling $\textbf{u}=\varepsilon\widetilde{\textbf{u}} +\varepsilon^2\textbf{u}^{\varepsilon}_R$ in \eqref{e3}, we have
\begin{equation}\label{equ63}
\begin{split}
|K_{83}|\leq & \left|\langle \partial_x^{\alpha}{\Delta}\phi^{\varepsilon}_R, \frac{1}{n}[\partial_x^{\alpha}, {\textbf{u}}]{\nabla}n^{\varepsilon}_R\rangle_0\right|\\
\leq & C\|\partial_x^{\alpha}{\Delta}\phi^{\varepsilon}_R\|_{L^2}^2 \{\|\nabla{n}^{\varepsilon}_R\|_{H^{k-1}} \|\nabla\textbf{u}\|_{L^{\infty}} +\|\nabla{n}^{\varepsilon}_R\|_{L^{\infty}} \|\nabla\textbf{u}\|_{H^{k-1}}\}\\
\leq & C\varepsilon\|\partial_x^{\alpha} {\Delta}\phi^{\varepsilon}_R\|_{L^2} \{(1+\varepsilon\|\nabla\textbf{u}^{\varepsilon}_R\|_{L^{\infty}}) \|n^{\varepsilon}_R\|_{H^3} +(1+\varepsilon\|\textbf{u}^{\varepsilon}_R\|_{H^k}) \|n^{\varepsilon}_R\|_{H^k}\}\\
\leq & C\varepsilon^2\|{\Delta}\phi^{\varepsilon}_R\|_{H^{s'}}^2 +C(1+\varepsilon^2\|\textbf{u}^{\varepsilon}_R\|_{H^{s'}}^2) \|n^{\varepsilon}_R\|_{H^{s'}}^2,
\end{split}
\end{equation}
where $k\leq s'$ and $s'\geq 4$. Here, we have used the fact that $1/2\leq n\leq 3/2$ is bounded from above and below by \eqref{e25}.

The estimate for $K_{84}$ and $K_{85}$ are similar to $K_{83}$. From \eqref{e27} and \eqref{e28}, we know
\begin{equation}\label{equ64}
\begin{split}
|K_{84}|+|K_{85}|\leq & C\varepsilon^2\|{\Delta}\phi^{\varepsilon}_R\|_{H^{s'}}^2 + C(1+\varepsilon^2\|\textbf{u}^{\varepsilon}_R\|_{H^{s'}}^2) \|(n^{\varepsilon}_R,\textbf{u}^{\varepsilon}_R, \phi^{\varepsilon}_R)\|_{H^{s'}}^2,
\end{split}
\end{equation}
where $s'\geq4$.
%

Adding \eqref{equ63} and \eqref{equ64} together and using Lemma \ref{L1}, we have
 \begin{equation}\label{equ66}
\begin{split}
|K_{83}|+|K_{84}|+|K_{85}|\leq &C(1+\varepsilon^2\|\textbf{u}^{\varepsilon}_R\|^2_{H^{s'}}) |\!|\!|(\textbf{u}^{\varepsilon}_R,\phi^{\varepsilon}_R)|\!|\!|^2_{s'},
\end{split}
\end{equation}
where $|\!|\!|(\textbf{u}^{\varepsilon}_R,\phi^{\varepsilon}_R)|\!|\!|^2_{s'}$ is defined in \eqref{tri-norm}.

The proof of Proposition \ref{prop2} is complete once $K_{81}$ and $K_{82}$ are estimated in the following Lemma \ref{L6} and Lemma \ref{L7}.
\end{proof}

\begin{lemma}\label{L6}
Let $(n^{\varepsilon}_R,\textbf{u}^{\varepsilon}_R,\phi^{\varepsilon}_R)$ be a solution to \eqref{re}, then
\begin{equation*}
\begin{split}
|K_{81}|\leq & C_1(\sqrt{\varepsilon}\|\phi^{\varepsilon}_R\|_{H^{3}}) (1+\varepsilon^2|\!|\!|(\textbf{u}^{\varepsilon}_R, \phi^{\varepsilon}_R)|\!|\!|_{s'}^2) |\!|\!|\phi^{\varepsilon}_R|\!|\!|_{s'}^2,
\end{split}
\end{equation*}
where $|\!|\!|(\textbf{u}^{\varepsilon}_R, \phi^{\varepsilon}_R)|\!|\!|_{s'}^2$ is given in \eqref{tri-norm} and $K_{81}$ is given in \eqref{equ67}.
\end{lemma}
\begin{proof}
Let $\alpha$ be the multi-index in Proposition \ref{prop2}. Taking $\partial_x^{\alpha}$ of \eqref{re-3}, we have
\begin{equation*}
\begin{split}
\partial_x^{\alpha}n^{\varepsilon}_R =\partial_x^{\alpha}\phi^{\varepsilon}_R -\varepsilon{\Delta}\partial_x^{\alpha}\phi^{\varepsilon}_R +\varepsilon\partial_x^{\alpha}(\phi^{(1)}\phi^{\varepsilon}_R) +\varepsilon^{3/2}\partial_x^{\alpha}R_{\phi}.
\end{split}
\end{equation*}
$K_{81}$ in \eqref{equ67} is then divided into
\begin{equation*}
\begin{split}
K_{81}=&-\langle\partial_x^{\alpha}{\Delta}\phi^{\varepsilon}_R, \frac{\textbf{e}_1-\textbf{u}}{n} \cdot{\nabla}\partial_x^{\alpha}n^{\varepsilon}_R\rangle_{0}\\ =&-\langle\partial_x^{\alpha}{\Delta}\phi^{\varepsilon}_R, \frac{\textbf{e}_1-\textbf{u}}{n} \cdot{\nabla}\partial_x^{\alpha}\phi^{\varepsilon}_R\rangle_0 +\varepsilon\langle\partial_x^{\alpha}{\Delta}\phi^{\varepsilon}_R, \frac{\textbf{e}_1-\textbf{u}}{n} \cdot{\nabla}{\Delta}\partial_x^{\alpha} \phi^{\varepsilon}_R\rangle_0\\
&-\varepsilon\langle\partial_x^{\alpha}{\Delta}\phi^{\varepsilon}_R, \frac{\textbf{e}_1-\textbf{u}}{n} \cdot{\nabla}\partial_x^{\alpha} (\phi^{(1)}\phi^{\varepsilon}_R)\rangle_0 -\varepsilon^{3/2}\langle\partial_x^{\alpha} {\Delta}\phi^{\varepsilon}_R, \frac{\textbf{e}_1-\textbf{u}}{n}\cdot{\nabla} \partial_x^{\alpha}R_{\phi}\rangle_0\\
=&:K_{811}+K_{812}+K_{813}+K_{814},
\end{split}
\end{equation*}
where $|\alpha|=k\leq s'$ for $s'\geq 3$.

\emph{Estimate of $K_{811}$.} By integrating by parts twice, we have
\begin{equation}\label{e30}
\begin{split}
K_{811}=& -\langle\partial_x^{\alpha}{\nabla}\cdot {\nabla}\phi^{\varepsilon}_R, \frac{\textbf{e}_1-\textbf{u}}{n}\partial_x^{\alpha} {\nabla}\phi^{\varepsilon}_R\rangle_0\\
=& \langle\partial_x^{\alpha}{\nabla}\phi^{\varepsilon}_R, \frac{\textbf{e}_1-\textbf{u}}{n}\partial_x^{\alpha} {\nabla}^2\phi^{\varepsilon}_R\rangle_0 +\langle\partial_x^{\alpha}{\nabla}\phi^{\varepsilon}_R, {\nabla}(\frac{\textbf{e}_1-\textbf{u}}{n})\partial_x^{\alpha} {\nabla}\phi^{\varepsilon}_R\rangle_0\\
=& -\frac12\langle\partial_x^{\alpha}{\nabla}\phi^{\varepsilon}_R, {\nabla}\cdot(\frac{\textbf{e}_1-\textbf{u}}{n})\partial_x^{\alpha} {\nabla}\phi^{\varepsilon}_R\rangle_0 +\langle\partial_x^{\alpha}{\nabla}\phi^{\varepsilon}_R, {\nabla}(\frac{\textbf{e}_1-\textbf{u}}{n})\partial_x^{\alpha} {\nabla}\phi^{\varepsilon}_R\rangle_0
\end{split}
\end{equation}
Recalling \eqref{e3}, we have
\begin{equation}\label{e29}
\begin{split}
\|{\nabla}(\frac{\textbf{e}_1-\textbf{u}}{n})\|_{L^{\infty}}, \|{\nabla}\cdot(\frac{\textbf{e}_1-\textbf{u}}{n})\|_{L^{\infty}}\leq C\varepsilon(1+\varepsilon\|\nabla n^{\varepsilon}_R\|_{L^{\infty}} +\varepsilon\|\nabla\textbf{u}^{\varepsilon}_R\|_{L^{\infty}}).
\end{split}
\end{equation}
Therefore, using Lemma \ref{L1}, we obtain
\begin{equation}\label{equ74}
\begin{split}
|K_{811}|\leq & C\varepsilon(1+\varepsilon(\|n^{\varepsilon}_R\|_{H^3} +\|\textbf{u}^{\varepsilon}_R\|_{H^3})) \|\overline{\nabla}\partial_x^{\alpha}\phi^{\varepsilon}_R\|^2\\
\leq & C(1+\varepsilon|\!|\!|(\textbf{u}^{\varepsilon}_R, \phi^{\varepsilon}_R)|\!|\!|_{s'}) |\!|\!|\phi^{\varepsilon}_R|\!|\!|_{s'}^2,
\end{split}
\end{equation}
where $s'\geq 3$.

\emph{Estimate of $K_{812}$.} By integrating by parts, we obtain
\begin{equation}
\begin{split}
K_{812}=  -\frac{\varepsilon}{2}\langle\partial_x^{\alpha} {\Delta}\phi^{\varepsilon}_R, {\nabla}\cdot\frac{\textbf{e}_1-\textbf{u}}{n} {\Delta}\partial_x^{\alpha}\phi^{\varepsilon}_R\rangle_0.
\end{split}
\end{equation}
Using \eqref{e29}, we obtain
\begin{equation}\label{equ71}
\begin{split}
|K_{812}|\leq & C\varepsilon^2(1+\varepsilon(\|(n^{\varepsilon}_R, \textbf{u}^{\varepsilon}_R)\|_{H^3})) \|{\Delta}\partial_x^{\alpha}\phi^{\varepsilon}_R\|^2\\
\leq & C(1+\varepsilon|\!|\!|(\textbf{u}^{\varepsilon}_R, \phi^{\varepsilon}_R)|\!|\!|_{s'}) |\!|\!|\phi^{\varepsilon}_R|\!|\!|_{s'}^2,
\end{split}
\end{equation}
thanks to Lemma \ref{L1}, where $s'\geq 3$.

\emph{Estimate of $K_{813}$.} By using the commutator and integrating by parts twice, we have
\begin{equation*}\label{equ73}
\begin{split}
K_{813}= & -\varepsilon\langle\partial_x^{\alpha}{\Delta}\phi^{\varepsilon}_R, \frac{(\textbf{e}_1-\textbf{u})\phi^{(1)}}{n} \cdot{\nabla}\partial_x^{\alpha}\phi^{\varepsilon}_R\rangle_0 -\varepsilon\langle\partial_x^{\alpha}{\Delta}\phi^{\varepsilon}_R, \frac{(\textbf{e}_1-\textbf{u})}{n} \cdot[{\nabla}\partial_x^{\alpha},\phi^{(1)}] \phi^{\varepsilon}_R\rangle_0\\
= & -\frac{\varepsilon}{2}\langle\partial_x^{\alpha}{\nabla} \phi^{\varepsilon}_R, {\nabla}\cdot(\frac{(\textbf{e}_1-\textbf{u})\phi^{(1)}}{n}) {\nabla}\partial_x^{\alpha}\phi^{\varepsilon}_R\rangle_0 +{\varepsilon}\langle\partial_x^{\alpha}{\nabla} \phi^{\varepsilon}_R, {\nabla}(\frac{(\textbf{e}_1-\textbf{u})\phi^{(1)}}{n}) {\nabla}\partial_x^{\alpha}\phi^{\varepsilon}_R\rangle_0\\
& -\varepsilon\langle\partial_x^{\alpha}{\Delta}\phi^{\varepsilon}_R, \frac{(\textbf{e}_1-\textbf{u})}{n} \cdot[{\nabla}\partial_x^{\alpha},\phi^{(1)}] \phi^{\varepsilon}_R\rangle_0\\
=& : K_{8131}+K_{8132}+K_{813}.
\end{split}
\end{equation*}
Since
\begin{equation}
\begin{split}
\|{\nabla} (\frac{(\textbf{e}_1-\textbf{u})\phi^{(1)}}{n})\|_{L^{\infty}}, \|{\nabla}\cdot(\frac{(\textbf{e}_1-\textbf{u})\phi^{(1)}} {n})\|_{L^{\infty}}\leq C(1+\varepsilon^2\|(\nabla n^{\varepsilon}_R,\nabla\textbf{u}^{\varepsilon}_R)\|_{L^{\infty}}),
\end{split}
\end{equation}
we obtain
\begin{equation}
\begin{split}
|K_{8131}|+|K_{8132}|\leq & C\varepsilon(1+\varepsilon^2(\|n^{\varepsilon}_R\|_{H^3} +\|\textbf{u}^{\varepsilon}_R\|_{H^3})) \|{\nabla}\partial_x^{\alpha}\phi^{\varepsilon}_R\|^2\\
\leq & C(1+\varepsilon|\!|\!|(\textbf{u}^{\varepsilon}_R, \phi^{\varepsilon}_R)|\!|\!|_{s'}) |\!|\!|\phi^{\varepsilon}_R|\!|\!|_{s'}^2.
\end{split}
\end{equation}
From \eqref{e26}, we know that
\begin{equation}
\begin{split}
|K_{8133}|\leq & C\varepsilon^2\|{\Delta}\partial_x^{\alpha} \phi^{\varepsilon}_R\|^2+C\|[{\nabla}\partial_x^{\alpha},\phi^{(1)}] \phi^{\varepsilon}_R\|_{L^2}^2\\
\leq & C\varepsilon^2\|{\Delta}\partial_x^{\alpha} \phi^{\varepsilon}_R\|^2+C\|\phi^{\varepsilon}_R\|_{H^{s'}}^2,
\end{split}
\end{equation}
where we have used \eqref{e31}. Therefore,
\begin{equation}
\begin{split}
|K_{813}| \leq & C(1+\varepsilon|\!|\!|(\textbf{u}^{\varepsilon}_R, \phi^{\varepsilon}_R)|\!|\!|_{s'}) |\!|\!|\phi^{\varepsilon}_R|\!|\!|_{s'}^2.
\end{split}
\end{equation}

\emph{Estimate of $K_{814}$.} Using \eqref{e52-1} in Lemma \ref{lem-A}, we have
\begin{equation}\label{equ68}
\begin{split}
|K_{814}|\leq & C\varepsilon^{3/2}\|{\Delta}\partial_x^{\alpha}\phi^{\varepsilon}_R\| \|\partial_x^{\alpha}{\nabla}R_{\phi}\|\\
\leq & C\varepsilon^2\|{\Delta}\phi^{\varepsilon}_R\|_{H^{s'}}^2 +\varepsilon C_1(\sqrt{\varepsilon}\|\phi^{\varepsilon}_R\|_{H^{s'}}) \|{\nabla}\phi^{\varepsilon}_R\|_{H^{s'}}^2,
\end{split}
\end{equation}
where we have used $|\alpha|=k\leq s'$ for $s'\geq3$.

Adding \eqref{equ74}-\eqref{equ68} together, we have
\begin{equation*}
\begin{split}
\left|K_{81}\right|\leq & C_1(\sqrt{\varepsilon}\|\phi^{\varepsilon}_R\|_{H^{3}})
(1+\varepsilon^2|\!|\!|(\textbf{u}^{\varepsilon}_R, \phi^{\varepsilon}_R)|\!|\!|_{s'}^2) |\!|\!|\phi^{\varepsilon}_R|\!|\!|_{H^{s'}}^2,
\end{split}
\end{equation*}
where $|\!|\!|(\textbf{u}^{\varepsilon}_R, \phi^{\varepsilon}_R)|\!|\!|_{s'}^2$ is given in \eqref{tri-norm}. The proof of Lemma \ref{L6} is complete.
\end{proof}

\begin{lemma}\label{L7}
Let $(n^{\varepsilon}_R,\textbf{u}^{\varepsilon}_R,\phi^{\varepsilon}_R)$ be a solution to \eqref{re} in 3D, then
\begin{equation}
\begin{split}
K_{821}\leq & -\frac12\frac{d}{dt}\langle\partial_x^{\alpha}\nabla\phi^{\varepsilon}_R, \frac{\varepsilon}{n}\partial_x^{\alpha}\nabla\phi^{\varepsilon}_R\rangle_0 -\frac12\frac{d}{dt}\langle\partial_x^{\alpha}\Delta\phi^{\varepsilon}_R, \frac{\varepsilon^2}{n}\partial_x^{\alpha}\Delta\phi^{\varepsilon}_R\rangle_0\\
&-\frac12\frac{d}{dt}\langle\partial_x^{\alpha}\nabla \phi^{\varepsilon}_R, (\frac{\varepsilon^2\phi^{(1)}}{n}) \partial_x^{\alpha} \nabla\phi^{\varepsilon}_R\rangle_0 -\frac12\frac{d}{dt}\langle\partial_x^{\alpha}\nabla\phi^{\varepsilon}_R, \frac{\varepsilon^3\phi^{\varepsilon}_R}{n} \partial_x^{\alpha}\nabla\phi^{\varepsilon}_R\rangle_0\\
&+C(C_1(\sqrt{\varepsilon}\tilde C)+\varepsilon^2|\!|\!|(\textbf{u}^{\varepsilon}_R, \phi^{\varepsilon}_R)|\!|\!|_{3}^2) \{1+|\!|\!|(\textbf{u}^{\varepsilon}_R, \phi^{\varepsilon}_R)|\!|\!|_{s'}^2\},
\end{split}
\end{equation}
where $|\alpha|=s'$ for any $s'\geq 4$, $K_{82}$ is given in \eqref{equ67} and $|\!|\!|(\textbf{u}^{\varepsilon}_R,\phi^{\varepsilon}_R)|\!|\!|_{s'}^2$ is defined in \eqref{tri-norm}.
\end{lemma}
\begin{proof}
Recall that $\overline{\nabla}=\nabla, \overline{\Delta}=\Delta$ in 3D and $K_{82}$ is defined in \eqref{equ67}
\begin{equation}
\begin{split}
K_{82}=\langle\partial_x^{\alpha}{\Delta}\phi^{\varepsilon}_R, \frac{\varepsilon}{n}\partial_t\partial_x^{\alpha}n^{\varepsilon}_R \rangle_{0}.
\end{split}
\end{equation}
Taking $\partial_x^{\alpha}$ of \eqref{re-3}, we have
\begin{equation}
\begin{split}
\partial_x^{\alpha}n^{\varepsilon}_R =\partial_x^{\alpha}\phi^{\varepsilon}_R -\varepsilon{\Delta}\partial_x^{\alpha}\phi^{\varepsilon}_R +\varepsilon\partial_x^{\alpha}(\phi^{(1)}\phi^{\varepsilon}_R) +\frac{\varepsilon^2}{2}\partial_x^{\alpha}[(\phi^{\varepsilon}_R)^2] +\varepsilon^{2}\partial_x^{\alpha}\overline{R}_{\phi}.
\end{split}
\end{equation}
Inserting this into $K_{82}$, we have
\begin{equation}\label{e58}
\begin{split}
K_{82}=&\langle\partial_x^{\alpha}\Delta\phi^{\varepsilon}_R, \frac{\varepsilon}{n}\partial_t\partial_x^{\alpha}\phi^{\varepsilon}_R\rangle_0 -\varepsilon\langle\partial_x^{\alpha}\Delta\phi^{\varepsilon}_R, \frac{\varepsilon}{n}\partial_t{\Delta}\partial_x^{\alpha} \phi^{\varepsilon}_R\rangle_0\\
&+\varepsilon\langle\partial_x^{\alpha}\Delta\phi^{\varepsilon}_R, \frac{\varepsilon}{n}\partial_t\partial_x^{\alpha} (\phi^{(1)}\phi^{\varepsilon}_R)\rangle_0 +\frac{\varepsilon^2}{2}\langle\partial_x^{\alpha}\Delta \phi^{\varepsilon}_R, \frac{\varepsilon}{n}\partial_t\partial_x^{\alpha} [(\phi^{\varepsilon}_R)^2]\rangle_0 \\
&+\varepsilon^{2}\langle\partial_x^{\alpha}\Delta\phi^{\varepsilon}_R, \frac{\varepsilon}{n}\partial_t\partial_x^{\alpha}\overline{R}_{\phi}\rangle_0\\
=&:K_{821}+K_{822}+K_{823}+K_{824}+K_{825}.
\end{split}
\end{equation}

\emph{Estimate of $K_{821}$.} By integration by parts, we obtain
\begin{equation}\label{e54}
\begin{split}
K_{821}=&-\langle\partial_x^{\alpha}\nabla\phi^{\varepsilon}_R, \frac{\varepsilon}{n} \partial_t\partial_x^{\alpha}\nabla\phi^{\varepsilon}_R\rangle_0 -\langle\partial_x^{\alpha}\nabla\phi^{\varepsilon}_R, \nabla(\frac{\varepsilon}{n}) \partial_t\partial_x^{\alpha}\phi^{\varepsilon}_R\rangle_0\\
=&: K_{8211}+K_{8212}.
\end{split}
\end{equation}
By integrating in time, we obtain
\begin{equation}
\begin{split}
K_{8211}=-\frac12\frac{d}{dt} \langle\partial_x^{\alpha}\nabla\phi^{\varepsilon}_R, \frac{\varepsilon}{n} \partial_x^{\alpha}\nabla\phi^{\varepsilon}_R\rangle_0 +\frac12\langle\partial_x^{\alpha}\nabla\phi^{\varepsilon}_R, \partial_t(\frac{\varepsilon}{n}) \partial_x^{\alpha}\nabla\phi^{\varepsilon}_R\rangle_0.
\end{split}
\end{equation}
From \eqref{equ41}, we then have
\begin{equation}\label{e55}
\begin{split}
K_{8211}\leq &-\frac12\frac{d}{dt}\langle\partial_x^{\alpha}\nabla\phi^{\varepsilon}_R, \frac{\varepsilon}{n}\partial_x^{\alpha}\nabla\phi^{\varepsilon}_R\rangle_0 +C(1+\varepsilon^2\|\partial_tn^{\varepsilon}_R\|_{L^{\infty}}) (\varepsilon\|\partial_x^{\alpha}\nabla\phi^{\varepsilon}_R\|^2)\\
\leq & -\frac12\frac{d}{dt}\langle\partial_x^{\alpha}\nabla\phi^{\varepsilon}_R, \frac{\varepsilon}{n}\partial_x^{\alpha}\nabla\phi^{\varepsilon}_R\rangle_0 +C(1+\varepsilon|\!|\!|(\textbf{u}^{\varepsilon}_R, \phi^{\varepsilon}_R)|\!|\!|_{3}) |\!|\!|\phi^{\varepsilon}_R|\!|\!|_{s'}^2,
\end{split}
\end{equation}
where $|\!|\!|\cdot|\!|\!|_{s'}$ is defined in \eqref{tri-norm} and we have used \eqref{equ19} in Corollary \ref{L2}.

Recalling \eqref{e3}, we have
\begin{equation}\label{e60}
\begin{split}
\|\nabla(\frac{\varepsilon}{n})\|_{L^{\infty}}\leq C\varepsilon^2(1+\varepsilon\|\nabla n^{\varepsilon}_R\|_{L^{\infty}}).
\end{split}
\end{equation}
It the follows from \eqref{e54} that
\begin{equation}\label{e56}
\begin{split}
K_{8212}\leq & C\varepsilon^2(1+\varepsilon\|\nabla n^{\varepsilon}_R\|_{L^{\infty}}) \|\partial_x^{\alpha}\nabla\phi^{\varepsilon}_R\|_{L^2} \|\partial_t\partial_x^{\alpha}\phi^{\varepsilon}_R\|_{L^2}\\
\leq & C(1+\varepsilon\|n^{\varepsilon}_R\|_{H^3}) \{\varepsilon\|\nabla\phi^{\varepsilon}_R\|_{H^k}^2 +\varepsilon \|\varepsilon\partial_t\nabla\phi^{\varepsilon}_R\|_{H^{k-1}}^2\}\\
\leq & C(1+\varepsilon|\!|\!|\phi^{\varepsilon}_R|\!|\!|_{3}) |\!|\!|(\textbf{u}^{\varepsilon}_R,\phi^{\varepsilon}_R)|\!|\!|_{s'}^2,
\end{split}
\end{equation}
where we have used Lemma \ref{L1}, Corollary \ref{cor2} and $k\leq s'$.

Therefore, by adding \eqref{e55} and \eqref{e56}, we obtain
\begin{equation}\label{e57}
\begin{split}
K_{821}\leq &-\frac12\frac{d}{dt}\langle\partial_x^{\alpha}\nabla\phi^{\varepsilon}_R, \frac{\varepsilon}{n}\partial_x^{\alpha}\nabla\phi^{\varepsilon}_R\rangle_0 +C(1+\varepsilon^2\|\partial_tn^{\varepsilon}_R\|_{L^{\infty}}) (\varepsilon\|\partial_x^{\alpha}\nabla\phi^{\varepsilon}_R\|^2)\\
\leq & -\frac12\frac{d}{dt}\langle\partial_x^{\alpha}\nabla\phi^{\varepsilon}_R, \frac{\varepsilon}{n}\partial_x^{\alpha}\nabla\phi^{\varepsilon}_R\rangle_0 +C(1+\varepsilon|\!|\!|(\textbf{u}^{\varepsilon}_R, \phi^{\varepsilon}_R)|\!|\!|_{3}) |\!|\!|(\textbf{u}^{\varepsilon}_R,\phi^{\varepsilon}_R)|\!|\!|_{s'}^2.
\end{split}
\end{equation}

\emph{Estimate of $K_{822}$.} Recall that in \eqref{e58}
\begin{equation}
\begin{split}
K_{822}=-\varepsilon\langle\partial_x^{\alpha}\Delta\phi^{\varepsilon}_R, \frac{\varepsilon}{n}\partial_t{\Delta}\partial_x^{\alpha} \phi^{\varepsilon}_R\rangle_0.
\end{split}
\end{equation}
By integration by parts in time, we obtain
\begin{equation}
\begin{split}
K_{822}=-\frac12\frac{d}{dt}\langle\partial_x^{\alpha}\Delta \phi^{\varepsilon}_R, \frac{\varepsilon^2}{n}\partial_x^{\alpha}{\Delta} \phi^{\varepsilon}_R\rangle_0 +\frac12\langle\partial_x^{\alpha}\Delta \phi^{\varepsilon}_R, \partial_t(\frac{\varepsilon^2}{n})\partial_x^{\alpha}{\Delta} \phi^{\varepsilon}_R\rangle_0.
\end{split}
\end{equation}
Similar to \eqref{e55}, it is estimated that
\begin{equation}\label{e59}
\begin{split}
K_{822}\leq & -\frac12\frac{d}{dt}\langle\partial_x^{\alpha}\Delta\phi^{\varepsilon}_R, \frac{\varepsilon^2}{n}\partial_x^{\alpha}\Delta\phi^{\varepsilon}_R\rangle_0 +C(1+\varepsilon|\!|\!|(\textbf{u}^{\varepsilon}_R, \phi^{\varepsilon}_R)|\!|\!|_{3}) |\!|\!|\phi^{\varepsilon}_R|\!|\!|_{s'}^2.
\end{split}
\end{equation}

\emph{Estimate of $K_{823}$.} Recall that in \eqref{e58}
\begin{equation}
\begin{split}
K_{823}=\langle\partial_x^{\alpha}\Delta\phi^{\varepsilon}_R, \frac{\varepsilon^2}{n}\partial_t\partial_x^{\alpha} (\phi^{(1)}\phi^{\varepsilon}_R)\rangle_0.
\end{split}
\end{equation}
By integrating by parts and using the commutator, we can rewrite
\begin{equation}\label{e66}
\begin{split}
K_{823}= &-\langle\partial_x^{\alpha}\nabla\phi^{\varepsilon}_R, \nabla(\frac{\varepsilon^2}{n})\partial_x^{\alpha} (\phi^{(1)}\partial_t\phi^{\varepsilon}_R)\rangle_0 -\langle\partial_x^{\alpha}\nabla\phi^{\varepsilon}_R, \nabla(\frac{\varepsilon^2}{n})\partial_x^{\alpha} (\partial_t\phi^{(1)}\phi^{\varepsilon}_R)\rangle_0\\
& -\langle\partial_x^{\alpha}\nabla\phi^{\varepsilon}_R, (\frac{\varepsilon^2\phi^{(1)}}{n}) \partial_t\partial_x^{\alpha} \nabla\phi^{\varepsilon}_R\rangle_0 -\langle\partial_x^{\alpha}\nabla\phi^{\varepsilon}_R, (\frac{\varepsilon^2}{n}) [\partial_x^{\alpha}\nabla, \phi^{(1)}] \partial_t\phi^{\varepsilon}_R\rangle_0\\
& -\langle\partial_x^{\alpha}\nabla\phi^{\varepsilon}_R, (\frac{\varepsilon^2}{n}) \partial_x^{\alpha}\nabla (\partial_t\phi^{(1)}\phi^{\varepsilon}_R)\rangle_0\\
=& :K_{8231}+K_{8232}+K_{8233}+K_{8234}+K_{8235}.
\end{split}
\end{equation}

We note that from \eqref{e60}, we have
\begin{equation}
\begin{split}
\|\nabla(\frac{\varepsilon^2}{n})\|_{L^{\infty}}\leq C\varepsilon^3(1+\varepsilon\|\nabla n^{\varepsilon}_R\|_{L^{\infty}}).
\end{split}
\end{equation}
Using the multiplicative estimates \eqref{mul-com} in Lemma \ref{Le-inequ}, we obtain
\begin{equation*}
\begin{split}
K_{8231}\leq & C\varepsilon^3(1+\varepsilon\|\nabla n^{\varepsilon}_R\|_{L^{\infty}}) \|\partial_x^{\alpha}\nabla\phi^{\varepsilon}_R\|_{L^2} \{\|\partial_t\phi^{\varepsilon}_R\|_{H^k}\|\phi^{(1)}\|_{L^{\infty}} +\|\partial_t\phi^{\varepsilon}_R\|_{L^{\infty}}\|\phi^{(1)}\|_{H^k}\}
\end{split}
\end{equation*}
Since
\begin{equation}\label{e70}
\begin{split}
\|\partial_t\phi^{\varepsilon}_R\|_{H^k}\leq \|\partial_t\phi^{\varepsilon}_R\|_{H^{k-1}} +\|\partial_t\nabla\phi^{\varepsilon}_R\|_{H^{k-1}},
\end{split}
\end{equation}
we obtain
\begin{equation}\label{e61}
\begin{split}
K_{8231}\leq & C(1+\varepsilon\|\nabla n^{\varepsilon}_R\|_{L^{\infty}})\\
&\times\{\varepsilon\|\nabla\phi^{\varepsilon}_R\|_{H^k}^2 +\|\varepsilon\partial_t\phi^{\varepsilon}_R\|_{H^{k-1}}^2 +\varepsilon\|\varepsilon\partial_t\nabla\phi^{\varepsilon}_R\|_{H^{k-1}}^2 +\|\varepsilon\partial_t\phi^{\varepsilon}_R\|_{H^2}^2\}\\
\leq & C(1+\varepsilon|\!|\!|\phi^{\varepsilon}_R|\!|\!|_{3}) \{1+ |\!|\!|(\textbf{u}^{\varepsilon}_R,\phi^{\varepsilon}_R)|\!|\!|_{s'}^2\}.
\end{split}
\end{equation}
where we have used Lemma \ref{L1} and Corollary \ref{cor2}.

Similarly, we have
\begin{equation}\label{e62}
\begin{split}
K_{8232} \leq & C(1+\varepsilon|\!|\!|\phi^{\varepsilon}_R|\!|\!|_{3}) (1+|\!|\!|\phi^{\varepsilon}_R|\!|\!|_{s'}^2).
\end{split}
\end{equation}

Now we estimate $K_{8233}$. By integrating by parts in time, we obtain
\begin{equation}\label{e63}
\begin{split}
K_{8233}=& -\frac12\frac{d}{dt}\langle\partial_x^{\alpha}\nabla \phi^{\varepsilon}_R, (\frac{\varepsilon^2\phi^{(1)}}{n}) \partial_x^{\alpha} \nabla\phi^{\varepsilon}_R\rangle_0 +\frac12\langle\partial_x^{\alpha}\nabla \phi^{\varepsilon}_R, \partial_t(\frac{\varepsilon^2\phi^{(1)}}{n}) \partial_x^{\alpha} \nabla\phi^{\varepsilon}_R\rangle_0\\
\leq & -\frac12\frac{d}{dt}\langle\partial_x^{\alpha}\nabla \phi^{\varepsilon}_R, (\frac{\varepsilon^2\phi^{(1)}}{n}) \partial_x^{\alpha} \nabla\phi^{\varepsilon}_R\rangle_0 +C(1+\varepsilon|\!|\!|(\textbf{u}^{\varepsilon}_R, \phi^{\varepsilon}_R)|\!|\!|_{3}) |\!|\!|\phi^{\varepsilon}_R|\!|\!|_{s'}^2.
\end{split}
\end{equation}

For the term $K_{8234}$ in \eqref{e58}, we have
\begin{equation*}
\begin{split}
K_{8234}=& -\langle\partial_x^{\alpha}\nabla\phi^{\varepsilon}_R, (\frac{\varepsilon^2}{n}) [\partial_x^{\alpha}\nabla, \phi^{(1)}] \partial_t\phi^{\varepsilon}_R\rangle_0\\
\leq & C\varepsilon^2\|\partial_x^{\alpha}\nabla\phi^{\varepsilon}_R\|_{L^2} \{\|\partial_t\phi^{\varepsilon}_R\|_{H^{k}}\|\phi^{(1)}\|_{L^{\infty}} +\|\partial_t\phi^{\varepsilon}_R\|_{L^{\infty}}\|\phi^{(1)}\|_{H^{k+1}} \}\\
\leq & C\varepsilon\|\partial_x^{\alpha}\nabla\phi^{\varepsilon}_R\|_{L^2}^2 +C\{\varepsilon\|\varepsilon\partial_t\phi^{\varepsilon}_R\|_{H^{k}}^2 +\varepsilon\|\varepsilon\partial_t\phi^{\varepsilon}_R\|_{H^2}^2\}.
\end{split}
\end{equation*}
Similarly to the estimate of $K_{8231}$ in \eqref{e61}, we have
\begin{equation}\label{e64}
\begin{split}
K_{8234}\leq & C(1+\varepsilon|\!|\!|\phi^{\varepsilon}_R|\!|\!|_{3}) \{1+ |\!|\!|(\textbf{u}^{\varepsilon}_R,\phi^{\varepsilon}_R)|\!|\!|_{s'}^2\}.
\end{split}
\end{equation}

For $K_{8235}$, by using multiplicative estimates in \eqref{mul-com}, we have
\begin{equation}\label{e65}
\begin{split}
K_{8235}=&-\langle\partial_x^{\alpha}\nabla\phi^{\varepsilon}_R, (\frac{\varepsilon^2}{n}) \partial_x^{\alpha} (\partial_t\phi^{(1)}\nabla\phi^{\varepsilon}_R)\rangle_0 -\langle\partial_x^{\alpha}\nabla\phi^{\varepsilon}_R, (\frac{\varepsilon^2}{n}) \partial_x^{\alpha} (\nabla \partial_t\phi^{(1)}\phi^{\varepsilon}_R)\rangle_0\\
\leq & C\varepsilon^2\|\partial_x^{\alpha}\nabla\phi^{\varepsilon}_R\|_{L^2} \{\|\nabla\phi^{\varepsilon}_R\|_{H^{k}}\|\partial_t\phi^{(1)}\|_{L^{\infty}} +\|\phi^{\varepsilon}_R\|_{L^{\infty}}\|\partial_t\phi^{(1)}\|_{H^{k+1}}\\
&+\|\phi^{\varepsilon}_R\|_{H^{k}}\|\partial_t\nabla\phi^{(1)}\|_{L^{\infty}} +\|\phi^{\varepsilon}_R\|_{L^{\infty}}\|\partial_t\nabla\phi^{(1)}\|_{H^{k+1}}\}\\
\leq & C\varepsilon\|\partial_x^{\alpha}\nabla\phi^{\varepsilon}_R\|_{L^2}^2 +C\{\varepsilon\|\nabla\phi^{\varepsilon}_R\|_{H^{k}}^2 +\|\phi^{\varepsilon}_R\|_{H^k}^2 +\|\phi^{\varepsilon}_R\|_{H^3}^2\}\\
\leq & C|\!|\!|\phi^{\varepsilon}_R|\!|\!|_{s'}^2,
\end{split}
\end{equation}
where $k\leq s'$ and $s'\geq 3$.

From \eqref{e66}, by adding \eqref{e61} to \eqref{e65} together, we obtain
\begin{equation}\label{e67}
\begin{split}
K_{823}\leq & -\frac12\frac{d}{dt}\langle\partial_x^{\alpha}\nabla \phi^{\varepsilon}_R, (\frac{\varepsilon^2\phi^{(1)}}{n}) \partial_x^{\alpha} \nabla\phi^{\varepsilon}_R\rangle_0\\
&+C(1+\varepsilon|\!|\!|(\textbf{u}^{\varepsilon}_R, \phi^{\varepsilon}_R)|\!|\!|_{3}) \{1+ |\!|\!|(\textbf{u}^{\varepsilon}_R,\phi^{\varepsilon}_R)|\!|\!|_{s'}^2\}.
\end{split}
\end{equation}

\emph{Estimate of $K_{824}$.} Recall that $K_{824}$ is defined in \eqref{e58}
\begin{equation}\label{e68}
\begin{split}
K_{824}=& \langle\partial_x^{\alpha}\Delta \phi^{\varepsilon}_R, \frac{\varepsilon^3}{n}\partial_x^{\alpha} (\phi^{\varepsilon}_R\partial_t\phi^{\varepsilon}_R)\rangle_0\\
=& -\langle\partial_x^{\alpha}\nabla\phi^{\varepsilon}_R, \frac{\varepsilon^3\phi^{\varepsilon}_R}{n} \partial_t\partial_x^{\alpha}\nabla\phi^{\varepsilon}_R\rangle_0 -\langle\partial_x^{\alpha}\nabla\phi^{\varepsilon}_R, \frac{\varepsilon^3}{n} [\partial_x^{\alpha}\nabla,\phi^{\varepsilon}_R] \partial_t\phi^{\varepsilon}_R\rangle_0\\
& -\langle\partial_x^{\alpha}\nabla\phi^{\varepsilon}_R, \frac{\varepsilon^3}{n} \partial_x^{\alpha}(\nabla\phi^{\varepsilon}_R \partial_t\phi^{\varepsilon}_R)\rangle_0\\
=& :K_{8241}+K_{8242}+K_{8243}.
\end{split}
\end{equation}

For the term $K_{8241}$, by integrating by parts in time, we obtain
\begin{equation}\label{e69}
\begin{split}
K_{8241}=& -\frac12\frac{d}{dt}\langle\partial_x^{\alpha}\nabla\phi^{\varepsilon}_R, \frac{\varepsilon^3\phi^{\varepsilon}_R}{n} \partial_x^{\alpha}\nabla\phi^{\varepsilon}_R\rangle_0 +\frac12\langle\partial_x^{\alpha}\nabla\phi^{\varepsilon}_R, \partial_t(\frac{\varepsilon^3\phi^{\varepsilon}_R}{n}) \partial_x^{\alpha}\nabla\phi^{\varepsilon}_R\rangle_0.
\end{split}
\end{equation}
By using H\"older inequality, Lemma \ref{L1}, \ref{L2} and Corollary \ref{cor2}, we obtain
\begin{equation}
\begin{split}
\|\partial_t(\frac{\varepsilon^2\phi^{\varepsilon}_R}{n})\|_{L^{\infty}}
\leq & C\varepsilon^2\|\partial_t\phi^{\varepsilon}_R\|_{L^{\infty}} +C\varepsilon^3\|\phi^{\varepsilon}_R\|_{L^{\infty}} \|\partial_t\widetilde{n}^{\varepsilon}_R\|_{L^{\infty}} +C\varepsilon^4\|\phi^{\varepsilon}_R\|_{L^{\infty}} \|\partial_t{n}^{\varepsilon}_R\|_{L^{\infty}}\\
\leq & C+C\varepsilon^2\|\varepsilon\partial_t\phi^{\varepsilon}_R\|_{H^2}^2 +C\varepsilon^2\|\phi^{\varepsilon}_R\|_{H^2}^2 +C\varepsilon^2\|\varepsilon\partial_t{n}^{\varepsilon}_R\|_{H^2}^2\\
\leq & C+C\varepsilon^2|\!|\!|(\textbf{u}^{\varepsilon}_R, \phi^{\varepsilon}_R)|\!|\!|_{3}^2.
\end{split}
\end{equation}
Therefore, $K_{8241}$ in \eqref{e69} is estimated as
\begin{equation}\label{e72}
\begin{split}
K_{8241}=& -\frac12\frac{d}{dt}\langle\partial_x^{\alpha}\nabla\phi^{\varepsilon}_R, \frac{\varepsilon^3\phi^{\varepsilon}_R}{n} \partial_x^{\alpha}\nabla\phi^{\varepsilon}_R\rangle_0 +C(1+\varepsilon^2|\!|\!|(\textbf{u}^{\varepsilon}_R, \phi^{\varepsilon}_R)|\!|\!|_{3}^2) \{\varepsilon\|\partial_x^{\alpha}\nabla\phi^{\varepsilon}_R\|_{L^2}^2\}.
\end{split}
\end{equation}

For the term $K_{8242}$ in \eqref{e68}, by commutator estimates in \eqref{mul-com}, we have
\begin{equation}
\begin{split}
K_{8242}=& -\langle\partial_x^{\alpha}\nabla\phi^{\varepsilon}_R, \frac{\varepsilon^3}{n} [\partial_x^{\alpha}\nabla,\phi^{\varepsilon}_R] \partial_t\phi^{\varepsilon}_R\rangle_0\\
\leq & C\varepsilon^3\|\partial_x^{\alpha}\nabla\phi^{\varepsilon}_R\|_{L^2} \{\|\partial_t\phi^{\varepsilon}_R\|_{H^k} \|\phi^{\varepsilon}_R\|_{L^{\infty}} +\|\partial_t\phi^{\varepsilon}_R\|_{L^{\infty}} \|\phi^{\varepsilon}_R\|_{H^{k+1}}\}\\
\leq & C\varepsilon\|\nabla\phi^{\varepsilon}_R\|_{H^k}^2 +C(\varepsilon^2\|\phi^{\varepsilon}_R\|_{H^3}^2 +\varepsilon^2\|\varepsilon\partial_t\phi^{\varepsilon}_R\|_{H^2}^2 )(\varepsilon\|\varepsilon\partial_t\phi^{\varepsilon}_R\|_{H^k}^2+ \varepsilon\|\phi^{\varepsilon}_R\|_{H^{k+1}}^2).
\end{split}
\end{equation}
Using \eqref{e70}, Lemma \ref{L1} and Corollary \ref{cor2}, we then have
\begin{equation}\label{e71}
\begin{split}
K_{8242}\leq & C(1+\varepsilon^2|\!|\!|(\textbf{u}^{\varepsilon}_R, \phi^{\varepsilon}_R)|\!|\!|_{3}^2) |\!|\!|(\textbf{u}^{\varepsilon}_R,\phi^{\varepsilon}_R)|\!|\!|_{s'}^2.
\end{split}
\end{equation}

For the term $K_{8243}$ in \eqref{e68}, by multiplicative estimates in \eqref{mul-com}, we have
\begin{equation}
\begin{split}
K_{8243}= & -\langle\partial_x^{\alpha}\nabla\phi^{\varepsilon}_R, \frac{\varepsilon^3}{n} \partial_x^{\alpha}(\nabla\phi^{\varepsilon}_R \partial_t\phi^{\varepsilon}_R)\rangle_0\\
\leq & C\varepsilon^3\|\partial_x^{\alpha}\nabla\phi^{\varepsilon}_R\|_{L^2} \{\|\partial_t\phi^{\varepsilon}_R\|_{H^k} \|\nabla\phi^{\varepsilon}_R\|_{L^{\infty}} +\|\partial_t\phi^{\varepsilon}_R\|_{L^{\infty}} \|\nabla\phi^{\varepsilon}_R\|_{H^{k}}\}\\
\leq & C(1+\varepsilon^2|\!|\!|(\textbf{u}^{\varepsilon}_R, \phi^{\varepsilon}_R)|\!|\!|_{3}^2) |\!|\!|(\textbf{u}^{\varepsilon}_R,\phi^{\varepsilon}_R)|\!|\!|_{s'}^2,
\end{split}
\end{equation}
where in the last inequality, we have used the same estimates as in \eqref{e71}.

From \eqref{e68}, adding \eqref{e69}, \eqref{e72} and \eqref{e71}, we obtain
\begin{equation}\label{e73}
\begin{split}
K_{824}\leq & -\frac12\frac{d}{dt}\langle\partial_x^{\alpha}\nabla\phi^{\varepsilon}_R, \frac{\varepsilon^3\phi^{\varepsilon}_R}{n} \partial_x^{\alpha}\nabla\phi^{\varepsilon}_R\rangle_0 +C\varepsilon(1+\varepsilon^2|\!|\!|(\textbf{u}^{\varepsilon}_R, \phi^{\varepsilon}_R)|\!|\!|_{3}^2) |\!|\!|(\textbf{u}^{\varepsilon}_R,\phi^{\varepsilon}_R)|\!|\!|_{s'}^2.
\end{split}
\end{equation}

\emph{Estimate of $K_{825}$.} Recall that $K_{825}$ is defined in \eqref{e58}. By \eqref{e52-2} in Lemma \ref{lem-A}, we have
\begin{equation}\label{e76}
\begin{split}
K_{825} \leq & C\varepsilon^2\|\partial_x^{\alpha}\Delta\phi^{\varepsilon}_R\|^2 +\varepsilon^{4}\|\partial_t\partial_x^{\alpha}R_{\phi}\|^2\\
\leq & C\varepsilon^2\|\partial_x^{\alpha}\Delta\phi^{\varepsilon}_R\|^2 +C_1(\sqrt{\varepsilon}\|\phi^{\varepsilon}_R\|_{H^{\delta}}) (1+\varepsilon^2\|\varepsilon\partial_t\phi^{\varepsilon}_R\|_{H^k}^2),
\end{split}
\end{equation}
where $\delta=\max\{2,k-1\}$ in Lemma \ref{lem-A}. Furthermore, when $0<\varepsilon<\varepsilon_1$,
\begin{equation}
\begin{split}
\varepsilon^2\|\varepsilon\partial_t\phi^{\varepsilon}_R\|_{H^k}^2 \leq & \varepsilon^2\|\varepsilon\partial_t\nabla\phi^{\varepsilon}_R\|_{H^{k-1}}^2 +\varepsilon^2\|\varepsilon\partial_t\phi^{\varepsilon}_R\|_{H^{k-1}}^2\\
\leq & C(1+|\!|\!|(\textbf{u}^{\varepsilon}_R,\phi^{\varepsilon}_R)|\!|\!|_{s'}^2),
\end{split}
\end{equation}
where we have used \eqref{e75} in Corollary \ref{cor2} in the last inequality. It then follows from \eqref{e76} that
\begin{equation}\label{e77}
\begin{split}
K_{825} \leq & C_1(\sqrt{\varepsilon}\|\phi^{\varepsilon}_R\|_{H^{\delta}}) \{1+|\!|\!|(\textbf{u}^{\varepsilon}_R,\phi^{\varepsilon}_R)|\!|\!|_{s'}^2\},
\end{split}
\end{equation}
where $\delta=\max\{2,k-1\}\leq s'-1$.

The proof of Lemma \ref{L7} is complete by adding \eqref{e57}, \eqref{e59}, \eqref{e67}, \eqref{e73} and \eqref{e77} together.
\end{proof}

\begin{proof}[\textbf{Proof of Theorem \ref{th} for $T_i=0$}]
From \eqref{e25}, there exists some $\varepsilon_1>0$ such that $1/2\leq 1+\varepsilon\phi^{(1)}+\varepsilon^2\phi^{\varepsilon}_R\leq 3/2$. By adding inequalities \eqref{e97} and \eqref{equ61}, then integrating over $[0,t]$ and taking summation over $|\alpha|=k$ for $0\leq k\leq s'$, we obtain
\begin{equation}
\begin{split}
|\!|\!|(\textbf{u}^{\varepsilon}_R,\phi^{\varepsilon}_R)|\!|\!|_{s'}^2
\leq & CC_{\varepsilon}(0)+CC_1\int_0^t(C_1+\varepsilon|\!|\!|(\textbf{u}^{\varepsilon}_R, n^{\varepsilon}_R, \phi^{\varepsilon}_R)|\!|\!|_{s'}^2)\{1 +|\!|\!|(\textbf{u}^{\varepsilon}_R, n^{\varepsilon}_R, \phi^{\varepsilon}_R)|\!|\!|_{s'}^2\}dr,
\end{split}
\end{equation}
where $C_{\varepsilon}(0)=|\!|\!|(\textbf{u}^{\varepsilon}_R, \phi^{\varepsilon}_R)(0)|\!|\!|_{s'}^2$. Recalling \eqref{assumption}, we know that there exists some constant $0<\varepsilon_0<\varepsilon_1$ such that $\varepsilon\|(\textbf{u}^{\varepsilon}_R, n^{\varepsilon}_R, \phi^{\varepsilon}_R)\|_{H^{s'}}^2\leq 1$. Since $C_1=C_1(\sqrt{\varepsilon}\|n^{\varepsilon}_R\|_{H^{s'}})$ and is nondecreasing, we know that $C_1\leq C_1(1)$ when $0<\varepsilon<\varepsilon_0$. Therefore, there exists some constant $C_3>1$ such that
\begin{equation}
\begin{split}
|\!|\!|(\textbf{u}^{\varepsilon}_R, \phi^{\varepsilon}_R)|\!|\!|_{s'}^2\leq C_3C_{\varepsilon}(0)+C_3\int_0^t\{1+|\!|\!|(\textbf{u}^{\varepsilon}_R, \phi^{\varepsilon}_R)|\!|\!|_{s'}^2\}dr.
\end{split}
\end{equation}
On the other hand, from Lemma \ref{L1}, there exists some constant $C_4\geq 1$ such that for any $0<\varepsilon<\varepsilon_0$,
\begin{equation}\label{e98}
\begin{split}
\|n^{\varepsilon}_R\|_{H^{s'}}^2\leq C_4(1+|\!|\!|\phi^{\varepsilon}_R|\!|\!|_{s'}^2).
\end{split}
\end{equation}

Let $C_0'=\sup_{0<\varepsilon<1}C_{\varepsilon}(0)$. Given given $0<\tau_0<\tau_*$, we let $\tilde C$ in \eqref{assumption} satisfy $\tilde C\geq 2C_4(1+C_3C_0')e^{C_3\tau_0}$, then by Gronwall inequality,
\begin{equation}
\begin{split}
\sup_{0\leq t\leq\tau_0}|\!|\!|(\textbf{u}^{\varepsilon}_R,\phi^{\varepsilon}_R)|\!|\!|_{s'}^2 \leq (1+C_3C_0')e^{C_3\tau_0}\leq \tilde C,
\end{split}
\end{equation}
and from \eqref{e98}
\begin{equation}
\begin{split}
\sup_{0\leq t\leq\tau_0}\|n^{\varepsilon}_R\|_{H^{s'}}^2\leq C_4\{1+(1+C_3C_0')e^{C_3\tau_0}\}\leq \tilde C.
\end{split}
\end{equation}
It is then standard to obtain uniform estimates for $|\!|\!|(n^{\varepsilon}_R,\textbf{u}^{\varepsilon}_R, \phi^{\varepsilon}_R)|\!|\!|_{s'}$ independent of $\varepsilon$ by the continuity method. The proof is complete for the case $T_i=0$.
\end{proof}

\renewcommand{\theequation}{\Alph{section}.\arabic{equation}}

\appendix
\section{Proof of Proposition \ref{prop-kp3} and Lemma \ref{lem-A}}
\setcounter{equation}{0}
\begin{proof}[Proof of Proposition \ref{prop-kp3}]
We need to derive the remainder system \eqref{rem-kp}. We first consider the remainder equation \eqref{rem-kp-n}. Plugging the expansion of $n$ and $\textbf{u}=(u_1,u_2)$ in \eqref{exp-kp} into the \eqref{k1-n}, we obtain a polynomial equation of $\varepsilon$, whose coefficients depend on $n^{(i)},\ \textbf{u}^{(i)}$, $n^{\varepsilon}_R$ and $\textbf{u}^{\varepsilon}_R$. 
Subtracting $\{\varepsilon\times\eqref{k-order1-n}+ \varepsilon^2\times\eqref{k-order2-n}+ \varepsilon^3\times\eqref{k-order3-n}\}$ from this polynomial, we obtain the following equation
\begin{equation}\label{e45}
\begin{split}
\varepsilon^4\partial_tn^{(3)}+\varepsilon^3\partial_tn^{\varepsilon}_R -\varepsilon^2\partial_xn^{\varepsilon}_R+A+B=0,
\end{split}
\end{equation}
where
\begin{equation*}
\begin{split}
A=&\varepsilon^4\sum_{\substack{i,j\geq1\\ i+j\geq4}}\varepsilon^{i+j-4}\partial_x(n^{(i)}u_1^{(j)}) +\varepsilon^3\{{u_1}^{\varepsilon}_R\partial_x\widetilde n+\partial_x\widetilde{u}_1n^{\varepsilon}_R\} +\varepsilon^2\{n\partial_x{u_1}^{\varepsilon}_R +u_1\partial_xn^{\varepsilon}_R\},\\
B=&\varepsilon^4\sum_{\substack{i,l\geq1\\ i+{l}\geq4}}\varepsilon^{i+l+\frac12-4} \partial_x(n^{(i)}u_2^{(l)}) +\varepsilon^3\{{u_2}^{\varepsilon}_R\partial_x\widetilde n+\partial_x\widetilde{u}_2n^{\varepsilon}_R\} +\varepsilon^2\{n\partial_x{u_2}^{\varepsilon}_R +u_2\partial_xn^{\varepsilon}_R\}.
\end{split}
\end{equation*}
Rearranging and dividing \eqref{e45} by $\varepsilon^3$, we obtain \eqref{rem-kp-n}, where
\begin{equation}\label{e50}
\begin{split}
R_n=\partial_tn^{(3)}+\sum_{\substack{i,j\geq1\\ i+j\geq4}}\varepsilon^{i+j-4}\partial_x(n^{(i)}u_1^{(j)}) +\sum_{\substack{i,l\geq1\\ i+l\geq4}}\varepsilon^{i+l+\frac12-4} \partial_x(n^{(i)}u_2^{(l)}).
\end{split}
\end{equation}

The derivation of \eqref{rem-kp-1} and \eqref{rem-kp-2} is similar. Subtracting $\{\varepsilon\times\eqref{k-order1-1}+ \varepsilon^2\times\eqref{k-order2-1}+ \varepsilon^3\times\eqref{k-order3-1}\}$ from the equation of \eqref{k1-1}, we obtain the remainder equation \eqref{rem-kp-1}. We only derive the remainder terms of the pressure term ${T_i\partial_{x_1}n}/{n}$. After subtracting, we obtain
\begin{equation}\label{e46}
\begin{split}
T_i\frac{\partial_{x_1}n}{n}-&T_i\{\varepsilon \partial_{x_1}n^{(1)} +\varepsilon^2(\partial_{x_1}n^{(2)}-n^{(1)}\partial_{x_1}n^{(1)})\\
&+\varepsilon^3(\partial_{x_1}n^{(3)} +\partial_{x_1}n^{(1)}(\frac12(n^{(1)})^2-n^{(2)}) -\partial_{x_1}n^{(2)}n^{(1)})\}.
\end{split}
\end{equation}
After divided by $\varepsilon^3$, \eqref{e46} can be rearranged into
\begin{equation}\label{e47}
\begin{split}
T_i\frac{\partial_{x_1}n^{\varepsilon}_R}{\varepsilon n} -T_i\frac{p_1}{n}n^{\varepsilon}_R -T_i\frac{\varepsilon R_{T1}}{n},
\end{split}
\end{equation}
where $p_1$ and $R_{T1}$ are finite combinations of $n^{(1)}$, $n^{(2)}$ and $n^{(3)}$ only. The expression of $\textbf{R}_\textbf{u}$ depends only on $\textbf{u}^{(i)}$ and $\phi^{i}$ and can be derived similarly to the derivation of $R_n$ in \eqref{e50}.

The derivation of \eqref{rem-kp-p} is slightly different, where the remainder $R_{\phi}$ depends on $\phi^{\varepsilon}_R$. Recall $\phi=\varepsilon\phi^{(1)}+\varepsilon^2\phi^{(2)}+\varepsilon^3\phi^{(3)} +\varepsilon^2\phi^{\varepsilon}_R$ in \eqref{exp-kp-p}. Consider the Taylor expansion in the integral form
\begin{equation*}
\begin{split}
e^{\phi}=&1+(\varepsilon\widetilde\phi+\varepsilon^2\phi^{\varepsilon}_R) +\frac{1}{2!}(\varepsilon\widetilde\phi+\varepsilon^2\phi^{\varepsilon}_R)^2 +\frac{1}{3!}(\varepsilon\widetilde\phi+\varepsilon^2\phi^{\varepsilon}_R)^3\\
&+\frac{1}{3!}\int_0^1e^{\theta\phi}(1-\theta)^3d\theta (\varepsilon\widetilde\phi+\varepsilon^2\phi^{\varepsilon}_R)^4.
\end{split}
\end{equation*}
Subtracting $\{\varepsilon\times\eqref{k-order1-p}+ \varepsilon^2\times\eqref{k-order2-p}+ \varepsilon^3\times\eqref{k-order3-p}\}$ from \eqref{k1-p}, we have
\begin{equation}\label{e51}
\begin{split}
&\varepsilon^3(\partial_{x_1}^2+\varepsilon\partial_{x_2}^2) \phi^{\varepsilon}_R +\varepsilon^4\{(\partial_{x_1}^2+\varepsilon\partial_{x_2}^2)\phi^{(3)} +\partial_{x_2}^2\phi^{(2)}\} \\
=&\varepsilon^4\hat R_{\phi} +\varepsilon^2\phi^{\varepsilon}_R +\varepsilon^3\widetilde\phi\phi^{\varepsilon}_R +\frac12\varepsilon^4(\phi^{\varepsilon}_R)^2 +\frac12\varepsilon^4(\widetilde\phi)^2\phi^{\varepsilon}_R +\frac12\varepsilon^5\widetilde\phi(\phi^{\varepsilon}_R)^2\\
&+\frac1{3!}\varepsilon^6(\phi^{\varepsilon}_R)^3 +\varepsilon^4\frac{1}{3!}\int_0^1e^{\theta\phi}(1-\theta)^3d\theta (\widetilde\phi+\varepsilon\phi^{\varepsilon}_R)^4 -\varepsilon^2n^{\varepsilon}_R,
\end{split}
\end{equation}
where $\hat R_{\phi}$ depends only $\phi^{(1)},\phi^{(2)}$ and $\phi^{(3)}$. After divided by $\varepsilon^2$, \eqref{e51} can be rewritten in the form
\begin{equation*}
\begin{split}
\varepsilon(\partial_{x_1}^2+\varepsilon\partial_{x_2}^2)\phi^{\varepsilon}_R=\phi^{\varepsilon}_R-n^{\varepsilon}_R +\varepsilon\phi^{(1)}\phi^{\varepsilon}_R +\varepsilon^{3/2}R'_{\phi} +\varepsilon^2R''_{\phi},
\end{split}
\end{equation*}
where $R''_{\phi}=\hat R_{\phi}-\{(\partial_{x_1}^2+\varepsilon\partial_{x_2}^2)\phi^{(3)} +\partial_{x_2}^2\phi^{(2)}\}$ and $R'_{\phi}=F(\sqrt{\varepsilon}\phi^{\varepsilon}_R)\phi^{\varepsilon}_R$ for some function of $F$ depending on $\sqrt{\varepsilon}\phi^{\varepsilon}_R$. Letting $R_{\phi}=R'_{\phi}+\sqrt{\varepsilon}R''_{\phi}$, we obtain \eqref{rem-kp-p}.

From \eqref{e51}, it is obvious that \eqref{rem-kp-p} can be written in an equivalent form
\begin{equation}\label{e34}
\begin{split}
\varepsilon(\partial_{x_1}^2+\varepsilon\partial_{x_2}^2)\phi^{\varepsilon}_R=\phi^{\varepsilon}_R-n^{\varepsilon}_R +\varepsilon\phi^{(1)}\phi^{\varepsilon}_R +\frac{\varepsilon^2}{2}(\phi^{\varepsilon}_R)^2 +\varepsilon^{2}\overline{R}_{\phi},
\end{split}
\end{equation}
where $\overline{R}_{\phi}$ is also of the form of $R_\phi$ and satisfies the same estimates of Lemma \ref{lem-A}.

The proof of Proposition \ref{prop-kp3} is complete.
\end{proof}


\begin{proof}[Proof of Lemma \ref{lem-A}]
We mainly consider the estimate for the integral term in \eqref{e51}, which has an important contribution to the remainder term $R'_{\phi}$, while the other contributions from \eqref{e51} can be estimated similarly. Let $\alpha=0$. By H\"older inequality and Sobolev embedding, we have
\begin{equation}
\begin{split}
\|I\|_{L^2}\leq & Ce^{\|\phi\|_{L^{\infty}}}\|\widetilde\phi+\varepsilon\phi^{\varepsilon}_R\|_{L^{\infty}}^3 \|\widetilde\phi+\varepsilon\phi^{\varepsilon}_R\|_{L^2}\\
\leq & C(\varepsilon\|\phi^{\varepsilon}_R\|_{L^{\infty}}) (\|\widetilde\phi\|_{L^2}+\|\phi^{\varepsilon}_R\|_{L^2})\\
\leq & C(\varepsilon\|\phi^{\varepsilon}_R\|_{H^2}) (1+\|\phi^{\varepsilon}_R\|_{L^2}).
\end{split}
\end{equation}
Similar results can be obtained for $\alpha\geq1$, once we note that $H^2$ is an algebra in $\Bbb R^3$. On the other hand, $R''_{\phi}$ depends only on $\phi^{(1)},\phi^{(2)}$ and $\phi^{(3)}$, $\|R''_{\phi}\|_{H^{k}}\leq C$ for any $0\leq k \leq s$. Therefore, we arrive at the estimate
\begin{align}
\|R_{\phi}\|_{H^{k}}\leq &C(\sqrt{\varepsilon}\|\phi^{\varepsilon}_R\|_{H^{\delta}}) (1+\|\phi^{\varepsilon}_R\|_{H^{k}}), \ \ \ \forall 0\leq k\leq s,
\end{align}
where we have used the fact that a uniform constant $C$ is also of the form $C(\sqrt{\varepsilon}\|\phi^{\varepsilon}_R\|_{H^{\delta}})$. Furthermore, if we let $C_1(r)=\sup_{0\leq s\leq r}C(r)$, the constant $C_1(r)$ is nondecreasing. Then \eqref{e52-1} is proved. The inequality \eqref{e52-2} can be proved similarly.
\end{proof}

\begin{corollary}\label{cor-a}
Let $k\geq 0$ be an integer, then there exists a constant $1\leq C_1=C_1(\sqrt{\varepsilon}\|\phi^{\varepsilon}_R\|_{H^{\delta}})$, such that
\begin{align}
\|\overline{\nabla}R_{\phi}\|_{H^{k}}\leq &C_1(\sqrt{\varepsilon}\|\phi^{\varepsilon}_R\|_{H^{\delta}}) (1+\|\overline{\nabla}\phi^{\varepsilon}_R\|_{H^{k}}), \ \ \text{ and}\label{e81-1}\\
\|\partial_t\overline{\nabla}R_{\phi}\|_{H^{k}}\leq &C_1(\sqrt{\varepsilon}\|\phi^{\varepsilon}_R\|_{H^{\delta}}) (1+\|\partial_t\overline{\nabla}\phi^{\varepsilon}_R\|_{H^{k}}),\label{e81-2}
\end{align}
where $\delta=\max\{2,k-1\}$. Furthermore, the constant $C_1(\cdot)$ can be chosen to be nondecreasing.
\end{corollary}

\section{Derivation of the ZKE}\label{app-b}
\setcounter{equation}{0}
The three dimensional Zakharov-Kuznetsov equation (ZKE) is of the form \cite{ZK74,LS09}
\begin{equation*}
\partial_tu+u\partial_{x_1}u+\partial_{x_1}\Delta u=0,\ \ \ \ x=(x_1,x_2,x_3)\in \Bbb R^3,t\in\Bbb R
\end{equation*}
where $\Delta=\partial_{x_1}^2+\partial_{x_2}^2+\partial_{x_3}^2$. In this appendix, we will derive the Zakharov-Kuznetsov equation (ZKE) from the Euler-Poisson system with static magnetic field,
\begin{equation}\label{equ1}
\begin{cases}
&\partial_tn+\nabla\cdot(n\textbf{u})=0\\
&\partial_t\textbf{u}+\textbf{u}\cdot\nabla\textbf{u} +T_i\frac{\nabla n}{n}+\textbf{e}_1\times\textbf{u}=-\nabla\phi\\
&\Delta\phi=e^{\phi}-n,
\end{cases}
\end{equation}
where $n(t,x),\textbf{u}(t,x)=(u_1(t,x),u_2(t,x),u_3(t,x))$ and $\phi(t,x)$ are respectively the density, velocity of the ions and the electric potential at time $t\geq 0$, position $x=(x_1,x_2,x_3)\in\Bbb R^3$. Here $\textbf{e}_1=(1,0,0)^T$ is the constant magnetic direction and $T_i\geq0$ is the ion temperature. 
The formal derivation of the ZKE when $T_i=0$ can also be found in \cite{LS09}.

\subsection{Formal expansion}
Consider the following Gardner-Morikawa transformation in \eqref{equ1}
\begin{equation}\label{e100}
\begin{split}
\varepsilon^{1/2}(x_1-Vt)\to x_1,\ \ \varepsilon^{1/2}x_2\to x_2,\ \ \varepsilon^{1/2}x_3\to x_3,\ \ \varepsilon^{3/2}t\to t.
\end{split}
\end{equation}
We obtain the parameterized system
\begin{subequations}\label{equ2}
\begin{numcases}{}
\varepsilon\partial_tn-V\partial_{x_1}n+\nabla\cdot(n\textbf{u})=0\label{equ2-n}\\
\varepsilon\partial_t\textbf{u}-V\partial_{x_1}\textbf{u} +\textbf{u}\cdot\nabla\textbf{u}+T_i\frac{\nabla n}{n}=-\nabla\phi +\frac{1}{\varepsilon^{1/2}}\textbf{u}\times{e_1}\label{equ2-u}\\
\varepsilon\Delta\phi=e^{\phi}-n,\label{equ2-p}
\end{numcases}
\end{subequations}
where $\varepsilon$ denotes the amplitude of the initial disturbance and is assumed to be small compared with unity and $V$ is the wave speed to be determined. We consider the following formal expansion
\begin{subequations}\label{formal}
\begin{numcases}{}
\ n=1+\varepsilon n^{(1)}\ \ \ \ \ \ \ \ \ \ \ \ \ \ \ \ \ +\varepsilon^2n^{(2)}\ \ \ \ \ \ \ \ \ \ \ \ \ \ \ \ \ +\varepsilon^3n^{(3)}\ \ \ \ \ \ \ \ \ \ \ \ \ \ \ \ \ +\cdots,\label{formal-n}\\
u_1=\ \ \ \ \ \varepsilon u_1^{(1)}\ \ \ \ \ \ \ \ \ \ \ \ \ \ \ \ \ \ +\varepsilon^2u_1^{(2)}\ \ \ \ \ \ \ \ \ \ \ \ \ \ \ \ \ +\varepsilon^3u_1^{(3)}\ \ \ \ \ \ \ \ \ \ \ \ \ \ \ \ \ +\cdots,\label{formal-1}\\
\ \phi=\ \ \ \ \ \ \varepsilon\phi^{(1)}\ \ \ \ \ \ \ \ \ \ \ \ \ \ \ \ \  +\varepsilon^2\phi^{(2)}\ \ \ \ \ \ \ \ \ \ \ \ \ \ \ \ \ +\varepsilon^3\phi^{(3)}\ \ \ \ \ \ \ \ \ \ \ \ \ \ \ \ \ +\cdots,\label{formal-p}\\
u_2=\ \ \ \ \ \varepsilon^{3/2}u_2^{(1)} +\varepsilon^2u_2^{(2)} +\varepsilon^{5/2}u_2^{(3)}+\varepsilon^{3}u_2^{(4)}+\varepsilon^{7/2}u_2^{(5)} +\varepsilon^{4}u_2^{(6)}+\cdots,\label{formal-2}\\
u_3=\ \ \ \ \ \varepsilon^{3/2}u_3^{(1)} +\varepsilon^2u_3^{(2)} +\varepsilon^{5/2}u_3^{(3)}+\varepsilon^{3}u_3^{(4)}+\varepsilon^{7/2}u_3^{(5)} +\varepsilon^{4}u_3^{(6)}+\cdots.\label{formal-3}
\end{numcases}
\end{subequations}
Plugging the formal expansion \eqref{formal} into the system \eqref{equ2}, we get a power series of $\varepsilon$, whose coefficients depend on $(n^{(k)},\textbf{u}^{(k)},\phi^{(k)})$ for $k\geq1$.

\subsubsection{Derivation of the ZKE for $n^{(1)}$}
At the order of $\varepsilon$, we obtain\\
\emph{Coefficients of $\varepsilon^1$:}
\begin{subequations}\label{order1}
\begin{numcases}{}
-V\partial_{x_1}n^{(1)}+\partial_{x_1}u_1^{(1)}=0,\label{order1-n}\\
V\partial_{x_1}u_1^{(1)}-T_i\partial_{x_1}n^{(1)}=\partial_{x_1}\phi^{(1)},\label{order1-1}\\
\phi^{(1)}=n^{(1)},\label{order1-p}\\
T_i\partial_{x_2}n^{(1)}=-\partial_{x_2}\phi^{(1)}+u_3^{(1)},\label{order1-2}\\
T_i\partial_{x_3}n^{(1)}=-\partial_{x_3}\phi^{(1)}-u_2^{(1)}.\label{order1-3}
\end{numcases}
\end{subequations}
Consider \eqref{order1-n}-\eqref{formal-p}. To get a nontrivial solution, it is necessary to require the determinant of the coefficient matrix of \eqref{order1-n}-\eqref{formal-p} to vanish to obtain
\begin{equation}\label{v=1}
\begin{split}
V^2=T_i+1.
\end{split}
\end{equation}
For definiteness, we let $V=\sqrt{T_i+1}$.

At the orders of $\varepsilon^{3/2}$ and $\varepsilon^2$, we obtain\\
\emph{Coefficients of $\varepsilon^{3/2}$:}
\begin{subequations}\label{order32}
\begin{numcases}{}
\partial_{x_2}u_2^{(1)}+\partial_{x_3}u_3^{(1)}=0,\label{order32-n}\\
-V\partial_{x_1}u_2^{(1)}=u_3^{(2)},\label{order32-2}\\
-V\partial_{x_1}u_3^{(1)}=-u_2^{(2)}.\label{order32-3}
\end{numcases}
\end{subequations}
\emph{Coefficients of $\varepsilon^2$:}
\begin{subequations}\label{order2}
\begin{numcases}{}
\partial_tn^{(1)}-V\partial_{x_1}n^{(2)}+\partial_{x_1}(n^{(1)}u_1^{(1)}) +\partial_{x_1}u_1^{(2)}+\partial_{x_2}u_2^{(2)}+\partial_{x_3}u_3^{(2)}=0,\label{order2-n}\\
\partial_tu_1^{(1)}-V\partial_{x_1}u_1^{(2)}+u_1^{(1)}\partial_{x_1}u_1^{(1)} +T_i\{\partial_{x_1}n^{(2)}-n^{(1)}\partial_{x_1}n^{(1)}\}=-\partial_{x_1}\phi^{(2)},\label{order2-1}\\
\Delta\phi^{(1)}=\phi^{(2)}+\frac12(\phi^{(1)})^2-n^{(2)},\label{order2-p}\\
-V\partial_{x_1}u_2^{(2)} +T_i\{\partial_{x_2}n^{(2)}-n^{(1)}\partial_{x_2}n^{(1)}\} =-\partial_{x_2}\phi^{(2)}+u_3^{(3)},\label{order2-2}\\
-V\partial_{x_1}u_3^{(2)} +T_i\{\partial_{x_3}n^{(2)}-n^{(1)}\partial_{x_3}n^{(1)}\} =-\partial_{x_3}\phi^{(2)}-u_2^{(3)}.\label{order2-3}
\end{numcases}
\end{subequations}
From \eqref{order1-n}-\eqref{order1-p} and \eqref{v=1}, we can assume without loss of generality that
\begin{equation}\label{equ3}
\begin{split}
u_1^{(1)}=Vn^{(1)},\ \ \ \phi^{(1)}=n^{(1)}.
\end{split}
\end{equation}
From \eqref{order1-2} and \eqref{order1-3}, we have
\begin{subequations}\label{equ4}
\begin{numcases}{}
u_2^{(1)}=-T_i\partial_{x_3}n^{(1)}-\partial_{x_3}\phi^{(1)}=-V^2\partial_{x_3}n^{(1)},\label{equ4-1}\\
u_3^{(1)}=T_i\partial_{x_2}n^{(1)}+\partial_{x_2}\phi^{(1)}=V^2\partial_{x_2}n^{(1)},\label{equ4-2}
\end{numcases}
\end{subequations}
thanks to \eqref{v=1} and \eqref{equ3}. Therefore, to solve $n^{(1)}, \textbf{u}^{(1)}$ and $\phi^{(1)}$, we need only to solve $n^{(1)}$.

To find out the equation satisfied by $n^{(1)}$, we take $\partial_{x_1}$ of \eqref{order2-p}, multiply \eqref{order2-n} with $V$, and then add them to \eqref{order2-1}. We thus obtain
\begin{equation}\label{equ5}
\begin{split}
\partial_tn^{(1)}+Vn^{(1)}\partial_{x_1}n^{(1)} +\frac{1}{2V}\partial_{x_1}\Delta n^{(1)} +\frac{1}{2}\{\partial_{x_2}u_2^{(2)}+\partial_{x_3}u_3^{(2)}\}=0.
\end{split}
\end{equation}
On the other hand, from \eqref{order32-3}, \eqref{order32-2} and \eqref{equ4}, we have
\begin{equation}\label{equ6}
\begin{split}
\partial_{x_2}u_2^{(2)}&=V\partial_{x_2x_1}u_3^{(1)} =V^3\partial_{x_1}\partial^2_{x_2}n^{(1)},\\
\partial_{x_3}u_3^{(2)}&=-V\partial_{x_3x_1}u_2^{(1)} =V^3\partial_{x_1}\partial^2_{x_3}n^{(1)},
\end{split}
\end{equation}
thanks to \eqref{v=1}. Inserting this into \eqref{equ5}, we obtain the Zakarov-Kuznetsov equation
\begin{equation}\label{ZKE}
\begin{split}
\partial_tn^{(1)}+n^{(1)}\partial_{x_1}n^{(1)} +\frac{1}{2V}\partial_{x_1}^3n^{(1)} +\frac{V^3}{2}\partial_{x_1}\Delta_{\perp}n^{(1)}=0,
\end{split}
\end{equation}
where $\Delta_{\perp}=\partial_{x_2}^2+\partial_{x_3}^2$ in 3D.

\begin{proposition}\label{prop-ZKE}
Let $s\geq9/8$, the Cauchy problem of ZKE \eqref{ZKE} is locally well-posed in $H^s(\Bbb R^3)$.
\end{proposition}
\begin{proof}
See \cite{LS09}.
\end{proof}

\begin{remark}\label{rmk2}
\eqref{equ3}, \eqref{equ4} and \eqref{ZKE} are a closed system. Once $n^{(1)}$ is solved from \eqref{ZKE}, we have all the other first order profiles $(\textbf{u}^{(1)},\phi^{(1)})$ from \eqref{equ3} and \eqref{equ4}. Furthermore, we can also solve$(u_2^{(2)},u_3^{(2)})$ from \eqref{order32-3} and \eqref{order32-2}. In other words, $(n^{(1)},\textbf{u}^{(1)},\phi^{(1)})$ and $(u_2^{(2)},u_3^{(2)})$ can be solved independently, although the equations \eqref{order1}, \eqref{order32} and \eqref{order2} for the coefficients of $\varepsilon$, $\varepsilon^{3/2}$ and $\varepsilon^2$ depend on the higher order profiles $(n^{(2)},u_1^{(2)},\phi^{(2)})$ and $(u_2^{(3)},u_3^{(3)},u_2^{(4)},u_3^{(4)})$.
\end{remark}

\subsubsection{Derivation of the Linearized ZKE for $n^{(2)}$}
Now, we derive the equation that satisfied by $n^{(2)}$. At the order of $\varepsilon^{5/2}$, we obtain\\
\emph{Coefficients of $\varepsilon^{5/2}$:}
\begin{subequations}\label{order52}
\begin{numcases}{}
\partial_{x_2}(u_2^{(3)}+n^{(1)}u_2^{(1)}) +\partial_{x_3}(u_3^{(3)}+n^{(1)}u_3^{(1)})=0, \label{order52-n}\\
u_2^{(1)}\partial_{x_2}u_1^{(1)}+u_3^{(1)}\partial_{x_3}u_1^{(1)},\label{order52-1}\\
\partial_t u_2^{(1)}-\partial_{x_1}u_2^{(3)}+u_1^{(1)}\partial_{x_1}u_2^{(1)} =u_3^{(4)},\label{order52-2}\\
\partial_t u_3^{(1)}-\partial_{x_1}u_3^{(3)}+u_1^{(1)}\partial_{x_1}u_3^{(1)} =-u_2^{(4)}.\label{order52-3}
\end{numcases}
\end{subequations}
We first note that \eqref{order52-1} is consistent with \eqref{order1-2} and \eqref{order1-3}. Indeed, from \eqref{order1-2} and \eqref{order1-3}, we can derive \eqref{order52-1} by noting \eqref{equ3}. Also, \eqref{order52-n} is consistent with \eqref{order1-2}, \eqref{order1-3},\eqref{order32-2},\eqref{order32-3},\eqref{order2-2},\eqref{order2-3}. Indeed, from \eqref{order2-2} and \eqref{order2-3}, we have
\begin{equation*}
\begin{split}
\partial_{x_2}u_2^{(3)}+\partial_{x_3}u_3^{(3)} =&V\partial_{x_1x_2}u_3^{(2)}-V\partial_{x_1x_3}u_2^{(2)}\\
=&-V\partial_{x_1}^2\partial_{x_2}u_2^{(1)}-V\partial_{x_1}^2\partial_{x_3}u_3^{(1)}\\
=&0,
\end{split}
\end{equation*}
where we have used \eqref{order32-2} and \eqref{order32-3} in the second equality and \eqref{order1-2} and \eqref{order1-3} in the third equality. Similarly, by using \eqref{order1-2} and \eqref{order1-3}, we obtain
\begin{equation*}
\begin{split}
\partial_{x_2}(n^{(1)}u_2^{(1)})+\partial_{x_3}(n^{(1)}u_3^{(1)}) =&\{n^{(1)}(\partial_{x_2}u_2^{(1)}+\partial_{x_3}u_3^{(1)})\} +\{u_2^{(1)}\partial_{x_2}n^{(1)}+u_3^{(1)}\partial_{x_3}n^{(1)}\}\\
=&0.
\end{split}
\end{equation*}
From \eqref{order52-3} and \eqref{order2-2}, we have
\begin{equation*}
\begin{split}
u_2^{(4)}=& \partial_{x_1}u_3^{(3)} -\{\partial_tu_3^{(1)}+u_1^{(1)}\partial_{x_1}u_3^{(1)}\}\\
=& \partial_{x_1x_2}\phi^{(2)}+T_i\partial_{x_1x_2}n^{(2)}+\underline{a_2}^{(1)},
\end{split}
\end{equation*}
where $\underline{a_2}^{(1)}=-\{V\partial_{x_1}^2u_2^{(2)} +\partial_tu_3^{(1)}+T_i\partial_{x_1}(n^{(1)}\partial_{x_2}n^{(1)}) +u_1^{(1)}\partial_{x_1}u_3^{(1)}\}$. Therefore,
\begin{equation}\label{equ7}
\begin{split}
\partial_{x_2}u_2^{(4)}=& \partial_{x_1}\partial^2_{x_2}\phi^{(2)} +T_i\partial_{x_1}\partial^2_{x_2}n^{(2)} +\partial_{x_2}\underline{a_2}^{(1)}.
\end{split}
\end{equation}
Similarly, from \eqref{order52-2} and \eqref{order2-3}, we obtain
\begin{equation}\label{equ8}
\begin{split}
\partial_{x_3}u_3^{(4)}=& \partial_{x_1}\partial^2_{x_3}\phi^{(2)} +T_i\partial_{x_1}\partial_{x_3}^2n^{(2)} +\partial_{x_3}\underline{a_3}^{(1)},
\end{split}
\end{equation}
where $\underline{a_3}^{(1)}=\{\partial_tu_2^{(1)} +u_1^{(1)}\partial_{x_1}u_2^{(1)}-V\partial_{x_1}^2u_3^{(2)} -T_i\partial_{x_1}(n^{(1)}\partial_{x_3}n^{(1)})\}$.

At the order of $\varepsilon^3$, we obtain\\
\emph{Coefficients of $\varepsilon^3$:}
\begin{subequations}\label{order3}
\begin{numcases}{}
\partial_tn^{(2)}-\partial_{x_1}n^{(3)}+\partial_{x_1}(n^{(1)}u_1^{(2)} +n^{(2)}u_1^{(1)})+\partial_{x_1}u_1^{(3)} +\partial_{x_2}u_2^{(4)}+\partial_{x_3}u_3^{(4)}\nonumber\\
\ \ \ \ \ \ \ \ \ \ \ +\{\partial_{x_2}(n^{(1)}u_2^{(2)})+\partial_{x_3}(n^{(1)}u_3^{(2)})\}=0, \label{order3-n}\\
\partial_tu_1^{(2)}-\partial_{x_1}u_1^{(3)} +u_1^{(1)}\partial_{x_1}u_1^{(2)}+u_1^{(2)}\partial_{x_1}u_1^{(1)} +T_i\{\partial_{x_1}n^{(3)}-n^{(1)}\partial_{x_1}n^{(2)}\nonumber\\
\ \ \ \ \ \ \ \ \ \ \ -(n^{(2)}-(n^{(1)})^2) \partial_{x_1}n^{(1)}\} =-\partial_{x_1}\phi^{(3)} +\{u_2^{(2)}\partial_{x_2}u_1^{(1)}+u_3^{(2)}\partial_{x_3}u_1^{(1)}\},\ \ \ \ \ \label{order3-1}\\
\Delta\phi^{(2)}=\phi^{(3)}+\phi^{(1)}\phi^{(2)} +\frac{1}{3!}(\phi^{(1)})^3-n^{(3)},\label{order3-p}\\
\partial_tu_2^{(2)}-\partial_{x_1}u_2^{(4)} +u_1^{(1)}\partial_{x_1}u_2^{(2)}+T_i\{\partial_{x_2}n^{(3)}-n^{(1)}\partial_{x_2}n^{(2)}\nonumber\\
\ \ \ \ \ \ \ \ \ \ \ -(n^{(2)}-(n^{(1)})^2) \partial_{x_2}n^{(1)}\} =-\partial_{x_2}\phi^{(3)}+u_3^{(5)},\label{order3-2}\\
\partial_tu_3^{(2)}-\partial_{x_1}u_3^{(4)}+u_1^{(1)}\partial_{x_1}u_3^{(2)} +T_i\{\partial_{x_3}n^{(3)}-n^{(1)}\partial_{x_3}n^{(2)}\nonumber\\
\ \ \ \ \ \ \ \ \ \ \ -(n^{(2)}-(n^{(1)})^2) \partial_{x_3}n^{(1)}\} =-\partial_{x_3}\phi^{(3)}-u_2^{(5)}.\label{order3-3}
\end{numcases}
\end{subequations}
We first note that from \eqref{order2-p}, we can assume without loss of generality that
\begin{equation}\label{equ9}
\begin{split}
\phi^{(2)}=n^{(2)}+\underline{\phi}^{(1)},
\end{split}
\end{equation}
where $\underline{\phi}^{(1)}=\Delta\phi^{(1)}-\frac12(\phi^{(1)})^2$ is known from \eqref{ZKE}, since $\phi^{(1)}=n^{(1)}$ from \eqref{equ3}. From \eqref{order2-n}, we have
\begin{equation*}
\begin{split}
\partial_{x_1}u_1^{(2)}=V\partial_{x_1}n^{(2)}+\underline{\mathfrak{n}}^{(1)},
\end{split}
\end{equation*}
where $\underline{\mathfrak{n}}^{(1)}=-\partial_tn^{(1)}-\partial_{x_1}(n^{(1)}u_1^{(1)}) -\partial_{x_2}u_2^{(2)}-\partial_{x_3}u_3^{(2)}$. Without loss of generality, we can assume that
\begin{equation}\label{equ10}
\begin{split}
u_1^{(2)}=n^{(2)}+\underline{n}^{(1)},
\end{split}
\end{equation}
where $\underline{n}^{(1)}=\int_{-\infty}^{x_1}\underline{\mathfrak{n}}^{(1)}dx_1$.
\begin{remark}
We claim that $\int_{-\infty}^{\infty}\underline{\mathfrak{n}}^{(1)}dx_1=0$. First, from \eqref{ZKE}, we have $\partial_t\int_{-\infty}^{\infty}n^{(1)}=0$, thanks to the divergence theorem. On the other hand, by \eqref{equ6}, we have $\underline{\mathfrak{n}}^{(1)} =-\partial_tn^{(1)}-\partial_{x_1}\{(n^{(1)}u_1^{(1)}) +\Delta_{\perp}n^{(1)}\}$. The claim then follows, again thanks to divergence theorem.
\end{remark}
By taking $\partial_{x_1}$ of \eqref{order3-p}, multiplying \eqref{order3-n} with $V$ and then adding them to \eqref{order3-1},
we obtain a linearized inhomogeneous ZKE for $n^{(2)}$:
\begin{equation}\label{lin-2}
\begin{split}
\partial_tn^{(2)}+\partial_{x_1}(n^{(1)}n^{(2)})+\frac12\partial_{x_1}^3n^{(2)} +\partial_{x_1}\Delta_{\perp}n^{(2)}=\underline{G}^{(1)},
\end{split}
\end{equation}
where $\underline{G}^{(1)}$ is the inhomogeneous term, depending only on $n^{(1)}$. Here, we have used \eqref{equ7}, \eqref{equ8}, \eqref{equ9} and \eqref{equ10}. Furthermore, we also get the coefficients of $\varepsilon^{7/2}$, which depend only on $n^{(2)}$  directly or indirectly:\\
\emph{Coefficients of $\varepsilon^{7/2}$:}
\begin{subequations}\label{order72}
\begin{numcases}{}
\partial_{x_2}(u_2^{(5)}+n^{(1)}u_2^{(3)}+n^{(2)}u_2^{(1)}) +\partial_{x_3}(u_3^{(5)}+n^{(1)}u_3^{(3)}+n^{(2)}u_3^{(1)})=0, \label{order72-n}\\
u_2^{(3)}\partial_{x_2}u_1^{(1)}+u_2^{(1)}\partial_{x_2}u_1^{(2)} +u_3^{(3)}\partial_{x_3}u_1^{(1)}+u_3^{(1)}\partial_{x_3}u_1^{(2)},\label{order72-1}\\
\partial_t u_2^{(3)}-\partial_{x_1}u_2^{(5)}+\textbf{u}^{(2)}\nabla u_2^{(1)} +u_1^{(1)}\partial_{x_1}u_2^{(3)} +u_2^{(1)}\partial_{x_2}u_2^{(2)} +u_3^{(1)}\partial_{x_3}u_2^{(2)} =u_3^{(6)},\label{order72-2}\\
\partial_t u_3^{(3)}-\partial_{x_1}u_3^{(5)}+\textbf{u}^{(2)}\nabla u_3^{(1)} +u_1^{(1)}\partial_{x_1}u_3^{(3)} +u_2^{(1)}\partial_{x_2}u_3^{(2)} +u_3^{(1)}\partial_{x_3}u_3^{(2)} =-u_2^{(6)}.\ \ \ \ \ \ \ \label{order72-3}
\end{numcases}
\end{subequations}

\begin{remark}[Continuation of Remark \ref{rmk2}]\label{rmk3}
Once $n^{(2)}$ is solved from \eqref{lin-2}, then $u_1^{(2)}$,  $\phi^{(2)}$, $u_2^{(3)}$, $u_3^{(3)}$, $u_2^{(4)}$ and $u_3^{(4)}$ are all known. Although the expression for the coefficients of $\varepsilon^3$ and $\varepsilon^{7/2}$ depend on the higher approximations $(n^{(3)},u_1^{(3)},\phi^{(3)})$ and $(u_2^{(5)},u_3^{(5)},u_2^{(6)},u_3^{(6)})$, they can be solved independently. Furthermore, $n^{(i)}$, $u_1^{(i)}$, $\phi^{(i)}$ for $i\leq 2$ and $u_2^{(j)}$ and $u_3^{(j)}$ for $j\leq 4$ will make the systems \eqref{order32}, \eqref{order2} and \eqref{order52} of the coefficients of $\varepsilon$, $\varepsilon^{3/2}$, $\varepsilon^2$, $\varepsilon^{5/2}$ valid exactly.
\end{remark}

\subsubsection{The linearized ZKE for $n^{(k)}$}
Inductively, we can derive all the profiles $n^{(k)},\textbf{u}^{(k)}$ and $\phi^{(k)}$. $n^{(k)}$ for $k\geq 3$ satisfy a linearized ZKE similar to \eqref{lin-2}
\begin{equation}\label{lin-k}
\begin{split}
\partial_tn^{(k)}+\partial_{x_1}(n^{(1)}n^{(k)})+\frac12\partial_{x_1}^3n^{(2)} +\partial_{x_1}\Delta_{\perp}n^{(k)}=\underline{G}^{(k-1)},
\end{split}
\end{equation}
where $\underline{G}^{(k-1)}$ depends only on $n^{(i)}$ for $j\leq k-1$.

\begin{proposition}\label{prop0}
Let $s\geq 9/8$, the Cauchy problem of the linearized inhomogeneous ZKE \eqref{lin-k} for $k\geq 9/8$ is well-posed in $H^s(\Bbb R^3)$.
\end{proposition}

\begin{remark}[Continuation of Remark \ref{rmk3}]\label{rmk4}
In particular, we consider the case of $k=3$. Let $n^{(i)}$ $(i=1,2,3)$ be solved from \eqref{ZKE} and \eqref{lin-k} for $k=2,3$. Then $u_1^{(i)}$ and $\phi^{(i)}$ for $i=1,2,3$ and $u_2^{(j)}$ and $u_3^{(j)}$ for $j=1,\cdots,6$ are all known. They will make the systems of the coefficients up to order of $\varepsilon^{7/2}$ valid exactly.
\end{remark}

\subsection{Remainder equation for ZKE}
To make the previous formal derivation rigorous, we consider the following expansion with remainder term $(n^{\varepsilon}_R, \textbf{u}^{\varepsilon}_R, \phi^{\varepsilon}_R)$,
\begin{subequations}\label{expan}
\begin{numcases}{}
\ n=1+\varepsilon n^{(1)}\ \ \ \ \ \ \ \ \ \ \ \ \ \ \ \ \ \ +\varepsilon^2n^{(2)}\ \ \ \ \ \ \ \ \ \ \ \ \ \ \ \ \ +\varepsilon^3n^{(3)}\ \ \ \ \ \ \ \ \ \ \ \ \ \ \ \ +\varepsilon^2n^{\varepsilon}_R,\label{expan-n}\\
u_1=\ \ \ \ \ \varepsilon u_1^{(1)}\ \ \ \ \ \ \ \ \ \ \ \ \ \ \ \ \ \ \ +\varepsilon^2u_1^{(2)}\ \ \ \ \ \ \ \ \ \ \ \ \ \ \ \ \ +\varepsilon^3u_1^{(3)}\ \ \ \ \ \ \ \ \ \ \ \ \ \ \ \ +\varepsilon^2{u_1}^{\varepsilon}_R,\label{expan-1}\\
\ \phi=\ \ \ \ \ \varepsilon\phi^{(1)}\ \ \ \ \ \ \ \ \ \ \ \ \ \ \ \ \ \ \ +\varepsilon^2\phi^{(2)}\ \ \ \ \ \ \ \ \ \ \ \ \ \ \ \ \ +\varepsilon^3\phi^{(3)}\ \ \ \ \ \ \ \ \ \ \ \ \ \ \ \ +\varepsilon^2\phi^{\varepsilon}_R,\label{expan-p}\\
u_2=\ \ \ \  \varepsilon^{3/2}{u_2}^{(1)} +\varepsilon^2u_2^{(2)} +\varepsilon^{5/2}u_2^{(3)}+\varepsilon^{3}u_2^{(4)}+\varepsilon^{7/2}u_2^{(5)} +\varepsilon^{4}u_2^{(6)}+\varepsilon^2{u_2}^{\varepsilon}_R,\label{expan-2}\\
u_3=\ \ \ \ \ \varepsilon^{3/2}u_3^{(1)}+\varepsilon^2u_3^{(2)} +\varepsilon^{5/2}u_3^{(3)}+\varepsilon^{3}u_3^{(4)}+\varepsilon^{7/2}u_3^{(5)} +\varepsilon^{4}u_3^{(6)}+\varepsilon^2{u_3}^{\varepsilon}_R,\ \ \ \ \label{expan-3}
\end{numcases}
\end{subequations}
where $\textbf{u}^{\varepsilon}_R=({u_1}^{\varepsilon}_R, {u_2}^{\varepsilon}_R, {u_3}^{\varepsilon}_R)$. Here $n^{(1)}$, $n^{(2)}$ and $n^{(3)}$ satisfy \eqref{ZKE}, \eqref{lin-2} and \eqref{lin-k} for $k=3$. The other profiles $u_1^{(i)}$ and $\phi^{(i)}$ for $i=1,2,3$ and $u_2^{(j)}$ and $u_3^{(j)}$ for $j=1,\cdots,6$ are solved from the systems \eqref{order1}, \eqref{order32}, \eqref{order2}, \eqref{order52}, \eqref{order3} and \eqref{order72} of coefficients up to order $\varepsilon^{7/2}$. See Remark \ref{rmk2}, \ref{rmk3} and \ref{rmk4}.

For notational convenience, we denote $\textbf{u}=({u_1},{u_2},u_3)^T$, $\textbf{u}^{\varepsilon}_R=({u_1}^{\varepsilon}_R, {u_2}^{\varepsilon}_R, {u_3}^{\varepsilon}_R)^T$ and
\begin{equation}\label{equ95}
\begin{split}
\widetilde n=&n^{(1)}+\varepsilon n^{(2)}+\varepsilon^2n^{(3)},\ \ \ \ \ \ \  \widetilde{\phi}=\phi^{(1)}+\varepsilon\phi^{(2)} +\varepsilon^2\phi^{(3)},\\
\widetilde{\textbf{u}}=&(\widetilde u_1,\widetilde u_2, \widetilde u_3)^T,\ \ \ \ \ \ \ \ \ \ \ \ \ \ \ \ \ \ \widetilde u_1=u_1^{(1)}+\varepsilon u_1^{(2)}+\varepsilon^2u_1^{(3)},\\
\widetilde u_2=&\varepsilon^{1/2}u_2^{(1)}+\varepsilon u_2^{(2)}+\varepsilon^{3/2}u_2^{(3)}+\varepsilon^2u_2^{(4)} +\varepsilon^{5/2}u_2^{(5)}+\varepsilon^3u_2^{(6)},\\
\widetilde u_3=&\varepsilon^{1/2}u_3^{(1)}+\varepsilon u_3^{(2)}+\varepsilon^{3/2}u_3^{(3)}+\varepsilon^2u_3^{(4)} +\varepsilon^{5/2}u_3^{(5)}+\varepsilon^3u_3^{(6)}.
\end{split}
\end{equation}

\begin{proposition}\label{prop-rem-zk}
Let $(n,\textbf{u},\phi)$ in \eqref{expan} be a solution of the Euler-Poisson system \eqref{equ1}, then $(n^{\varepsilon}_R, \textbf{u}^{\varepsilon}_R, \phi^{\varepsilon}_R)$ satisfy the following remainder system
\begin{subequations}\label{rem}
\begin{numcases}{}
\partial_tn^{\varepsilon}_R-\frac{V\textbf{e}_1-\textbf{u}}{\varepsilon}\cdot\nabla n^{\varepsilon}_R+\frac{n}{\varepsilon}\nabla\cdot\textbf{u}^{\varepsilon}_R +n^{\varepsilon}_R\nabla\cdot\widetilde{\textbf{u}} +\textbf{u}^{\varepsilon}_R\cdot\nabla\widetilde{n}+\varepsilon R_n=0,\label{rem-1}\\
\partial_t\textbf{u}^{\varepsilon}_R-\frac{V\textbf{e}_1-\textbf{u}}{\varepsilon} \cdot\nabla\textbf{u}^{\varepsilon}_R +\textbf{u}^{\varepsilon}_R\nabla\cdot\widetilde{\textbf{u}} +\frac{T_i}{\varepsilon n}\overline{\nabla}n^{\varepsilon}_R\nonumber\\
\ \ \ \ \ \ \ \ \ \ \ \ \ -\frac{T_i \textbf{p}}{\varepsilon n}n^{\varepsilon}_R-\frac{T_i \varepsilon}{n}\textbf{R}_{T} +\varepsilon \textbf{R}_\textbf{u}=-\frac{1}{\varepsilon}\nabla\phi^{\varepsilon}_R +\frac{1}{\varepsilon^{3/2}}\textbf{u}^{\varepsilon}_R\times \textbf{e}_1,\label{rem-2}\\
\varepsilon\Delta\phi^{\varepsilon}_R=\phi^{\varepsilon}_R-n^{\varepsilon}_R +\varepsilon\phi^{(1)}\phi^{\varepsilon}_R +\varepsilon^{3/2}R_{\phi},\label{rem-3}
\end{numcases}
\end{subequations}
where $\textbf{u}^{\varepsilon}_R=({u_1}^{\varepsilon}_R, {u_2}^{\varepsilon}_R,{u_3}^{\varepsilon}_R)$ and $\widetilde n$, $\widetilde{\textbf{u}}$ and $\widetilde\phi$ are given in \eqref{equ95}. Here, $R_n,\textbf{R}_\textbf{u}=(R_{\textbf{u}1}, R_{\textbf{u}2}, R_{\textbf{u}3})$ depend only on $n^{(k)}$, $\textbf{u}^{(k)}$ and $\phi^{(k)}$, and $R_{\phi}$ depends on $\phi^{\varepsilon}_R$ in the form $R_{\phi}=F(\sqrt{\varepsilon}\phi^{\varepsilon}_R)\phi^{\varepsilon}_R +\sqrt{\varepsilon}R'_{\phi}$ for some $R'_{\phi}$ depending only on $n^{(k)}$, $\textbf{u}^{(k)}$ and $\phi^{(k)}$. In \eqref{rem-2}, $\textbf{p}=(p_1,p_2,p_3)$ and $\textbf{R}_T=(R_{T1},R_{T2},R_{T3})$ are finite combinations of $n^{(1)}, n^{(2)}$ and $n^{(3)}$. In \eqref{rem-1}, $\textbf{e}_1=(1,0,0)'$ is a constant vector.
\end{proposition}

The derivation of such a system for $(n^{\varepsilon}_R, \textbf{u}^{\varepsilon}_R, \phi^{\varepsilon}_R)$ is similar to that of \eqref{rem-kp} in the KPE limit case. \eqref{rem-3} can also be written in the equivalent form of \eqref{e35} in Proposition \ref{prop-kp5}. The remainder term $R_{\phi}$ and $\overline{R}_{\phi}$ satisfy the same estimates of Lemma \ref{lem-A}. These claims can be proved exactly as those in Appendix A and hence omitted.

\section{Commutator estimates}
\setcounter{equation}{0}
We give two important inequalities which are widely used throughout this paper \cite{KP88}.
\begin{lemma}\label{Le-inequ}
Let $\alpha$ be any multi-index with $|\alpha|=k$ and $p\in (1,\infty)$. Then there exists some constant $C>0$ such that
\begin{equation}\label{mul-com}
\begin{split}
\|\partial_x^{\alpha}(fg)\|_{L^p}\leq & C\{\|f\|_{L^{p_1}}\|g\|_{\dot H^{s,p_2}}+\|f\|_{\dot H^{s,p_3}}\|g\|_{L^{p_4}}\},\\
\|[\partial_x^{\alpha},f]g\|_{L^p}\leq & C\{\|\nabla f\|_{L^{p_1}}\|g\|_{\dot H^{k-1,p_2}}+\|f\|_{\dot H^{k,p_3}}\|g\|_{L^{p_4}}\},
\end{split}
\end{equation}
where $f,g\in \mathscr{S}$, the Schwartz class and $p_2,p_3\in (1,+\infty)$ such that
\begin{equation*}
\begin{split}
\frac1p=\frac1{p_1}+\frac1{p_2}=\frac1{p_3}+\frac1{p_4}.
\end{split}
\end{equation*}
\end{lemma}

\paragraph{\emph{\textbf{Acknowledgment}}} This research is supported by NSFC under grant 11001285. The author thanks B. Pausader for pointing out the problem, helpful discussions and encouragement.

\end{CJK*}
\end{document}